\documentclass{article}
\usepackage[a4paper, total={6in, 8in}]{geometry}
\usepackage{etex}
\usepackage{dsfont}
\usepackage{times}
\usepackage{helvet}
\usepackage{courier}
\usepackage[utf8]{inputenc} %if you compile by pdflatex
\usepackage{fixltx2e} %if you compile by pdflatex
\usepackage{mathtools}
\usepackage{complexity}
\usepackage{balance}
\usepackage{mathdots}
\usepackage{bbold}
\RequirePackage{latexsym}
\RequirePackage{mathtools}
\RequirePackage{amsfonts}
\usepackage{amssymb}
\RequirePackage{amsthm}
\RequirePackage{array}
\RequirePackage{xspace}
\usepackage{url}
\usepackage{graphicx}
\usepackage{cancel}
\usepackage{newfloat}
\DeclareFloatingEnvironment[
    fileext=los,
    listname=List of Algorithms,
    name=Algorithm,
    placement=tbhp,
    within=none,
]{algorithmenv}

\newcommand{\ie}{i.e., }
\newcommand{\eg}{e.g., }

\newcommand{\wrt}{w.r.t.\ }

\let\oldnl\nl
\newcommand{\nonl}{\renewcommand{\nl}{\let\nl\oldnl}}

\newcommand{\RR}{\mathbf{R}\xspace}
\newcommand{\Q}{\mathbf{Q}\xspace}
\newcommand{\N}{\mathbf{N}\xspace}
\newcommand{\Z}{\mathbf{Z}\xspace}

\newcommand{\STN}{STN\xspace}
\newcommand{\MPG}{MPG\xspace}

\newcommand{\HTN}{HyTN\xspace}
\newcommand{\DTP}{DTP\xspace}
\newcommand{\HTNC}{HyTN-Consistency\xspace}
\newcommand{\CSTP}{CSTP\xspace}
\newcommand{\CSTN}{CSTN\xspace}
\newcommand{\eHyDCC}{CHyTN-$\epsilon$-DC\xspace}

\newcommand{\GHyDCC}{General-CHyTN-DC\xspace}
\newcommand{\HyDCC}{CHyTN-DC\xspace}
\newcommand{\DCC}{CSTN-DC\xspace}

\newcommand{\scHst}{\emph{scHst}\xspace}
\newcommand{\Con}{\emph{Con}\xspace}
\newcommand{\Sub}{\emph{Sub}\xspace}

\def\Ord{{\cal O}}
\def\C{{\cal C}}

\def\H{{\cal H}}

\def\A{{\cal A}}

\def\S{{\cal S}}

\newcommand{\figref}[1]{Fig.~\ref{#1}}

\newtheorem{Thm}{Theorem}%[section]
\newtheorem{Cor}[]{Corollary}

\newtheorem{Lem}[]{Lemma}

\theoremstyle{plain}
\newtheorem{Def}{Definition}
\newtheorem{Rem}{Remark}
\newtheorem{Ex}{Example}

\usepackage[ruled,vlined,linesnumbered]{algorithm2e}
\usepackage{subfig}
\SetCommentSty{textsf}
\SetKwRepeat{DoWhile}{do}{while}
\SetAlFnt{\small} 
\makeatletter
\newcommand{\removelatexerror}{\let\@latex@error\@gobble}
\makeatother

\usepackage{tikz}%Front end per PGF, package di disegno 
\usepackage{tikz-qtree}
\usetikzlibrary{calc,positioning,fit}%calc e positioning permettono il posizionamento avanzato
\usetikzlibrary{shapes,shapes.multipart,shapes.arrows}
\usetikzlibrary{decorations,decorations.pathmorphing,decorations.pathreplacing,decorations.markings,decorations.shapes}
\usetikzlibrary{arrows}%se si usano le frecce diverse
\usetikzlibrary{fit,backgrounds}
\usetikzlibrary{plotmarks}
\tikzstyle{node}=[circle,draw,inner sep=2pt,transform shape,minimum size=1.75em]
\tikzstyle{task}=[draw,rectangle,inner sep=1.5pt,transform shape
                         ,font=\sffamily\scriptsize,text width=27pt%vedi nota sopra
                         ,text centered,minimum height=2em
                         ]
\tikzstyle{connector}=[draw,diamond,shape aspect=1,inner sep=1pt,transform shape
                 ,font=\sffamily\scriptsize,text width=15pt%vedi nota sopra. 
                 ,text centered
                 ]

\tikzstyle{multiHead}=[dashed,transform shape]
\tikzstyle{multiTail}=[dotted,thick,transform shape]
                               
\newcommand{\AND}{\tikz[baseline,arrows=-] \draw[thick] (0,2ex) -- (4ex,2ex) (2ex,0) -- (2ex,4ex);}%$\boldsymbol{+}$}
\tikzstyle{StartCase}=[circle,draw,minimum size=.75cm,transform shape]
\tikzstyle{EndCase}=[circle,draw,ultra thick,minimum size=.75cm,transform shape]
\tikzstyle{smallLabel}=[font=\sffamily\scriptsize,inner sep=1pt,transform shape]
\tikzstyle{timeLabel}=[smallLabel,midway,transform shape]
\tikzstyle{minWidth}=[text width=.9cm]
\tikzstyle{info}=[rounded corners,fill=yellow,text width=1cm,text centered,inner sep=1pt]
\tikzstyle{infoLine}=[thin,decorate,decoration={snake,amplitude=.4mm, segment length=2mm, post length=1mm}]
\tikzstyle{edgeLabel}=[font=\tiny,sloped]
\tikzstyle{infoRow}=[rounded corners,fill=yellow,inner sep=1pt]
\tikzstyle{punto}=[circle,draw,fill=black,minimum size=2bp,inner sep=0pt,outer sep=0pt]
\tikzstyle{blackNode}=[fill=black!30]
\tikzstyle{crosses}=[decorate,decoration={name=crosses,segment length=2mm,post length=2mm}]
\tikzstyle{every picture}=[>=latex]
\tikzstyle{every label}=[inner sep=2pt]
\newcommand{\GTNC}[1][]{\textsc{General-\HTN-Consistency#1}\xspace}
\renewcommand{\H}{\ensuremath{\mathcal{H}}\xspace}
\newcommand{\TTNC}[1][]{\textsc{Tail-\HTN-Consistency#1}\xspace}

\begin{document}

\title{Checking Dynamic Consistency of Conditional Hyper Temporal Networks via 
Mean Payoff Games \\ \text{\large Hardness and (pseudo) Singly-Exponential Time Algorithm}}

\author{Carlo Comin\footnote{Department of Mathematics, University of Trento, Trento, Italy; 
LIGM, Universit{\'e} Paris-Est, Marne-la-Vall{\'e}e, Paris, France. E-mail: carlo.comin@unitn.it} 
\and Romeo Rizzi\footnote{Department of Computer Science, University of Verona, Verona, Italy. E-mail: romeo.rizzi@univr.it}}

\date{}

\maketitle

\begin{abstract}
Conditional Simple Temporal Network (\CSTN) is a constraint-based graph-formalism for conditional temporal planning.
It offers a more flexible formalism than the equivalent CSTP model of~\cite{TVP2003}, 
from which it was derived mainly as a sound formalization.
Three notions of consistency arise for \CSTN{s}: weak, strong, and dynamic.
Dynamic consistency is the most interesting notion, but it is also the most challenging and it was conjectured to be hard to assess.
\cite{TVP2003} gave a doubly-exponential time algorithm for checking dynamic consistency in \CSTN{s}
and to produce an exponentially sized dynamic execution strategy whenever the input \CSTN is dynamically-consistent.
\CSTN{s} may be viewed as an extension of Simple Temporal Networks (\STN{s})~\cite{DechterMP91},
directed weighted graphs where nodes represent events to be scheduled in time and 
arcs represent temporal distance constraints between pairs of events.
Recently, \STN{s} have been generalized into \emph{Hyper Temporal Networks} (\HTN{s}), by considering
weighted directed hypergraphs where each hyperarc models a \emph{disjunctive} temporal constraint named \emph{hyperconstraint};
being directed, the hyperarcs can be either \emph{multi-head} or \emph{multi-tail}.
The computational equivalence between checking consistency in \HTN{s} and 
determining winning regions in Mean Payoff Games (\MPG{s}) was also pointed out;
\MPG{s} are a family of 2-player infinite pebble games played on finite graphs,
which is well known for having applications in model-checking and formal verification.
In this work we introduce the \emph{Conditional Hyper Temporal Network (CHyTN)} model,
a natural extension and generalization of both the \CSTN and the \HTN model which is obtained by blending them together.
We show that deciding whether a given \CSTN or CHyTN is dynamically-consistent is \coNP-hard;
and that deciding whether a given CHyTN is dynamically-consistent is \PSPACE-hard,
provided that the input instances are allowed to include both multi-head and multi-tail hyperarcs.
In light of this, we continue our study by focusing on CHyTNs that allow only multi-head hyperarcs,
and we offer the first deterministic (pseudo) singly-exponential 
time algorithm for the problem of checking the dynamic consistency of such CHyTNs,
also producing a dynamic execution strategy whenever the input CHyTN is dynamically-consistent.
Since \CSTN{s} are a special case of CHyTNs, as a byproduct this provides the first sound-and-complete
(pseudo) singly-exponential time algorithm for checking dynamic consistency in CSTNs.
The proposed algorithm is based on a novel connection between CHyTN{s} and \MPG{s};
due to the existence of efficient pseudo-polynomial time algorithms for \MPG{s}, 
it is quite promising to be competitive in practice.
The presentation of such connection is mediated by the \HTN model.
In order to analyze the time complexity of the algorithm, we introduce a refined notion of dynamic consistency, 
named $\epsilon$-dynamic consistency, and present a sharp lower bounding analysis 
on the critical value of the reaction time $\hat{\varepsilon}$ where a CHyTN transits from being,
to not being, dynamically-consistent. The proof technique introduced in this analysis of $\hat{\varepsilon}$ is applicable
more generally when dealing with linear difference constraints which include strict inequalities.
\end{abstract}

\text{\bf{Keywords:}}
Dynamic Consistency, Mean Payoff Games, \STN{s}, Hyper Temporal Networks, Singly-Exponential Time, Reaction Time.

\section{Introduction and Motivation}\label{sect:introduction}
In many areas of Artificial Intelligence (AI), including temporal planning and scheduling,
  the representation and management of quantitative temporal aspects is
  of crucial importance~(see \eg \cite{Pani:2001tb,SmithFJ00,EderPR99,BettiniWJ02,CombiGPP12,CombiGMP14}).
Examples of possible quantitative temporal aspects include constraints on the earliest start time and latest end time of activities and
  constraints over the minimum and maximum temporal distance between activities.
In many cases these constraints can be represented by \emph{Simple Temporal Networks} (\STN{s})~\cite{DechterMP91},
\ie directed weighted graphs where nodes represent events to be scheduled in time
and arcs represent temporal distance constraints between pairs of events.
Recently, \STN{s} have been generalized into \emph{Hyper Temporal Networks} (\HTN{s}) \cite{CPR2014, CPR2015},
a strict generalization of \STN{s} introduced to overcome the limitation of considering only conjunctions
of constraints, but maintaining a practical efficiency in the consistency checking of the instances.
In a \HTN a single temporal hyperarc constraint is defined as a set of two or more maximum delay constraints
which is satisfied when at least one of these delay constraints is satisfied.
\HTN{s} are meant as a light generalization of \STN{s} offering an interesting compromise.
On one side, there exist practical pseudo-polynomial time algorithms for checking the consistency of \HTN{s} and computing feasible schedules for them.
On the other side, \HTN{s} offer a more powerful model accommodating natural disjunctive constraints that cannot be expressed by \STN{s}.
In particular, \HTN{s} are weighted directed hypergraphs where each hyperarc models a disjunctive temporal constraint called \emph{hyperconstraint}.
The computational equivalence between checking consistency in {\HTN}s and determining winning regions in
\emph{Mean Payoff Games}~({\MPG}s)~\cite{EhrenfeuchtMycielski:1979, ZwickPaterson:1996, brim2011faster}
was also pointed out in~\cite{CPR2014, CPR2015}, where the approach was shown to be robust thanks to experimental evaluations (also see~\cite{BC12}).
\MPG{s} are a family of 2-player infinite pebble games played on finite graphs which is well known for having theoretical interest in computational complexity,
being one of the few natural problems lying in $\NP\cap\coNP$, as well as various applications in model checking and formal verification~\cite{Gradel2002}.

However, in the representation of quantitative temporal aspects of systems,
\emph{conditional} temporal constraints pose a serious challenge for conditional temporal planning,
where a planning agent has to determine whether a candidate plan will satisfy the specified conditional temporal constraints.
This can be difficult, because the temporal assignments that satisfy the constraints associated
with one conditional branch may fail to satisfy the constraints along a different branch (see, \eg\cite{TVP2003}).
The present work unveils that \HTN{s} and \MPG{s} are a natural underlying combinatorial model for checking the consistency of certain conditional
temporal problems that are known in the literature and that are useful in some practical applications of temporal planning,
especially, for managing the temporal aspects of Workflow Management Systems (WfMSs)~\cite{BettiniWJ02,CombiGPP12}
  and for modeling Healthcare's Clinical Pathways~\cite{CombiGMP14}. Thus we focus on \emph{Conditional Simple Temporal Networks~(\CSTN{s})} \cite{TVP2003, HPC12},
a constraint-based model for conditional temporal planning. The \CSTN formalism extends {\STN}s in that:
(1) some of the nodes are \emph{observation events}, to each of them is associated a boolean variable whose value is disclosed only at execution time;
(2) \emph{labels} (i.e. conjunctions over the literals)
are attached to all nodes \emph{and} constraints, to indicate the situations in which each of them is required.
The planning agent (or \emph{Planner}) must schedule all the required nodes, meanwhile respecting all the required temporal constraints among them.
This extended framework allows for the off-line construction of conditional plans that are guaranteed to satisfy complex networks of
temporal constraints. Importantly, this can be achieved even while allowing for the decisions about the precise timing of actions
to be postponed until execution time, in a least-commitment manner, thereby adding flexibility and making it possible to
adapt the plan dynamically, in response to the observations that are made during execution.
See~\cite{TVP2003} for further details and examples.

Three notions of consistency arise for \CSTN{s}: weak, strong, and \emph{dynamic}.
Dynamic consistency (DC) is the most interesting one; it requires the existence of conditional plans where
decisions about the precise timing of actions are postponed until execution time,
but it nonetheless guarantees that all the relevant constraints will be ultimately satisfied.
Still, it is the most challenging and it was conjectured to be hard to assess by~\cite{TVP2003}.
Indeed, to the best of our knowledge, the tightest currently known upper bound on the time complexity of
deciding whether a given \CSTN is dynamically-consistent is doubly-exponential time~\cite{TVP2003}.
It first builds an equivalent Disjunctive Temporal Problem (\DTP) of size exponential in the input \CSTN,
and then applies to it an exponential-time \DTP solver to check its consistency.
However, this approach turns out to be quite limited in practice: experimental studies have already
shown that the resolution procedures, as well as the currently known heuristics, for solving general
\DTP{s} become quite burdensome with $30$ to $35$ \DTP variables (see \eg\cite{TsamardinosP03, MoffittP05, Oddi14}),
thus dampening the practical applicability of the approach.

\subsection{Contribution}
In this work we introduce and study the \emph{Conditional Hyper Temporal Network (CHyTN)} model,
a natural extension and generalization of both the \CSTN and the \HTN model which is obtained by blending them together.
One motivation for studying it is to transpose benefits and opportunities for application,
that have arisen from the introduction of \HTN{s} (see~\cite{CPR2014, CPR2015}), to the context of \emph{conditional} temporal planning.
In so doing, the main and perhaps most important contribution is that to offer the first sound-and-complete
deterministic (pseudo) singly-exponential time algorithm for checking the dynamic consistency of \CSTN{s}.
After having formally introduced the CHyTN model, we start by showing that deciding whether a given \CSTN or CHyTN is dynamically-consistent is \coNP-hard.
Then, we offer a proof that deciding whether a given CHyTN is dynamically-consistent
is \PSPACE-hard, provided that the input CHyTN instances are allowed to include both multi-head and multi-tail hyperarcs.
In light of this, we focus on CHyTNs that allow only multi-head hyperarcs.
Concerning multi-head CHyTN{s}, perhaps most importantly, we unveil a connection between the problem of checking their dynamic consistency
and that of determining winning regions in \MPG{s} (of a singly-exponential size in the number of propositional variables of the input CHyTN),
thus providing the first sound-and-complete (pseudo) singly-exponential time algorithm
for this same task of deciding the dynamic consistency and yielding a dynamic execution strategy for multi-head CHyTNs.
The resulting worst-case time complexity of the DC-Checking procedure is actually
  	$O\big(2^{3|P|}|V||\A|m_{\A} + 2^{4|P|}|V|^2|\A||P| + 2^{4|P|}|V|^2m_{\A} + 2^{5|P|}|V|^3|P|\big)W$,
      where $|P|$ is the number of propositional variables, $|V|$ is the number of event nodes,
        $|\A|$ is the number of hyperarcs, $m_\A$ is the size (\ie roughly, the encoding length of $\A$),
          and $W$ is the maximum absolute integer value of the weights of the input CHyTN.
The algorithm is still based on representing a given CHyTN instance on an exponentially sized network, as first suggested in \cite{TVP2003}.
The difference, however, is that we propose to map \CSTN{s} and CHyTNs on (exponentially sized) \HTN{s}/\MPG{s} rather than on \DTP{s}.
This makes an important difference, because the consistency check for \HTN{s}
can be reduced to determining winning regions in \MPG{s}, as shown in~\cite{CPR2014, CPR2015},
which admits practical and effective pseudo-polynomial time algorithms (in some cases the algorithms for determining winning regions in \MPG{s} exhibit even
a strongly polynomial time behaviour, see \eg\cite{CPR2015, BC12, AllamigeonBG14, ChatterjeeHKN14}).
To summarize, we obtain an improved upper bound on the theoretical time complexity of the DC-checking of \CSTN{s}
(\ie from $\text{2-EXP}$ to $\text{pseudo-}E\cap\text{NE}\cap \text{coNE}$) together with a faster DC-checking procedure,
which can be used on CHyTN{s} with a larger number of propositional variables and event nodes than before.
At the heart of the algorithm a suitable reduction to \MPG{s} is mediated by the \HTN model,
\ie the algorithm decides whether a CHyTN{s} is dynamically-consistent by solving a carefully constructed \MPG.
In order to analyze the algorithm, we introduce a novel and refined notion of dynamic consistency,
named $\epsilon$-dynamic consistency (where $\epsilon\in\RR_+$), and present a sharp lower bounding analysis on the critical value of the
\emph{reaction time} $\hat{\varepsilon}$ where a CHyTN{s} transits from being, to not being, dynamically-consistent.
We believe that this contributes to clarifying (\wrt some previous literature, \eg\cite{TVP2003, HPC12})
the role played by the reaction time $\hat{\varepsilon}$ in checking the dynamic consistency of \CSTN{s}.
Moreover, the proof technique introduced in this analysis of $\hat{\varepsilon}$ is applicable more generally
when dealing with linear difference constraints which include strict inequalities;
  thus it may be useful in the analysis of other models of temporal constraints.

A preliminary version of this article appeared in the proceedings of the TIME symposium~\cite{CR2015}.
Here, the presentation is extended as follows: (1) the definition of \CSTN has been extended and generalized
to that of CHyTN in order to allow the presence of hyperarcs as labeled temporal
constraints already in the input instances;
(2) some further facts and pertinent properties about \CSTN{s} and CHyTN{s} have been established;
(3) for instance, the following hardness result: deciding whether a given CHyTN is dynamically-consistent
is {\PSPACE}-hard (the reduction goes from 3-CNF-TQBF),
provided that the input instances are allowed to include both multi-head and multi-tail hyperarcs;
(4) the proposed (pseudo) singly-exponential time algorithm  is presented here in its full generality, \ie \wrt the CHyTN model;
(5) several proofs have been polished, expanded and clarified (\eg those concerning the reaction time analysis of $\hat{\epsilon}$).

\subsection{Organization} The rest of the article is organized as follows.
Section~\ref{sect:backgroundandnotation} recalls the basic formalism, terminology and known results on \STN{s} and \HTN{s}.
Particularly, Subsection~\ref{subsect:stn} deals with \STN{s}; Subsection~\ref{subsect:HTN} deals with \HTN{s},
its computational equivalence with \MPG{s} and the related algorithmic results.
Section~\ref{sect:CSTN} surveys \CSTN{s} and, then, it introduces CHyTN{s}, also presenting some basic properties of the model.
Section~\ref{sect:Algo} tackles on the algorithmics of dynamic consistency: firstly, we provide a $\coNP$-hardness lower bound,
then we offer a $\PSPACE$-hardness lower bound. Next, it is described the connection with \HTN{s}/\MPG{s} and it is
devised a (pseudo) singly-exponential time DC-checking algorithm. Section~\ref{sect:epsilon} offers a
sharp lower bounding analysis on the critical value of the \emph{reaction time} $\hat{\varepsilon}$ where the \CSTN
transits from being, to not being, dynamically-consistent.
In Section~\ref{sect:relatedworks}, related works are discussed.
The article concludes in Section~\ref{sect:conclusions}.

\section{Background}\label{sect:backgroundandnotation}
%This section provides background notions concerning \STN{s} and \HTN{{s}, necessary to follow the rest of the treatise.
\subsection{Simple Temporal Networks}\label{subsect:stn}
Some definitions, notation and well know results about graphs and conservative graphs are introduced below;
also, we recall the relation between the consistency property of \STN{s} and the conservative property of weighted graphs.
Our graphs are directed and weighted on the arcs. Thus, if $G=(V,A)$ is a graph, 
then every arc $a\in A$ is a triplet $(t_a,h_a,w_a)$, where:
$t_a \in V$ is the \textit{tail} of $a$, $h_a \in V$ is the \textit{head} of $a$, and $w_a\in\RR$ is the \textit{weight} of $a$.
Moreover, since we use graphs to represent distance constraints, they do not need to have either loops
(unary constraints are meaningless) or parallel arcs (two parallel constraints represent two different distance
constraints between the same pair of nodes: only the most restrictive one is meaningful).
We also use the notations $h(a)$ for $h_a$, $t(a)$ for $t_a$, and $w(a)$ or $w(t_a,h_a)$ for $w_a$, when it helps.

The \textit{order} and \textit{size} of a graph $G = (V,A)$ are denoted by $n \triangleq |V|$ and $m \triangleq  |A|$, respectively;
the size is actually a measure for the encoding length of $G$.
Let $\N_+$ and $\RR_+$ be the set of positive natural and positive real numbers, respectively.
Let $[n]\triangleq \{1, 2, \ldots, n\}$, for every $n\in \N_+$.
A \emph{cycle} of $G$ is a set of arcs $C\subseteq A$ cyclically sequenced 
as $a_0, \ldots, a_{\ell-1}$ so that $h(a_i) = t(a_j)$ if and only if
$j=(i+1)\mod \ell$; this is called a \emph{negative cycle} if $w(C) < 0$, where $w(C)\triangleq \sum_{e\in C} w_e$.
A graph is called \textit{conservative} when it contains no negative cycle. A \textit{potential} is a map $p: V \mapsto \RR$.
The \textit{reduced weight} of an arc $a = (u,v,w_a)$ \wrt a potential $p$ is defined as $w^{p}_a \triangleq w_a - p_v + p_u$.
A potential $p$ of $G = (V,A)$ is called \textit{feasible} if $w^{p}_a\geq 0$ for every $a\in A$.
Notice that, for any cycle $C$, $w^p(C) = w(C)$. Therefore, the existence of a feasible potential implies that
the graph is conservative as $w(C) = w^p(C) \geq 0$ for every cycle $C$.
The Bellman-Ford algorithm~\cite{Cormen01} can be used to produce in $O(mn)$ time:

-- either a proof that $G$ is conservative in the form of a feasible potential function;

-- or a proof that $G$ is not conservative in the form of a negative cycle $C$ in $G$.

When the graph is conservative, the \emph{shortest path} between the nodes is well defined,
and for a fixed root node $r$ in $G$ the potentials returned by the Bellman-Ford algorithm are, 
for each node $v$, the shortest path from $r$ to $v$.
Moreover, if all the arc weights are integers, then these potentials are integers as well.
Therefore, the Bellman-Ford algorithm provides a proof of the following theorem.
% The above fact provides an algorithmic proof of the following main fact.
\begin{Thm}[\hspace{-1sp}\cite{Bellman58,Ford,Cormen01}]\label{teo:charConservativeGraphs}
A graph admits a feasible potential if and only if it is conservative.
When all the arc weights are integer valued,
 \ie $w_a\in\Z$ for every $a\in A$, then the feasible potential is integer valued as well.
\end{Thm}

An \STN can be viewed as a weighted directed simple graph whose nodes are events 
that must be placed on the real line and whose arcs express
mutual constraints on the allocations of their end-points.
An \STN $G = (V,A)$ is called \textit{consistent} if it admits 
a \emph{feasible schedule}, \ie a schedule $s: V\mapsto \RR$ such that:
\[
  s(v) \leq s(u) + w(u,v), \quad \forall \text{ arc $(u,v)$ of $G$.}
\]

\begin{Cor}[\hspace{-1sp}\cite{Bellman58,DechterMP91,Cormen01}]\label{cor.STNcharCons}
   An \STN $G$ is consistent if and only if $G$ is conservative.
\end{Cor}

In this paper, we also deal with directed weighted \emph{hypergraphs}.
% Usually, hypergraphs offer generalizations for graphs.
% We are no exception to this rule here:
\begin{Def}[General Hypergraph]
A \emph{general hypergraph} $\H$ is a pair $(V,\A)$, where $V$ is the set of nodes, and $\A$ is the set of \emph{hyperarcs}.
Each hyperarc $A\in\A$ is either a \emph{multi-head} or a \emph{multi-tail} hyperarc.

-- A \emph{multi-head} hyperarc $A=(t_A, H_A, w_A)$ has a distinguished node $t_A$, called the \emph{tail} of $A$, and a non-empty set
$H_A\subseteq V\setminus\{t_A\}$ containing the \emph{heads} of $A$;
to each head $v\in H_A$ is associated a \emph{weight} $w_A(v)\in\RR$, which is a real number (unless otherwise specified).
\figref{fig.multi-head} depicts a possible representation of a multi-head hyperarc:
the tail is connected to each head by a dashed arc labeled by the name of the hyperarc and the weight associated to the considered head.

-- A \textit{multi-tail} hyperarc $A=(T_A, h_A, w_A)$ has a distinguished node $h_A$, called the \emph{head} of $A$, and a non-empty set
$T_A\subseteq V\setminus\{h_A\}$ containing the \emph{tails} of $A$;
to each tail $v\in T_A$ is associated a \emph{weight} $w_A(v)\in\RR$, which is a real number (unless otherwise specified).
\figref{fig.multi-tail} depicts a possible representation of a multi-tail hyperarc:
the head is connected to each tail by a dotted arc labeled by the name of the hyperarc and the weight associated to the tail.
\end{Def}
\begin{figure}[!h]
\centering
\subfloat[Multi-Head Hyperarc $A=(t_A, H_A, w_A)$\label{fig.multi-head}.]{
\begin{tikzpicture}[arrows=->,scale=.8,node distance=.5 and 2]
    \node[node,xshift=1ex,label={above, yshift=.5ex:$H_A$}] (v1) {$v_1$};
    \node[node,below=of v1] (v2) {$v_2$};
    \node[node,below=of v2] (v3) {$v_3$};
	\node[node,left=of v2] (u) {$t_A$};
% 	\node[smallLabel,anchor=north east] at (v3.south -| u) {a) Hyperarc $a$.};
   	\draw[>=o, multiHead] (u) to node[timeLabel,above,sloped] {$A, w_A(v_1)$} (v1);%
   	\draw[>=o, multiHead] (u) to node[timeLabel,above,sloped] {$A, w_A(v_2)$} (v2);%
   	\draw[>=o, multiHead] (u) to node[timeLabel,above,sloped] {$A, w_A(v_3)$} (v3);%
 %%%%%%%%%%%%%
  \draw[dashed, ultra thin, rounded corners=15pt] (-.55,1) rectangle (.75,-2.75);
\end{tikzpicture}
}\qquad\qquad
\subfloat[Multi-Tail Hyperarc $A=(T_A, h_A, w_A)$\label{fig.multi-tail}.]{
\begin{tikzpicture}[arrows=<-,scale=.8,node distance=.5 and 2]
    \node[node,label={above, yshift=.5ex:$T_A$}] (v1) {$v_1$};
    \node[node,below=of v1] (v2) {$v_2$};
    \node[node,below=of v2] (v3) {$v_3$};
	\node[node,right=of v2] (u) {$h_A$};
% 	\node[smallLabel,anchor=north east] at (v3.south -| u) {a) Hyperarc $a$.};
   	\draw[>=o, multiTail] (u) to node[timeLabel,above,sloped] {$A, w_A(v_1)$} (v1);%
   	\draw[>=o, multiTail] (u) to node[timeLabel,above,sloped] {$A, w_A(v_2)$} (v2);%
   	\draw[>=o, multiTail] (u) to node[timeLabel,above,sloped] {$A, w_A(v_3)$} (v3);%
    %%%%%%%%%
    \draw[dashed, ultra thin, rounded corners=15pt] (-.65,1) rectangle (.6,-2.75);
\end{tikzpicture}
}
\caption{Multi-Head and Multi-Tail Hyperarcs.}
\end{figure}
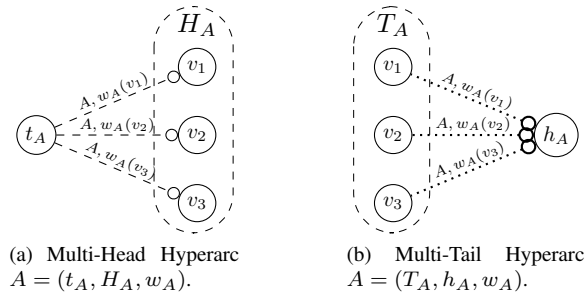

The \emph{cardinality} of a hyperarc $A\in \A$ is given by  $|A| \triangleq  |H_A\cup \{t_A\}|$ if $A$ is multi-head,
and $|A| \triangleq |T_A\cup\{h_A\}|$ if $A$ is multi-tail; if $|A|=2$, then $A=(u, v, w)$ is a standard arc.
The \textit{order} and \textit{size} of a general hypergraph $(V,\A)$
are denoted by $n \triangleq |V|$ and $m_\A \triangleq \sum_{A\in \A} |A|$, respectively.

\subsection{Hyper Temporal Networks}\label{subsect:HTN}

This subsection surveys the \textit{Hyper Temporal Network} (\HTN) model,
which is a strict generalization of \STN{s}, 
introduced to partially overcome the limitation of allowing only conjunctions of constraints.
\HTN{s} have been introduced in~\cite{CPR2014, CPR2015}, the reader is referred there for an in-depth treatment of the subject.
Compared to \STN distance graphs, which they naturally extend, 
\HTN{s} allow for a greater flexibility in the definition of the temporal constraints.

A general \HTN is a directed weighted general hypergraph $\H=(V,\A)$ where a node represents a time-point variable (or event node),
and where a multi-head\slash multi-tail hyperarc stands for a set of temporal 
distance constraints between the tail\slash heads and the head\slash tails (respectively).
Also, we shall consider two special cases of the general \HTN model, 
one in which all hyperarcs are only multi-head, and one where they're only multi-tail.
%###
In general, we say that a hyperarc is \textit{satisfied} when at least one of its distance constraints is satisfied.
Then, we say that a \HTN{} is \textit{consistent}
when it is possible to assign a value to each time-point variable so that all of its hyperarcs are satisfied.

More formally, in the \HTN framework the consistency problem is defined as the following decision problem.
\newcommand{\savefootnote}[2]{\footnote{\label{#1}#2}}
\newcommand{\repeatfootnote}[1]{\textsuperscript{\ref{#1}}}
\begin{Def}[\GTNC]
Given a general \HTN \mbox{$\H=(V,\A)$}, decide whether there exists a schedule \mbox{$s:V \rightarrow \RR$}
such that, for every hyperarc $A\in\A$, the following hold:
\begin{itemize}
\item if $A=(t,h,w)$ is a standard arc, then:
\[
   s(h)-s(t)\leq w;
\]
\item if $A=(t_A, H_A, w_A)$ is a multi-head hyperarc, then:
\[
   s(t_A) \geq \min_{v\in H_A} \{s(v) - w_A(v) \};
\]
\item if $A=(T_A, h_A, w_A)$ is a multi-tail hyperarc, then:
\[
   s(h_A) \leq \max_{v\in T_A} \{s(v) + w_A(v) \}.
\]
\end{itemize}
\end{Def}

Any such schedule \mbox{$s:V \rightarrow \RR$} is called \textit{feasible}.
  A \HTN that admits at least one feasible schedule is called \textit{consistent}.

Comparing the consistency of \HTN{s} with the consistency of \STN{s},
the most important aspect of novelty is that, while in a distance graph of a \STN each arc
represents a distance constraint and all such constraints have to be satisfied by any feasible schedule,
in a \HTN each hyperarc represents a disjunction of one or more distance constraints
and a feasible schedule has to satisfy at least one of such distance constraints for each hyperarc.

Let us survey some interesting properties about the consistency problem for \HTN{s}.
The first one is that any integer weighted \HTN admits an integer feasible schedule when it is consistent, as stated in the following lemma.
\begin{Lem}[\cite{CPR2015}]\label{Lem:int_sched}
Let $\H=(V,\A)$ be an integer weighted and consistent general $\HTN$.

Then $\H$ admits an integer feasible schedule $s:V \rightarrow \{-T,-T+1, \ldots, T-1, T\}$,
where $T = \sum_{A\in\A, v\in V} |w_A(v)|$.
\end{Lem}

%\delCar{The second interesting property is that deciding \GTNC is \NP-complete, as proved in the following theorem.}
\noindent The following theorem states that \GTNC is \NP-complete, in a strong sense.
\begin{Thm}[\cite{CPR2015}]\label{Teo:npcompleteness}
\GTNC is an \NP-complete problem even if input instances $\H=(V, \A)$
are restricted to satisfy $w_A(\cdot) \in\{-1, 0, 1\}$ and $|H_A|, |T_A|\leq 2$ for every $A\in\A$.
\end{Thm}

Theorem~\ref{Teo:npcompleteness} motivates the study of consistency problems on \HTN{s} having either
only multi-head or only multi-tail hyperarcs. In the former case, the consistency problem is called \HTNC,
while in the latter it is called \TTNC; as stated in Theorem~\ref{Teo:MainAlgorithms},
the complexity of checking these two problems turns out to be lower than that for \DTP{s}, \ie $\NP\cap\coNP$ instead of $\NP$-complete.

In the following theorem we observe that the two problems are inter-reducible,
\ie we can check consistency for any one of the two models in $f(m,n,W)$-time whenever we
have a $f(m,n,W)$-time procedure for checking consistency for the other one.

\begin{Thm}[\cite{CPR2015}]
\HTNC and \TTNC are inter-reducible by means of $\log$-space, linear-time, local-replacement reductions.
\end{Thm}

In the rest of this work we shall adopt the multi-head hypergraph as our reference model;
but we will consider general hypergraphs again in the forthcoming sections, when proving \PSPACE-hardness.
Let's say that, when considering hypergraphs and \HTN{s}, we will be implicitly referring
to the multi-head variant unless it is explicitly specified otherwise.
So, let us consider the following specialized notion of consistency for \HTN{s}.
\begin{Def}[\HTNC]
Given a (multi-head) \HTN \mbox{$\H=(V,\A)$}, decide whether there exists a schedule \mbox{$s:V \rightarrow \RR$} such that:
\[
   s(t_A) \geq \min_{v\in H_A} \{s(v) - w_A(v) \},\quad \forall A\in\A.
\]
%Any such a schedule \mbox{$s:V \rightarrow \R$} is called \textit{feasible}.
%A \HTN that admits at least one feasible schedule is called \textit{consistent}.
\end{Def}

\begin{Rem}\label{Rem:PLconvex}
Notice that this notion of consistency for \HTN{s} is a strict generalization of consistency for \STN{s}.
Generally, the feasible schedules of an \STN are the solutions of a linear system and, therefore, they form a convex polytope.
Since an \STN may be viewed as a \HTN, the space of feasible schedules of an \STN can always be described as the space of feasible schedules of a \HTN.
The converse is not true because feasible schedules for a \HTN need not form a convex polytope.
Let us consider, for example, a \HTN of just three nodes $x_1$, $x_2$, $x_3$ and a single hyperarc with heads $\{x_1,x_2\}$ and tail $x_3$ expressing the
constraint $x_3\geq \min\{x_1,x_2\}$.
Observe that $(0,2,2)$ and $(-2,0,2)$ are both admissible schedules, but $(1,1,0) = \frac{1}{2}(0,2,2) - \frac{1}{2}(-2,0,2)$ is not an admissible schedule.
In conclusion, the \STN model is a special case of the Linear Programming paradigm, whereas the \HTN model is not.
\end{Rem}

Next, we extend the characterization of \STN consistency (recalled in Subsection~\ref{subsect:stn}) to \HTN{s}.
% A mapping from $V$ to $\R$ may be called a \textit{schedule} or a \textit{potential} according to the current use.

\begin{Def}[Reduced Slack Value $w^{p}_A(v)$]
With reference to a potential $p:V\rightarrow \RR$, we define, for every arc $A\in \A$ and every $v\in H_A$,
the \textit{reduced slack value} $w^{p}_A(v)$ as $w_A(v) + p(t_A) - p(v)$
and the \textit{reduced slack} $w^{p}_A$ as follows:
\[ w^{p}_A \triangleq \max\{w^{p}_A(v) \mid v\in H_A\}. \]
A potential $p$ is said to be \emph{feasible} if and only if $ w^{p}_A \geq 0$ for every $A\in \A$.
\end{Def}
Notice that $w^{p}_A$ has been defined with $\max$ (instead of $\min$) because if every multi-head hyperarc $A$
  has at least one arc $(t_A,v)$ with positive $w^{p}_A(v)$ value, then the corresponding multi-head \HTN is consistent;
    also notice the similarity \wrt the potentials that are computed by the Bellman-Ford algorithm on \STN{s}.

Again, as it was the case for \STN{s}, a mapping $\phi:V \rightarrow \RR$ is a feasible potential if and only if it is a feasible schedule.
In order to better characterize feasible schedules, a notion of \emph{negative cycle} is introduced next.
\begin{Def}\label{def:negative_cycle}
Given a multi-head \HTN $\H=(V, \A)$, a \textit{cycle} is a pair $(S,\C)$ with $S\subseteq V$ and $\C \subseteq \A$ such that:
\begin{enumerate}
  \item $S = \bigcup_{A\in \C} \big(H_A\cup \{t_A\}\big)$ and $S\neq\emptyset$;
  \item $\forall v\in S$ there exists an unique $A\in \C$ such that $t_A = v$.
\end{enumerate}
Moreover, we let $a(v)$ denote the
unique arc $A\in \C$ with $t_A = v$ , as required in item~2 above.
% Cycles have a nice closure property:
% assume a pebble is initially placed in a node $v_1\in S$ and then, at time $t$, for $t=1,2,\ldots $, gets moved from $v_t$ to $v_{t+1}$, where  $v_{t+1}$ is any node in $T_{a(v_t)}$.
% By~1) and~2) this process will dictate an infinite sequence.
Every infinite path in a cycle $(S,\C)$ contains, at least, one \textit{finite cyclic sequence} $v_i, v_{i+1}, \ldots, v_{i+p}$,
where $v_{i+p} = v_i$ is the only repeated node in the sequence.
A cycle $(S,\C)$ is \textit{negative} if and only if the following holds:
\[
   \sum_{t=1}^{p-1} w_{a(v_t)}(v_{t+1}) < 0,\, \text{ for \emph{any} finite cyclic sequence } v_1, v_{2}, \ldots, v_{p}.
\]
\end{Def}

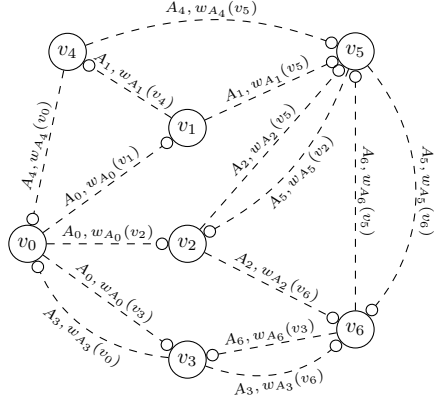
\begin{figure}[!h]
\centering
\begin{tikzpicture}[arrows=->,scale=.8,node distance=.5 and 2]
    \node[node,xshift=1ex] (v1) {$v_1$};
    \node[node,below=of v1, yshift=-5ex] (v2) {$v_2$};
    \node[node,below=of v2, yshift=-5ex] (v3) {$v_3$};
	  \node[node,left=of v2] (v0) {$v_0$};
    \node[node,xshift=3ex,yshift=2ex,above left=of v1] (v4) {$v_4$};
    \node[node,xshift=2ex,yshift=14ex,below right=of v1] (v5) {$v_5$};
    \node[node,xshift=2ex,yshift=-3ex,below right=of v2] (v6) {$v_6$};
   	\draw[>=o, multiHead] (v0) to node[timeLabel,above,sloped] {${A_0}, w_{A_0}(v_1)$} (v1);%
   	\draw[>=o, multiHead] (v0) to node[timeLabel,above,sloped] {${A_0}, w_{A_0}(v_2)$} (v2);%
   	\draw[>=o, multiHead] (v0) to node[timeLabel,above,sloped] {${A_0}, w_{A_0}(v_3)$} (v3);%
    \draw[>=o, multiHead] (v1) to node[timeLabel,above,sloped] {${A_1}, w_{A_1}(v_4)$} (v4);%
   	\draw[>=o, multiHead] (v1) to node[timeLabel,above,sloped] {${A_1}, w_{A_1}(v_5)$} (v5);%
    \draw[>=o, multiHead] (v2) to node[timeLabel,above,sloped] {${A_2}, w_{A_2}(v_5)$} (v5);%
    \draw[>=o, multiHead] (v2) to node[timeLabel,above,sloped] {${A_2}, w_{A_2}(v_6)$} (v6);%
    \draw[>=o, multiHead] (v4) to node[timeLabel,above,sloped] {${A_4}, w_{A_4}(v_0)$} (v0);%
    \draw[>=o, multiHead, bend left=20] (v4) to node[timeLabel,above,sloped] {${A_4}, w_{A_4}(v_5)$} (v5);%
    \draw[>=o, multiHead] (v6) to node[timeLabel,above,sloped] {${A_6}, w_{A_6}(v_5)$} (v5);%
    \draw[>=o, multiHead] (v6) to node[timeLabel,above,sloped] {${A_6}, w_{A_6}(v_3)$} (v3);%
    \draw[>=o, multiHead, bend right=30] (v3) to node[timeLabel,below,sloped] {${A_3}, w_{A_3}(v_6)$} (v6);%
    \draw[>=o, multiHead, bend left=30] (v3) to node[timeLabel,below,sloped] {${A_3}, w_{A_3}(v_0)$} (v0);%
    \draw[>=o, multiHead, bend left=20] (v5) to node[timeLabel,below,sloped] {${A_5}, w_{A_5}(v_2)$} (v2);%
    \draw[>=o, multiHead, bend left=40] (v5) to node[timeLabel,above,sloped] {${A_5}, w_{A_5}(v_6)$} (v6);%
\end{tikzpicture}
\caption{A Cycle $(S,\C)$, where $S=\{v_0, \ldots, v_6\}$ and $\C=\{A_0, \ldots, A_6\}$.}\label{fig:cycle}
\end{figure}
%, for $S=\{v_0, \ldots, v_6\}$ and $\C=\{\}$
\begin{Ex}
An example of a cycle $(S,\C)$ is shown in \figref{fig:cycle}; here,
$S=\{v_0, \ldots, v_6\}$ and $\C=\{A_0, \ldots, A_6\}$, where $t_{A_i}=v_i$ for every $i\in\{0,\ldots,6\}$;
moreover, $H_{A_0}=\{v_1,v_2,v_3\}$, $H_{A_1}=\{v_4,v_5\}$, $H_{A_2}=\{v_5,v_6\}$,
  $H_{A_3}=\{v_0,v_6\}$, $H_{A_4}=\{v_0,v_5\}$, $H_{A_5}=\{v_2,v_6\}$, $H_{A_6}=\{v_3,v_5\}$.
\end{Ex}
\begin{Lem}[\cite{CPR2015}]\label{lem:nc}
   A \HTN with a negative cycle $(S,\C)$ admits no feasible schedule.
\end{Lem}

At first sight, it may appear that checking whether $(S,\C)$ is a negative cycle might take exponential time since one should check a possibly exponential
number of cyclic sequences.
The next lemma asserts instead that it is possible to check the presence of a negative cycle in polynomial time.
% this is not the case.
% and, then, Theorem~\ref{teo:charCons} offers a good characterization of consistency for \HTN{s} generalizing the characterization of consistency for \STN{s} in
% Theorem~\ref{teo:charConservativeGraphs}.

\begin{Lem}[\cite{CPR2015}]
Let $(S, \C)$ be a cycle in a \HTN. Then,
  checking whether $(S,\C)$ is a negative cycle can be done in polynomial time.
\end{Lem}

A hypergraph $\H$ is called \emph{conservative} when it contains no negative cycle $(S,\C)$.

In the next paragraphs we will recall the existence of pseudo-polynomial time algorithms
that always return either a feasible schedule or a negative cycle certificate, thus extending the
validity of the classical good-characterization of \STN consistency to general \HTN consistency.
Here, we anticipate the statement of the main result in order to complete this brief introduction to \HTN{s}.
% is section about thesince it indicates a main laitmotif of the investigations to follow.
\begin{Thm}[\cite{CPR2015}]\label{teo:charCons}
   A \HTN $\H$ is \textit{consistent} if and only if it is \textit{conservative}.
   Moreover, when all weights are integers, then $\H$ admits an integer schedule if and only if it is conservative.
\end{Thm}
%%%%%%%%%%%%%%%%%%

From now on we shall focus on integer weighted multi-head hypergraphs and \HTN{s}.
\begin{Def}[Hypergraph]
A \emph{hypergraph} $\H$ is a pair $\langle V,\A\rangle$, where $V$ is the set of nodes, and $\A$ is the set of \emph{hyperarcs}.
Each hyperarc $A=\langle t_A, H_A, w_A\rangle\in \A$ has a distinguished node $t_A$, called the \emph{tail} of $A$, and a non-empty set
$H_A\subseteq V\setminus\{t_A\}$ containing the \emph{heads} of $A$;
to each head $v\in H_A$ is associated a \emph{weight} $w_A(v)\in\Z$.
\end{Def}

Again, provided that $|A| \triangleq |H_A\cup \{t_A\}|$, the \emph{size} of a hypergraph
$\H = \langle V,\A\rangle$ is defined as $m_{\A}\triangleq \sum_{A\in\A}|A|$,
and it is used as a measure for the encoding length of $\H$;
if $|A|=2$, then $A=\langle u, v, w \rangle$ is a \emph{standard arc}.
In this way, hypergraphs generalize graphs.

At this point, a (multi-head) \HTN is thus a weighted hypergraph $\H=\langle V,\A\rangle$ where a node represents an \emph{event} to be scheduled in time,
and a hyperarc represents a set of temporal distance \emph{constraints} between the \emph{tail} and the \emph{heads}.

The computational equivalence between checking the consistency of (integer weighted multi-head) {\HTN}s
and determining winning sets in {\MPG}s was pointed out in~\cite{CPR2014, CPR2015}.
The tightest currently known worst-case time complexity upper bound for solving \HTNC is expressed in the following theorem,
which was proved in \cite{CPR2015} by relying on the Value-Iteration Algorithm for \MPG{s}~\cite{brim2011faster};
the approach was shown to be robust thanks to experimental evaluations (also see \cite{BC12}).

\begin{Thm}{\cite{CPR2014, CPR2015}}\label{Teo:MainAlgorithms}
The following propositions hold on (integer weighted multi-head) \HTN{s}.
\begin{enumerate}
\item There exists an $O\big((|V|+|\A|) m_{\A} W\big)$ pseudo-polynomial time algorithm for checking \HTNC;
\item \label{Cor:PseudoPolyschedule}
  There exists an $O\big((|V|+|\A|) m_{\A} W\big)$ pseudo-polynomial time algorithm such that,
  given as input any consistent \HTN $\H=(V, \A)$, it returns as output a feasible schedule $\phi:V\rightarrow \RR$ of $\H$;
\item \label{Cor:PseudoPolyNegCycle}
There exists an $O\big((|V|+|\A|) m_{\A} W\big)$ pseudo-polynomial time algorithm such that,
  given as input any inconsistent \HTN $\H=(V, \A)$, it returns as output a negative cycle $(S,\C)$ of $\H$;
\end{enumerate}
Here, $W\triangleq \max_{A\in\A, v\in H_A} |w_A(v)|$.
\end{Thm}

In the forthcoming section we shall turn our attention to \emph{conditional} temporal planning,
where we generalize Conditional Simple Temporal Networks (\CSTN{s}) by introducing Conditional Hyper Temporal Networks (CHyTNs).

\section{Conditional Simple / Hyper Temporal Networks}\label{sect:CSTN}
In order to provide a formal support to the present work, this section recalls the basic formalism, terminology and known results on \CSTP{s} and \CSTN{s}.
Since the forthcoming definitions concerning \CSTN{s} are mostly inherited from the literature,
the reader is referred to \cite{TVP2003} and \cite{HPC12} for an intuitive semantic discussion and for some clarifying examples of the very same \CSTN model.
\cite{TVP2003} introduced the \emph{Conditional Simple Temporal Problem (\CSTP)}
as an extension of standard temporal constraint-satisfaction models used in non-conditional temporal planning.
\CSTP{s} augment \STN{s} by including \emph{observation} events, each one having a \emph{boolean variable} (or \emph{proposition})
associated with it. When an observation event is executed, the truth-value of its associated proposition becomes known.
In addition, each event node and each constraint has a \emph{label} that restricts the scenarios in which it plays a role.
Although not included in the formal definition, \cite{TVP2003} discussed some supplementary assumptions that any well-defined
\CSTP must satisfy. Subsequently, those conditions have been further analyzed and formalized by~\cite{HPC12},
leading to the definition of \emph{Conditional Simple Temporal Network} (\CSTN), which is now recalled.

Let $P$ be a set of boolean variables, a \emph{label} is any (possibly empty) conjunction of variables, or negations of variables, drawn from $P$.
The \emph{empty label} is denoted by $\lambda$.
The \emph{label universe} $P^*$ is the set of all (possibly empty) labels whose (positive or negative) literals are drawn from $P$.
Two labels, $\ell_1$ and $\ell_2$, are called \emph{consistent},
denoted\footnote{The notation $\Con(\cdot, \cdot)$ and $\Sub(\cdot, \cdot)$ is inherited from~\cite{TVP2003, HPC12}.} by $\Con(\ell_1, \ell_2)$,
when $\ell_1\wedge\ell_2$ is satisfiable. A label $\ell_1$ \emph{subsumes} a label $\ell_2$, denoted$^\text{1}$ by $\Sub(\ell_1, \ell_2)$,
when the implication $\ell_1\Rightarrow\ell_2$ holds.
Let us recall the formal definition of {\CSTN}s from~\cite{TVP2003, HPC12}.
\begin{Def}[{\CSTN}s]\label{def:cstn}
A \emph{Conditional Simple Temporal Network (\CSTN)} is a tuple $\langle V, A, L, \Ord, \Ord{V}, P \rangle$ where:
\begin{itemize}
\item $V$ is a finite set of \emph{events}; $P=\{p_1, \ldots, p_q\}$ (some $q\in\N$) is a finite set of \emph{boolean variables} (or \emph{propositions});
\item $A$ is a set of \emph{labeled temporal constraints (LTCs)} each having the form $\langle v-u\leq w(u,v), \ell\rangle$,
where $u,v\in V$, $w(u,v) \in \RR$, and $\ell\in P^*$;
\item $L:V\rightarrow P^*$ is a map that assigns a label to each event node in $V$;
$\Ord{V}\subseteq V$ is a finite set of \emph{observation events};
$\Ord:P\rightarrow \Ord{V}$ is a bijection mapping a unique observation event $\Ord(p)=\Ord_p$ to each $p\in P$;
\item The following \emph{well definedness assumptions} must hold:

(\emph{WD1})\; for any labeled constraint $\langle v-u\leq w, \ell\rangle\in A$ the label $\ell$ is satisfiable
and subsumes both $L(u)$ and $L(v)$; \ie whenever a constraint $v-u\leq w$ is required to be satisfied,
both of its endpoints $u$ and $v$ must be scheduled (sooner or later) by the Planner;

(\emph{WD2})\; for each $p\in P$ and each $u\in V$ such that either $p$ or $\neg p$ appears in $L(u)$, we require:
$\Sub(L(u), L(\Ord_p))$, and $\langle \Ord_p-u\leq-\epsilon, L(u)\rangle \in A$ for some (small) real $\epsilon > 0$;
\ie whenever a label $L(u)$ of an event node $u$ contains a proposition $p$, and $u$ gets eventually scheduled,
the observation event $\Ord_p$ must have been scheduled strictly before $u$ by the Planner.

(\emph{WD3})\; for each labeled constraint $\langle v-u\leq w, \ell\rangle$ and $p\in P$,
for which either $p$ or $\neg p$ appears in $\ell$, it holds that $\Sub(\ell, L(\Ord_p))$;
\ie assuming a required constraint contains proposition $p$,
the observation event $\Ord_p$ must be scheduled (sooner or later) by the Planner.

\end{itemize}
\end{Def}
We are now in the position to introduce the \emph{Conditional Hyper Temporal Network (CHyTN)},
a natural extension and generalization of both the \CSTN and the \HTN model obtained by blending them together.
Even though the original \STN and \CSTN models allow for real weights,
hereafter we shall restrict ourselves to the integers in order to rely on Theorem~\ref{Teo:MainAlgorithms}.
All of our \CSTN{s} and CHyTNs will be integer weighted from now on.
\begin{Def}[{CHyTN}s]
A general \emph{Conditional Hyper Temporal Network (CHyTN)} is a tuple $\langle V, \A, L, \Ord, \Ord{V}, P \rangle$, where
$V, P, L, \Ord$ and $\Ord{V}$ are defined as in CSTNs (see Definition~\ref{def:cstn}), and where
$\A$ is a set of \emph{labeled temporal hyper constraints (LTHCs)}, each having one of the following forms:
\begin{itemize}
\item $A=(t,h,w,\ell)$, where $(t,h,w)$ is a standard arc and $\ell\in P^*$; in this case, $A$ is called a \emph{standard} LTHC.
\item $A=(t_A, H_A, w_A, L_{H_A})$, where $(t_A, H_A, w_A)$  is a \emph{multi-head} hyperarc and $L_{H_A}:H_A\rightarrow P^*$ is
    a map sending each head $h\in H_A$ to a label $\ell_h$ in $P^*$; in this case, $A$ is called a \emph{multi-head} LTHC.
\item $A=(T_A, h_A, w_A, L_{T_A})$, where $A=(T_A, h_A, w_A)$ is a \emph{multi-tail} hyperarc and $L_{T_A}:T_A\rightarrow P^*$ is
    a map sending each tail $t\in T_A$ to a label $\ell_t$ in $P^*$; in this case, $A$ is called a \emph{multi-tail} LTHC.
\end{itemize}

\begin{itemize}

\item The following \emph{well definedness assumptions} must hold:

(\emph{WD1'})\; for any labeled constraint $A$:
\begin{itemize}
\item if $A=(t,h,w,\ell)$ is a standard LTHC, the label $\ell$ is satisfiable and subsumes both $L(t)$ and $L(h)$;
\item if $A=(t_A, H_A, w_A, L_{H_A})$ is a multi-head LTHC, for each $h\in H_A$ the label $L_{H_A}(h)$ is satisfiable and subsumes both $L(t_A)$ and $L(h)$;
\item if $A=(T_A, h_A, w_A, L_{T_A})$ is a multi-tail LTHC, for each $t\in T_A$ the label $L_{T_A}(t)$ is satisfiable and subsumes both $L(h_A)$ and $L(t)$;
\end{itemize}

(\emph{WD2})\; for each $p\in P$ and each $u\in V$ such that either $p$ or $\neg p$ appears in $L(u)$, we require:
$\Sub(L(u), L(\Ord_p))$, and $\langle \Ord_p-u\leq-\epsilon, L(u)\rangle \in \A$ for some (small) real $\epsilon > 0$;
  this is the same WD2 as defined for \CSTN{s}.

(\emph{WD3'})\; for each labeled constraint $A\in\A$ and boolean variable $p\in P$:
\begin{itemize}
\item if $A=(t,h,w,\ell)$ is a standard LTHC and $p$ or $\neg p$ appears in $\ell$, then $\Sub(\ell, L(\Ord_p))$;
\item if $A=(t_A, H_A, w_A, L_{H_A})$ is a multi-head LTHC
and either $p$ or $\neg p$ appears in $L_{H_A}(h)$ for some $h\in H_A$, then $\Sub(L_{H_A}(h), L(\Ord_p))$;
\item if $A=(T_A, h_A, w_A, L_{T_A})$ is a multi-tail LTHC
and either $p$ or $\neg p$ appears in $L_{T_A}(t)$ for some $t\in T_A$, then $\Sub(L_{T_A}(t), L(\Ord_p))$;
\end{itemize}
\end{itemize}
\end{Def}

Of course every CSTN is a CHyTN (\ie one having only standard LTHCs).
We shall adopt the notation $x\overset{[a,b], \ell}{\longrightarrow} y$, where $x,y\in V$, $a,b\in \N, a<b$ and $\ell\in P^*$,
to compactly represent the pair $\langle y-x\leq b, \ell \rangle, \langle x-y\leq -a, \ell \rangle\in A$;
also, whenever $\ell=\lambda$, we shall omit $\ell$ from the graphics, see~\eg~\figref{FIG:cstn1} and~\figref{FIG:chytn1} here below.

\begin{Ex}\label{example1}
Fig.~\ref{FIG:cstn1} depicts an example \CSTN $\Gamma_0=\langle V, A, L, \Ord, \Ord{V}, P\rangle$ having three event nodes
$A$, $B$ and $C$ as well as two observation events $\Ord_p$ and $\Ord_q$.
Formally, $V=\{A,B,C,\Ord_p, \Ord_q\}$, $P=\{p,q\}$, $\Ord{V}=\{\Ord_p, \Ord_q\}$,
$L(v)=\lambda$ for every $v\in V\setminus\{\Ord_q\}$ and $L(\Ord_q)=p$, $\Ord(p)=\Ord_p, \Ord(q)=\Ord_q$.
Next, the set of LTCs is: $A=\{ \langle C-A\leq 10, \lambda \rangle, \langle A-C\leq -10,
\lambda \rangle, \langle B-A\leq 3, p\wedge\neg q \rangle,
\langle A-B\leq 0, \lambda \rangle,
\langle \Ord_p-A\leq 5, \lambda \rangle,
\langle A-\Ord_p\leq 0, \lambda \rangle,
\langle \Ord_q-A\leq 9, p \rangle,
\langle A-\Ord_q\leq 0, p \rangle,
\langle C-B\leq 2, q \rangle,
\langle C-\Ord_p\leq 10, \lambda\rangle$.

Fig.~\ref{FIG:chytn1} depicts an example of a multi-head CHyTN $\Gamma_1=\langle V, \A, L, \Ord, \Ord{V}, P\rangle$.
Notice that $V,L,\Ord,\Ord{V}$ and $P$ are the same as in the CSTN $\Gamma_0$, whereas $\A$ is defined as follows:
$\A=A\cup\{\alpha\triangleq
  (B, \{C, \Ord_q\}, \langle w_\alpha(C), w_\alpha(\Ord_q)\rangle =
    \langle 2,-1\rangle, \langle L_\alpha(C), L_\alpha(\Ord_q)\rangle=\langle\lambda, p \rangle)\}$,
where $A$ is the set of LTCs of the CSTN $\Gamma_0$ and the additional constraint $\alpha$ is
a multi-head LTHC with tail $t_\alpha=B$ and heads $H_\alpha=\{C,\Ord_q\}$.
\end{Ex}

Sometimes we will show the scheduling time of a node with a label in boldface on the sidelines of the node itself,
  as for $A$ in \figref{fig:example1}.

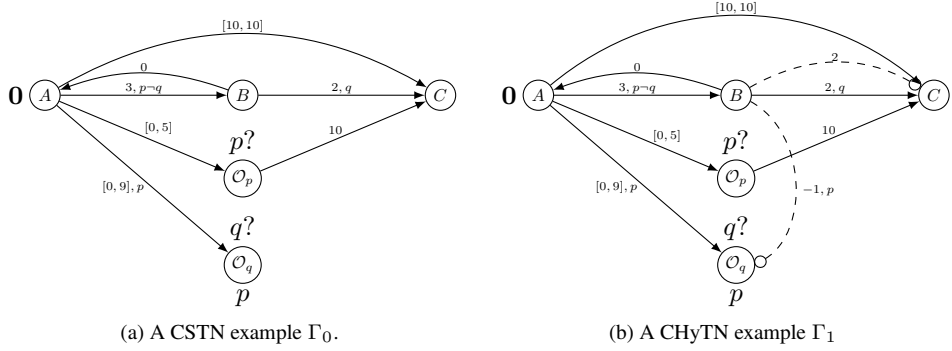
\begin{figure}[!h]
\centering
\subfloat[A CSTN example $\Gamma_0$.\label{FIG:cstn1}]{
\begin{tikzpicture}[arrows=->,scale=.65,node distance=1 and 1]
 	\node[node, label={left:$\bf 0$}] (A') {$A$};
	\node[node, xshift=15ex,right=of A'] (B') {$B$};
	\node[node, xshift=15ex,right=of B'] (C') {$C$};
	\node[node,below=of B', label={above,yshift=0ex:$p?$}] (P') {$\Ord_p$};
	\node[node,below=of P', label={above,yshift=0ex:$q?$}, label={below:$p$}] (Q') {$\Ord_q$};
	%arcs
	\draw[] (A') to [bend left=30] node[timeLabel,above] {$[10,10]$} (C');
	\draw[] (A') to [] node[timeLabel,above] {$3, p \neg q$} (B');
  \draw[] (B') to [bend right=20] node[timeLabel,above] {$0$} (A');
	\draw[] (B') to [] node[timeLabel,above] {$2, q$} (C');
	\draw[] (A') to [] node[xshift=2.1ex, yshift=0ex, timeLabel,above] {$[0,5]$} (P');
	\draw[] (A') to [] node[xshift=-2.5ex,yshift=0ex, timeLabel,below] {$[0,9],p$} (Q');
	\draw[] (P') to [] node[xshift=-1ex, timeLabel,above] {$10$} (C');
	\end{tikzpicture}
}
\quad
\subfloat[A CHyTN example $\Gamma_1$\label{FIG:chytn1}]{
\begin{tikzpicture}[arrows=->,scale=.65,node distance=1 and 1]
 	\node[node, label={left:$\bf 0$}] (A') {$A$};
	\node[node, xshift=15ex,right=of A'] (B') {$B$};
	\node[node, xshift=15ex,right=of B'] (C') {$C$};
	\node[node,below=of B', label={above,yshift=0ex:$p?$}] (P') {$\Ord_p$};
	\node[node,below=of P', label={above,yshift=0ex:$q?$}, label={below:$p$}] (Q') {$\Ord_q$};
	%arcs
	\draw[] (A') to [bend left=40] node[timeLabel,above] {$[10,10]$} (C');
  \draw[] (A') to [] node[timeLabel,above] {$3, p \neg q$} (B');
  \draw[] (B') to [bend right=20] node[timeLabel,above] {$0$} (A');
	\draw[] (B') to [] node[timeLabel,above] {$2, q$} (C');
	\draw[] (A') to [] node[xshift=4ex, yshift=-1ex, timeLabel,above] {$[0,5]$} (P');
	\draw[] (A') to [] node[xshift=-2.5ex,yshift=0ex, timeLabel,below] {$[0,9],p$} (Q');
	\draw[] (P') to [] node[xshift=-1ex, timeLabel,above] {$10$} (C');
		%hyper arc
	\draw[multiHead,>=o] (B') to [bend left=30] node[timeLabel,above] {$2$} (C');
	\draw[multiHead,>=o] (B') to [bend left=60] node[timeLabel,above,xshift=3ex, yshift=-2ex] {$-1, p$} (Q'.east);
\end{tikzpicture}
}
\caption{An example CSTN~(a), and an example CHyTN~(b).}\label{fig:example1}
\end{figure}

In the following definitions we will implicitly
  refer to some CHyTN which is denoted by
    $\Gamma=\langle V, \A, L, \Ord, \Ord{V}, P \rangle$.
\begin{Def}[Scenario]
A \emph{scenario} over a set $P$ of boolean variables is a truth assignment $s:P\rightarrow \{\top, \bot\}$, \ie
$s$ is a map that assigns a truth value to each proposition $p\in P$.
The set of all scenarios over $P$ is denoted by $\Sigma_P$.

If $s\in\Sigma_P$ and $\ell\in P^*$, then $s(\ell)\in\{\top, \bot\}$ denotes the truth value of $\ell$ induced by $s$ in the natural way.
\end{Def}
Notice that any scenario $s\in\Sigma_P$ can be described by means of the label
$\ell_s\triangleq l_1\wedge\cdots\wedge l_{|P|}$ such that,
for every $1\leq i\leq |P|$, the literal $l_i\in\{p_i, \neg p_i\}$ satisfies $s(l_i)=\top$.
\begin{Ex}
Let $P=\{p,q\}$. The scenario $s:P\rightarrow\{\top, \bot\}$ defined as $s(p)=\top$ and
$s(q)=\bot$ can be compactly described by the label $\ell_s=p\wedge \neg q$.
\end{Ex}

\begin{Def}[Schedule]
A \emph{schedule} for a subset of events $U\subseteq V$ is a map $\phi:U\rightarrow\RR$ that assigns a real number to each
event node in $U$. The set of all schedules over $U$ is denoted by~$\Phi_U$.
\end{Def}

\begin{Def}[Scenario Restriction]
Let $s\in\Sigma_{P}$ be a scenario.
The \emph{restriction} of $V$ and $\A$ \wrt $s$ are defined as:
\[ V^+_s\triangleq \Big\{v\in V\mid s(L(v))=\top\Big\}; \hspace*{4in} \]
\[\begin{array}{l}
	\A^+_s\triangleq \Big\{(u,v,w) \mid
			\exists ({\ell}\in P^*)\,\text{ s.t. } (u,v,w,\ell)\in \A \text{ and } s(\ell)=\top\Big\} \cup \\
\hspace{38pt}		\cup \Big\{(t, H'_A, w'_A) \mid
			\exists (H_A\supseteq H'_A; L_{H_A}:H_A\rightarrow P^*; w_A:H_A\rightarrow\Z)\, \text{ s.t. }
				(t, H_A, w_A, L_{H_A} ) \in \A,  \\
			\hfill w'_A={w_A}_{|_{H'_A}}, \forall (h\in H_A)\, s(L_{H_A}(h))=\top\iff h\in H'_A\Big\} \cup \\
\hspace{38pt}		\cup \Big\{(T'_A, h, w'_A) \mid
			\exists (T_A\supseteq T'_A; L_{T_A}:T_A\rightarrow P^*; w_A:T_A\rightarrow\Z)\, \text{ s.t. }
				(T_A, h, w_A, L_{T_A} ) \in \A,  \\
			\hfill w'_A={w_A}_{|_{T'_A}}, \forall (t\in T_A)\, s(L_{T_A}(t))=\top\iff t\in T'_A\Big\}. \\
\end{array}
\]

The restriction of $\Gamma$ \wrt $s$ is defined as $\Gamma^+_s\triangleq \langle V^+_s, \A^+_s\rangle$.

Finally, it is worthwhile to introduce the notation $V^+_{s_1, s_2} \triangleq V^+_{s_1}\cap V^+_{s_2}$.
\end{Def}
Note that if $\Gamma$ is a CHyTN, then $\Gamma^+_s$ is a \HTN; and if $\Gamma$ is a CSTN, then $\Gamma^+_s$ is an \STN.

%Moreover, two scenarios $s_1$ and $s_2$ are \emph{equivalent} \wrt $\Gamma$ when they induce the same restriction network,
%\ie when $\Gamma^+_{s_1}=\Gamma^+_{s_2}$, this defines an equivalence relation on $\Sigma_P$.
\begin{Ex}
\figref{FIG:restriction_cstn2} depicts the restriction \STN ${\Gamma_0}^+_s$ of the \CSTN $\Gamma_0$,
  and the restriction \HTN ${\Gamma_1}^+_s$ of the CHyTN $\Gamma_1$ (see Example~\ref{example1} and \figref{fig:example1}),
    \wrt the scenario $s(p)=s(q)=\bot$.
\end{Ex}
\begin{figure}[!htb]
\centering
\subfloat[The restriction STN ${\Gamma_0}^+_s$ of the CSTN $\Gamma_0$ \wrt $s(p)=s(q)=\bot$]{
\begin{tikzpicture}[arrows=->,scale=.65,node distance=1 and 1]
 	\node[node, label={left:$\bf 0$}] (A') {$A$};
	\node[node, xshift=15ex,right=of A'] (B') {$B$};
	\node[node, xshift=15ex,right=of B'] (C') {$C$};
	\node[node,below=of B'] (P') {$\Ord_p$};
	 %arcs
	\draw[] (A') to [bend left=30] node[timeLabel,above] {$[10,10]$} (C');
  \draw[] (B') to [] node[timeLabel,above] {$0$} (A');
	\draw[] (A') to [] node[xshift=2ex, yshift=0ex, timeLabel,above] {$[0,5]$} (P');
	\draw[] (P') to [] node[xshift=-1ex, timeLabel,above] {$10$} (C');
\end{tikzpicture}
}\quad
\subfloat[The restriction HyTN ${\Gamma_1}^+_s$ of the CHyTN $\Gamma_1$ \wrt $s(p)=s(q)=\bot$]{
\begin{tikzpicture}[arrows=->,scale=.65,node distance=1 and 1]
 	\node[node, label={left:$\bf 0$}] (A') {$A$};
	\node[node, xshift=15ex,right=of A'] (B') {$B$};
	\node[node, xshift=15ex,right=of B'] (C') {$C$};
	\node[node,below=of B'] (P') {$\Ord_p$};
       %arcs
	\draw[] (A') to [bend left=30] node[timeLabel,above] {$[10,10]$} (C');
  \draw[] (B') to [] node[timeLabel,above] {$0$} (A');
	\draw[] (A') to [] node[xshift=2ex, yshift=0ex, timeLabel,above] {$[0,5]$} (P');
	\draw[] (P') to [] node[xshift=-1ex, timeLabel,above] {$10$} (C');
	%hyper arc
	\draw[multiHead, >=o] (B') to [] node[timeLabel,above] {$2$} (C');
\end{tikzpicture}
}
\caption{The restriction ${\Gamma_0}^+_s$ (a), and the restriction ${\Gamma_1}^+_s$ (b),
\wrt the scenario $s(p)=s(q)=\bot$}\label{FIG:restriction_cstn2}
\end{figure}
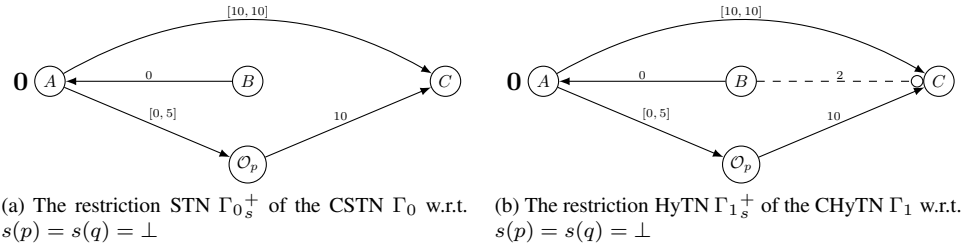
\begin{Def}[Execution Strategy~\cite{HPC12}]\label{def:executionstrategy}
An \emph{execution strategy} for $\Gamma$ is a mapping $\sigma:\Sigma_P\rightarrow \Phi_V$ such that,
for any scenario $s\in\Sigma_P$, the domain of the schedule $\sigma(s)$ is $V^+_{s}$.
The set of execution strategies of $\Gamma$ is denoted by $\S_{\Gamma}$.
The \emph{execution time} of an event node $v\in V^+_{s}$ in the schedule $\sigma(s)\in\Phi_{V^+_s}$ is denoted~by~$[\sigma(s)]_v$.
\end{Def}
\begin{Def}[Scenario History~\cite{HPC12}]\label{def:scenario_history}
Let $\sigma\in\S_{\Gamma}$ be an execution strategy, let $s\in\Sigma_P$ be a scenario and let $v\in V^+_{s}$ be an event node.
The \emph{scenario history} $\text{scHst}(v,s,\sigma)$ of $v$ in the scenario $s$ for the strategy $\sigma$ is defined as follows:
\[\text{scHst}(v,s,\sigma)\triangleq \Big\{\big(p, s(p)\big)\mid  p\in P,\, \Ord_p \in V^+_{s}\cap{\Ord}V,\, [\sigma(s)]_{\Ord_p} < [\sigma(s)]_v \Big\}.\]

\end{Def}
The scenario history can be compactly expressed by the conjunction
of the literals corresponding to the observations comprising it.
Thus, we may treat a scenario history as though it were a label.
\begin{Def}[Viable Execution Strategy~\cite{HPC12}]
We say that $\sigma\in\S_{\Gamma}$ is a \emph{viable} execution strategy whenever, for each scenario $s\in\Sigma_P$,
the schedule $\sigma(s)\in\Phi_V$ is feasible for the restriction \HTN (or \STN) $\Gamma^+_s$.
\end{Def}
\begin{Def}[Dynamic Consistency~\cite{HPC12}]\label{def:consistency}
An execution strategy $\sigma\in \S_{\Gamma}$ is called \emph{dynamic} if,
for any $s_1, s_2\in \Sigma_P$ and any event node $v\in V^+_{s_1, s_2}$,
the following implication holds:
\[\Con(\scHst(v, s_1, \sigma), s_2) \Rightarrow [\sigma(s_1)]_v = [\sigma(s_2)]_v.\tag{DC}\]
We say that $\Gamma$ is \emph{dynamically-consistent}
if it admits $\sigma\in\S_{\Gamma}$ which is both viable and dynamic.
\end{Def}
\begin{Def}[DC-Checking~\cite{HPC12}]\label{def:consistency_problems}
The problem of checking whether a given CHyTN (which allows \emph{both} multi-head and multi-tail LTHCs)
is dynamically-consistent is named \emph{\GHyDCC}.

That of checking whether a given CHyTN, allowing \emph{only} multi-head \emph{or only} multi-tail LTHCs,
is dynamically-consistent is named \emph{\HyDCC}. Checking whether a given CSTN is dynamically-consistent is named \emph{\DCC}.
\end{Def}

\begin{Ex}
Consider the CHyTN $\Gamma_1$ of \figref{FIG:chytn1}, and let the scenarios $s_1, s_2, s_3, s_4$ be defined as:
$s_1(p)=\top$, $s_1(q)=\top$; $s_2(p)=\top$, $s_2(q)=\bot$; $s_3(p)=\bot$, $s_3(q)=\top$; $s_4(p)=\bot$, $s_4(q)=\bot$.
The following defines an execution strategy $\sigma\in\S_\Gamma$:
$[\sigma(s_i)]_A=0$ for every $i\in\{1,2,3,4\}$;
$[\sigma(s_i)]_B=8$ for every $i\in\{1,3,4\}$ and $[\sigma(s_2)]_B=3$;
$[\sigma(s_i)]_C=10$ for every $i\in\{1,2,3,4\}$;
$[\sigma(s_i)]_{\Ord_p}=1$ for every $i\in\{1,2,3,4\}$.
The reader can check that $\sigma$ is viable and dynamic. Indeed,
  $\sigma$ admits the tree-like representation depicted in Fig~\ref{FIG:cstn2-strategy}.
\end{Ex}

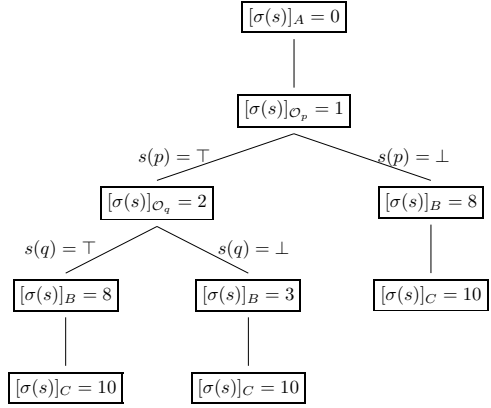
\begin{figure}[!htb]
\centering
\begin{tikzpicture}[scale=0.7, level distance=50pt,sibling distance=30pt]
\Tree [. \framebox{$[\sigma(s)]_{A}=0$}
\edge node[]{};
[. \framebox{$[\sigma(s)]_{\Ord_p}=1$}
\edge node[left,xshift=-1ex]{$s(p)=\top$};
[. \framebox{$[\sigma(s)]_{\Ord_q}=2$}
\edge node[left,xshift=-1ex]{$s(q)=\top$};
[. \framebox{$[\sigma(s)]_B=8$}
\edge node[]{}; \framebox{$[\sigma(s)]_C=10$}
]
\edge node[right,xshift=1ex]{$s(q)=\bot$};
[. \framebox{$[\sigma(s)]_B=3$}
\edge node[]{}; \framebox{$[\sigma(s)]_C=10$} ]
]
\edge node[right,xshift=1ex]{$s(p)=\bot$};
[. \framebox{$[\sigma(s)]_B=8$}
\edge node[]{}; [.
 \framebox{$[\sigma(s)]_C=10$} ]
 ]
]]
\end{tikzpicture}
\caption{A tree-like representation of a dynamic execution strategy $\sigma$ for the CHyTN $\Gamma_1$ of \figref{FIG:chytn1},
where $s$ denotes scenarios and $[\sigma(s)]_X$ is the corresponding schedule.}
\label{FIG:cstn2-strategy}
\end{figure}

Next, we recall a crucial notion for studying the dynamic consistency of CHyTNs:
  the \emph{difference set} $\Delta(s_1; s_2)$.
\begin{Def}[Difference Set~\cite{TVP2003}]
Let $s_1, s_2\in\Sigma_P$ be two scenarios.
The set of observation events in $V^+_{s_1}\cap{\Ord}V$ at which $s_1$ and $s_2$ differ is denoted by $\Delta(s_1;s_2)$.
Formally,
\[\Delta(s_1; s_2)\triangleq\big\{\Ord_p \in V^+_{s_1}\cap{\Ord}V \mid s_1(p)\neq s_2(p)\big\}.\]
\end{Def}
  Notice that commutativity may not hold (\ie generally it may be the case that $\Delta(s_1; s_2)\neq \Delta(s_2; s_1)$).
\begin{Ex}
Consider the \CSTN $\Gamma_0$ of \figref{FIG:cstn1} and the scenarios $s_1, s_2$ defined as follows:
$s_1\triangleq p\wedge q$; $s_2\triangleq \neg p\wedge \neg q$.

Then, $\Delta(s_1;s_2)=\{\Ord_p,\Ord_q\}$ and $\Delta(s_2;s_1)=\{\Ord_p\}$.
\end{Ex}
The next lemma will be useful later on in Section~\ref{sect:Algo}.
\begin{Lem}[\cite{TVP2003}]\label{lem:dynamicimplequality}
Let $s_1, s_2\in \Sigma_P$ and  $v\in V^+_{s_1, s_2}$.
Let $\sigma\in\S_{\Gamma}$ be an execution strategy.

Then, $\sigma$ is dynamic if and only if the following implication holds for every $s_1, s_2\in\Sigma_P$ and for every $u\in V^+_{s_1, s_2}$:
\[\Big(\bigwedge_{v\in\Delta(s_1; s_2)} [\sigma(s_1)]_u \leq [\sigma(s_1)]_v\Big)
\Rightarrow [\sigma(s_1)]_u = [\sigma(s_2)]_u \tag{L\ref{lem:dynamicimplequality}}\]
\end{Lem}
\begin{proof}
Notice that, by definition of $\Con(\cdot, \cdot)$ and $\scHst(\cdot, \cdot, \cdot)$, $\Con(\scHst(u, s_1, \sigma), s_2)$ holds
if and only if there is no observation event $v\in\Delta(s_1; s_2)$ which is scheduled by $\sigma(s_1)$ strictly before $u$.
Therefore, $\Con(\scHst(u, s_1, \sigma), s_2)$ holds if and only if $\bigwedge_{v\in\Delta(s_1; s_2)} [\sigma(s_1)]_u \leq [\sigma(s_1)]_v$.
At this point, substituting the $\Con(\scHst(u, s_1, \sigma), s_2)$ expression with the equivalent
formula $\bigwedge_{v\in\Delta(s_1; s_2)} [\sigma(s_1)]_u \leq [\sigma(s_1)]_v$ inside the definition of dynamic execution strategy
(see Definition~\ref{def:consistency}), the thesis follows.
\end{proof}

\section{Algorithmics of Dynamic Consistency}\label{sect:Algo}

Firstly, let us offer the following $\coNP$-hardness result for \DCC; notice that,
	since any CSTN is also a CHyTN, the same hardness result holds for CHyTN{s}.
\begin{Thm}\label{thm:dcc-conp-hard}
\DCC is $\coNP$-hard even if the input instances $\Gamma=\langle V, A, L, \Ord, \Ord{V}, P \rangle$
are restricted to satisfy $w_A(\cdot) \in\{-1, 0\}$ and $\ell\in \{p,\neg p\mid p\in P\}\cup\{\lambda\}$ for every $(u,v,w,\ell)\in A$.
\end{Thm}
\begin{proof}
We reduce $3$-\SAT\, to the complement of \DCC. Let $\varphi$ be a boolean formula in 3CNF.
Let $X$ be the set of variables and let $\C=\{C_0, \ldots, C_{m-1}\}$ be the set of clauses comprising $\varphi = \bigwedge_{j=0}^{m-1} C_j$.

(1) Let $N^\varphi$ be the \CSTN $\langle V^\varphi, A^\varphi, L^\varphi, \Ord^\varphi, \Ord{V}^\varphi, P^\varphi \rangle$,
where: $V^{\varphi}\triangleq X\cup \C $, and all the nodes are given an empty label, \ie $L^{\varphi}(v)=\lambda$ for every $v\in V^{\varphi}$;
each variable in $X$ becomes an observation event and each clause in $\C$ a non-observation,
	\ie $P^{\varphi}\triangleq \Ord{V^{\varphi}} \triangleq X$, so, $\Ord^{\varphi}$ is the identity map;
moreover, all observation events will be forced to be executed simultaneously before any of the non-observation events, thus
for every $u,v\in\Ord{V^{\varphi}}$ we have $\langle u-v\leq 0, \lambda\rangle \in A^{\varphi}$,
and for every $x\in X$ and $C\in \C$ we have $\langle x-C\leq -1, \lambda\rangle \in A^{\varphi}$;
finally, there is a negative loop among all the $C\in \C$ which plays an important role in the rest of the proof, particularly,
for each $j=0,\ldots, m-1$ and for each literal $\ell\in C_j$, we have $\langle C_j-C_{(j+1)\text{mod } m}\leq -1, \ell\rangle \in \A_\varphi$.
Notice that $|V^{\varphi}| = n+m$ and $|A^{\varphi}|=n^2+nm+3m$.

(2) We show that, if $\varphi$ is satisfiable, there must be an unavoidable negative circuit among all the $C_j\in \C$.
Assume that $\varphi$ is satisfiable. Let $\nu$ be a satisfying truth-assignment of $\varphi$.
In order to prove that $N^\varphi$ is not dynamically-consistent, observe that the restriction of $N^\varphi$ \wrt the scenario $\nu$ is an inconsistent \STN.
Indeed, if for every $j=0, \ldots, m-1$ we pick a standard arc $\langle C_j-C_{(j+1)\text{mod } m}\leq-1,\ell_j\rangle$
with $\ell_j$ being a literal in $C_j$ such that $\nu(\ell_j)=\top$, then we obtain a negative circuit.

(3) We show that, if $\varphi$ is unsatisfiable, there can't be a negative circuit among the $C_j\in \C$ because for each scenario,
there will be at least one $j$ such that all three labels, $\alpha_j$, $\beta_j$ and $\gamma_j$ will be false.
Assume that $\varphi$ is unsatisfiable. In order to prove that $N^{\varphi}$ is dynamically-consistent,
we exhibit a viable and dynamic execution strategy $\sigma$ for $N^{\varphi}$.
Firstly, schedule every $x\in X$ at $\sigma(x)\triangleq 0$.
Therefore, by time $1$, the strategy has full knowledge of the observed scenario $\nu$.
Since $\varphi$ is unsatisfiable, there exists an index $j_{\nu}$ such that $\nu(C_{j_\nu})=\bot$.
At this point, set $\sigma(C_{(j_{\nu}+k)\text{mod }m})\triangleq k$ for each $k=1, \ldots, m$.
The reader can verify that $\sigma$ is viable and dynamic for $N^{\varphi}$.
\end{proof}

\begin{figure}[!htb]
\centering
\begin{tikzpicture}[arrows=->,scale=.8,node distance=.5 and .6, node/.style={circle,draw,minimum size=22pt}]
 	% ZERO NODE:
	\node[node,label={below:$\bf 0$}] (0C) {$0_C$};
	% LETTER NODES
	\node[node, below = of 0C, yshift=-7ex, scale=0.1, blackNode] (Rp2) {};
	\node[node, left = of Rp2, xshift=3ex, scale=0.1, blackNode] (Rp1) {};
	\node[node, right = of Rp2, xshift=-3ex, scale=0.1, blackNode] (Rp3) {};
	\node[node, left = of Rp1, xshift=0ex, yshift=0ex, label={above:$x_1?$}] (Ox1) {$\Ord_{x_1}$};
	\node[node, right = of Rp3, xshift=0ex, yshift=0ex, label={above:$x_n?$}] (Oxn) {$\Ord_{x_n}$};
	% CLAUSE NODES
	\node[node, above left = of 0C, xshift=0ex, yshift=5ex] (Ci) {$c_i$};
	\node[node, above right = of 0C, xshift=0ex, yshift=5ex] (Cip1) {$c_{i+1}$};
	\node[node, left = of Ci, xshift=0ex, scale=0.1, blackNode] (Lp1) {};
	\node[node, left = of Lp1, xshift=3ex, scale=0.1, blackNode] (Lp2) {};
	\node[node, left = of Lp2, xshift=3ex, scale=0.1, blackNode] (Lp3) {};
	\node[node, left = of Lp3, xshift=0ex, yshift=0ex] (C1) {$c_1$};
	\node[node, right = of Cip1, xshift=0ex, scale=0.1, blackNode] (Lp11) {};
	\node[node, right = of Lp11, xshift=-3ex, scale=0.1, blackNode] (Lp12) {};
	\node[node, right = of Lp12, xshift=-3ex, scale=0.1, blackNode] (Lp13) {};
	\node[node, right = of Lp13, xshift=0ex, yshift=0ex] (Cm) {$c_m$};
	%arcs
	\draw[] (Ox1) to [] node[xshift=-2ex, yshift=0ex, timeLabel,above] {$[0,0]$} (0C);
	\draw[] (Oxn) to [] node[xshift=2ex, yshift=0ex, timeLabel,above] {$[0,0]$} (0C);
	\draw[] (C1) to [] node[xshift=0ex, yshift=0ex, timeLabel,below left] {$-1$} (0C);
	\draw[] (Ci) to [] node[xshift=0ex, yshift=0ex, timeLabel,below left] {$-1$} (0C);
	\draw[] (Cip1) to [] node[xshift=0ex, yshift=0ex, timeLabel,below right] {$-1$} (0C);
	\draw[] (Cm) to [] node[xshift=0ex, yshift=0ex, timeLabel,below right] {$-1$} (0C);
	%%%%%%%%%%%%%%%%%%%%%%%%%%%%%%%%%%%%%%%%%%%%%%%%%%%%%%%%%%%%%%%%%% c_1 => c_{i}%
	\draw[decoration={zigzag}, dotted] (C1) to [bend left = 16] node[xshift=0ex, yshift=0ex, timeLabel,above, scale=1.5] {} (Ci);
	\draw[dotted] (C1) to [] node[xshift=0ex, yshift=0ex, timeLabel,below, scale=1.5] {} (Ci);
	\draw[dotted] (C1) to [bend right = 21] node[xshift=0ex, yshift=0ex, timeLabel,below, scale=1.5] {} (Ci);
	%%% c_i => c_{i+1}
	\draw[] (Ci) to [bend left = 30] node[xshift=0ex, yshift=0ex, timeLabel,above] {$-1, \alpha_i$} (Cip1);
	\draw[] (Ci) to [] node[xshift=0ex, yshift=0ex, timeLabel,below] {$-1, \beta_i$} (Cip1);
	\draw[] (Ci) to [bend right = 35] node[xshift=0ex, yshift=0ex, timeLabel,below] {$-1, \gamma_i$} (Cip1);
	%%% c_1 => c_{i}
	\draw[dotted] (Cip1) to [bend left = 16] node[xshift=0ex, yshift=0ex, timeLabel,above, scale=1.5] {} (Cm);
	\draw[dotted] (Cip1) to [] node[xshift=0ex, yshift=0ex, timeLabel,below, scale=1.5] {} (Cm);
	\draw[dotted] (Cip1) to [bend right = 21] node[xshift=0ex, yshift=0ex, timeLabel,below, scale=1.5] {} (Cm);
	%%% c_1 => c_m
	\draw[] (Cm) to [bend right = 25] node[xshift=0ex, yshift=0ex, timeLabel,above] {$-1, \gamma_m$} (C1);
	\draw[] (Cm) to [bend right = 35] node[xshift=0ex, yshift=0ex, timeLabel,above] {$-1, \beta_m$} (C1);
	\draw[] (Cm) to [bend right = 45] node[xshift=0ex, yshift=0ex, timeLabel,above] {$-1, \alpha_m$} (C1);
	%%%%%%%%%%%%%%%%%%%%%%%%%%%%%%%%%%%%%%%%%%%%%%%%%%%%%%%%%%%%
\end{tikzpicture}
\caption{The \CSTN $N^{\varphi}$ where $\varphi(x_1, \ldots, x_n) = \bigwedge_{i=1}^m c_i$ for $c_i =(\alpha_i \vee \beta_i \vee \gamma_i)$.}
\label{FIG:cstn_Cvarphi}
\end{figure}
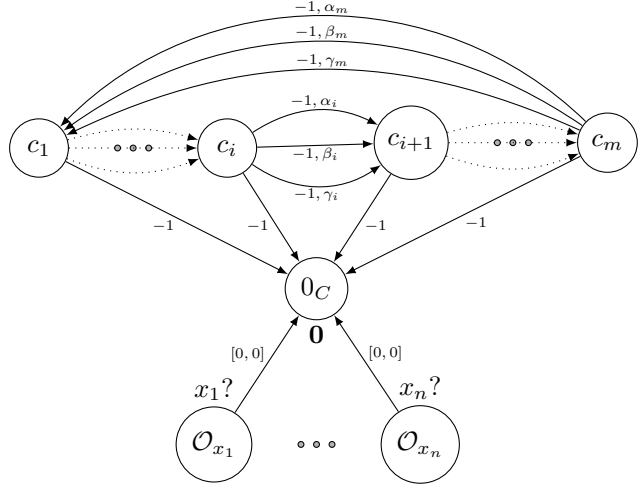

An illustration of the CSTN $N^{\varphi}$, which was constructed in the proof of Theorem~\ref{thm:dcc-conp-hard}, is shown in \figref{FIG:cstn_Cvarphi};
to ease the representation we have introduced an additional non-observation event $0_C$ in \figref{FIG:cstn_Cvarphi},
	which is executed at time $t=0$, together with all of the observation events in $X$.

Next, we show that when the input CHyTN instances are allowed to have \emph{both} multi-heads \emph{and} multi-tail LTHCs
then the DC-Checking problem becomes \PSPACE-hard.
\begin{Thm}\label{Teo:pspacehardness}
\GHyDCC is \PSPACE-hard, even if the input instances $\Gamma=\langle V, \A, L, \Ord, \Ord{V}, P \rangle$
	are restricted to satisfy the following two constraints:

-- $w_a(\cdot) \in [-n-1, n+1]\cap\Z$ and $\ell_a\in \{p, \neg p\mid p\in P\}\cup\{\lambda\}$
	for every weight $w_a$ and label $\ell_a$ appearing in any \emph{standard} LTHC $a\in\A$;

-- $w_A(\cdot)\in\{-1, 0, 1\}$, $\ell_A=\lambda$ and $|A|\leq 2$ for every weight $w_A$
	and label $\ell_A$ appearing in any \emph{multi-tail/head} LTHC $A\in \A$.
\end{Thm}
\begin{proof}
To show that \GHyDCC is \PSPACE-hard, we describe a reduction from the problem 3-CNF-TQBF (True Quantified Boolean Formula in 3-CNF).

Let us consider a $3$-CNF quantified boolean formula  with $n\geq 1$ variables and $m\geq 1$ clauses:
\[\varphi(x_1, \ldots, x_n) = Q_1x_1\ldots Q_n x_n \bigwedge_{i=1}^m (\alpha_i \vee \beta_i \vee \gamma_i), \]
where for every $j\in [n]$ the symbol $Q_j$ is either $\exists$ or $\forall$, and where
$\C_i = (\alpha_i \vee \beta_i \vee \gamma_i)$ is the $i$-th clause of $\varphi$
and  each $\alpha_i,\beta_i,\gamma_i\in \{x_j, \neg x_j\mid 1\leq j\leq n\}$ is a positive or negative literal.
We also say that $Q_1x_1\ldots Q_n x_n$ is the \emph{prefix} of $\varphi$.

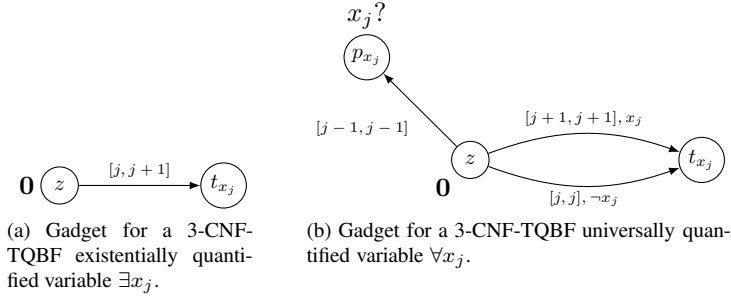
\begin{figure}[!htb]
\centering
\subfloat[Gadget for a 3-CNF-TQBF existentially quantified variable $\exists x_j$.\label{FIG:timeline_existential}]{
\begin{tikzpicture}[arrows=->,scale=.8,node distance=2 and 2]
	\node[node, label={left:$\bf 0$}] (zero) {$z$};
	\node[node, right=of zero] (Xjex) {$t_{x_j}$};
	%arcs
	\draw[] (zero) to [] node[above, timeLabel, yshift=.5ex] {$[j,j+1]$} (Xjex);
	\end{tikzpicture}\label{FIG:timeline_existential}
}\qquad
\subfloat[Gadget for a 3-CNF-TQBF universally quantified variable $\forall x_j$.\label{FIG:timeline_universal}]{
\begin{tikzpicture}[arrows=->,scale=.8,node distance=2 and 2]
	\node[node, label={below left:$\bf 0$}] (zero) {$z$};
	\node[node, above left=of zero, xshift=5ex, yshift=-5ex, label={above:$x_j?$}] (Xjunp) {$p_{x_j}$};
	\node[node, right=of zero, xshift=7ex] (Xjun) {$t_{x_j}$};
	%arcs
	\draw[] (zero) to [] node[above, timeLabel, xshift=-6ex, yshift=-3ex] {$[j-1,j-1]$} (Xjunp);
	\draw[] (zero) to [bend left=20] node[above, timeLabel, yshift=.5ex] {$[j+1,j+1], x_j$} (Xjun);
	\draw[] (zero) to [bend right=20] node[below, timeLabel, yshift=-.25ex] {$[j,j], \neg x_j$} (Xjun);
\end{tikzpicture}\label{FIG:timeline_universal}
}
\caption{Gadgets for quantified variables used in the reduction from 3-CNF-TQBF to \GHyDCC.}\label{fig.timeline}
\end{figure}

\underline{Construction.} We associate to $\varphi$ a CHyTN $\Gamma_{\varphi}=\langle V, \A, L, \Ord, \Ord{V}, P \rangle$.
In so doing, our first goal is to simulate the interaction between two players:
Player-$\exists$ (corresponding to the Planner in CHyTNs) and Player-$\forall$ (corresponding to the Nature in CHyTNs),
which corresponds directly to the chain of alternating quantifiers in the prefix of $\varphi$.
Naturally, the Planner is going to control those variables that are quantified existentially in $\varphi$,
whereas the Nature is going to control (by means of some observation events in $\Ord{V}$) those variables that are quantified universally in $\varphi$.
Briefly, $P$ contains one boolean variable for each universally quantified variable of $\varphi$,
and $V$ contains the following: two special events $z$ and $z’$ to be executed at time $0$ and $n+1$, respectively;
an observation event $p_{x_j}$ for each universally quantified variable $\forall x_j$;
a non-observation event $t_{x_j}$ for each quantified variable $x_j$; two non-observation events $l_{x_j}$ and $l_{\overline{x}_j}$ for each quantified variable $x_j$,
these will play (respectively) the role of positive and negative literals of $\varphi$
	(\ie the $\alpha$, $\beta$ and $\gamma$ in each clause $\C_i$); finally, a non-observation event $\C_i$ for each clause.

Let us describe the low-level details of $\Gamma_{\varphi}$.
We let $P\triangleq \{x_j\mid ``\forall x_j" \text{ appears in the prefix of } \varphi\}$.
Moreover, $V$ contains a node $z$ (\ie the \emph{zero} node that has to be executed at time $t=0$).

Next, for each existential quantification $\exists x_j$ in the prefix of $\varphi$,
$V$ contains a node named $t_{x_j}$ and $\A$ contains the following two standard LTHCs:
$(z,t_{x_j}, j+1, \lambda)$ and $(t_{x_j}, z, -j, \lambda)$; the underlying intuition being that,
during execution, it will be the responsibility of the Planner to schedule $t_{x_j}$ either at time $j$
(and this means that the Planner chooses to set $x_j$ to \texttt{false} in $\varphi$) or
	at time $j+1$ (and this means that he chooses to set $x_j$ to \texttt{true} in $\varphi$).
See \figref{FIG:timeline_existential} for an illustration of the $\exists x_j$ gadget.

Moreover, for each universal quantification $\forall x_j$ in the prefix of $\varphi$ (\ie for each $x_j\in P$),
$V$ contains two nodes named $p_{x_j}$ and $t_{x_j}$. Particularly, $p_{x_j}$ is an observation event (\ie $p_{x_j}\in {\Ord}V$) such
that $\Ord(x_j)=p_{x_j}$; hence, ${\Ord}V\triangleq \{p_{x_j}\mid x_j\in P\}$.
Also, for each $\forall x_j$ in $\varphi$'s prefix (\ie for each $x_j\in P$), $\A$ contains the following six standard LTHCs:
$(z, p_{x_j}, j-1, \lambda)$, $(p_{x_j}, z, -j+1, \lambda)$, $(z, t_{x_j}, j+1, x_j)$, $(t_{x_j}, z, -j-1, x_j)$,
$(z, t_{x_j}, j, \neg x_j)$ and $(t_{x_j}, z, -j, \neg x_j)$; the underlying intuition being that
the Nature must choose whether to schedule $t_{x_j}$ at time $j$ (setting $x_j$ to \texttt{false} in $\varphi$
by controlling the observation event $p_{x_j}$) or at time $j+1$ (setting $x_j$ to \texttt{true} in $\varphi$
again, by controlling the observation event $p_{x_j}$).
\figref{FIG:timeline_universal} illustrates the gadget for universally quantified variables $\forall x_j$.

In both cases (existentially and universally quantified variables),
the weights of the involved standard temporal constraints depend on $j$
in such a way that their scheduling times and their corresponding propositional
choices must occur one after the other in time. More precisely, for every $j\in [n]$,
$t_{x_j}$ is going to be scheduled either at time $j$ (if $x_j$ is \texttt{true} in $\varphi$)
or at time $j+1$ (when instead $x_j$ is \texttt{false} in $\varphi$). In addition to this,
when $x_j$ is quantified universally in $\varphi$ (\ie when $x_j\in P$),
the observation event that determines its propositional value (\ie $p_{x_j}$) is always scheduled at time $j-1$
(and this leaves enough space for the reaction time; actually, an entire unit of time between time $j-1$ and time $j$).

This concludes the description of our gadgets for simulating the chain of alternating quantifiers in the prefix of $\varphi$.

At this point, we have an additional node in $V$, named $z'$, which is always scheduled at time $n+1$; for this,
$\A$ contains the following two standard LTHCs: $(z, z', n+1, \lambda)$ and $(z', z, -n-1, \lambda)$.
Next, we shall describe two additional gadgets (that make use of $z’$) for simulating the 3-CNF formula
$\bigwedge_{i=1}^m (\alpha_i \vee \beta_i \vee \gamma_i)$, one for the literals, and one for the clauses.
We have a gadget for the positive (\ie $x_j$) and the negative (\ie $\neg x_j$) literals.
It goes as follows: for each $j\in [n]$, $V$ contains two nodes named $l_{x_j}$ (\ie positive literal) and $l_{\overline{x}_j}$ (\ie negative literal).
Moreover, $\A$ contains the following four standard LTHCs, $(z', l_{x_j}, 1, \lambda), (l_{x_j}, z', 0, \lambda)$
and $(z', l_{\overline{x}_j}, 1, \lambda), (l_{\overline{x}_j}, z', 0, \lambda)$,
plus the following \emph{multi-head} LTHC,
\[A^{h}({l_{x_j}, l_{\overline{x}_j}}) \triangleq  \Big(z', \{l_{x_j}, l_{\overline{x}_j}\}, \langle w(l_{x_j}), w(l_{\overline{x}_j})\rangle
	= \langle 0, 0\rangle, \langle L(l_{x_j}), L(l_{\overline{x}_j})\rangle=\langle\lambda,\lambda\rangle\Big),\]
and the following \emph{multi-tail} LTHC,
\[A^{t}({l_{x_j}, l_{\overline{x}_j}}) \triangleq  \Big(\{l_{x_j}, l_{\overline{x}_j}\}, z', \langle w(l_{x_j}), w(l_{\overline{x}_j})\rangle
	= \langle -1, -1\rangle, \langle L(l_{x_j}), L(l_{\overline{x}_j})\rangle=\langle\lambda,\lambda\rangle\Big).\]
The idea here is that the standard LTHCs are going to force the scheduling times of both $l_{x_j}$ and $l_{\overline{x}_j}$
to fall within the real interval $[n+1, n+2]$ (\ie not before $z'$ and at most $1$ time unit after $z'$).
Meanwhile, the multi-head constraint $A^{h}({l_{x_j}, l_{\overline{x}_j}})$
forces that at least one between $l_{x_j}$ and $l_{\overline{x}_j}$ happen
not later than time $n+1$ (\ie not later than the scheduling time of $z'$);
similarly, the multi-tail constraint $A^{t}({l_{x_j}, l_{\overline{x}_j}})$ is going to force that
at least one between $l_{x_j}$ and $l_{\overline{x}_j}$ happen not before time $n+2$ (\ie not before the scheduling time of $z'$ plus $1$).
Therefore, exactly one between $l_{x_j}$ and $l_{\overline{x}_j}$ will be forced to happen at time $n+1$, and the other one at time $n+2$.

Up to this point, the key idea is that, for every $j\in [n]$, we can force the scheduling time of each node $l_{x_j}$ and $l_{\overline{x}_j}$
to be uniquely determined, according to a suitable translation of the scheduling time of $t_{x_j}$.
Particularly, we want to schedule at time $n+1$ (\ie at the same scheduling time of $z'$) the one node between $l_{x_j}$
and $l_{\overline{x}_j}$ whose corresponding literal was chosen to be \texttt{false} in $\varphi$
(that is $l_{x_j}$ if $t_{x_j}$ was scheduled at time $j$, and $l_{\overline{x}_j}$ if $t_{x_j}$ was scheduled at time $j+1$);
similarly, we want to schedule at time $n+2$ (\ie at the same scheduling time of $z'$ \emph{plus} $1$ time unit)
the one node between $l_{x_j}$ and $l_{\overline{x}_j}$ whose corresponding literal was chosen to be \texttt{true}
(that is $l_{x_j}$ if $t_{x_j}$ was scheduled at time $j+1$, and $l_{\overline{x}_j}$ if $t_{x_j}$ was scheduled at time $j$).
In order to achieve this, for each $j\in [n]$, $\A$ contains the following two standard LTHCs:
$(t_{x_j}, l_{x_j}, n+1-j, \lambda)$ and $(l_{x_j}, t_{x_j}, -n-1+j, \lambda)$
	(in~\figref{FIG:Var_j} they are depicted with a unique arc $t_{x_j}\overset{[k,k], \lambda}{\longrightarrow} l_{x_j}$ where $k=n+1-j$);
in this way, $l_{x_j}$ is forced to happen at the same time of $t_{x_j}$ \emph{plus} $n+1-j$ units of time.
Therefore, if $t_{x_j}$ was scheduled at time $j$ (\ie $x_j$ is \texttt{false} in $\varphi$), then node $l_{x_j}$ is scheduled at time $j+n+1-j=n+1$;
otherwise, if $t_{x_j}$ was scheduled at time $j+1$ (\ie $x_j$ is \texttt{true} in $\varphi$), then node $l_{x_j}$ is scheduled at time $j+1+n+1-j=n+2$.
At this point, the scheduling time of the node $l_{\overline{x}_j}$ is determined uniquely thanks to the hyperarcs $A^{h}(l_{x_j}, l_{\overline{x}_j}), A^{t}(l_{x_j}, l_{\overline{x}_j})$
and the standard constraints $(z', l_{\overline{x}_j}, 1, \lambda)$, $(l_{\overline{x}_j}, z', 0, \lambda)$:
if the node $l_{x_j}$ is scheduled at time $n+1$ (\ie if $x_j$ is \texttt{false} in $\varphi$),
then $l_{\overline{x}_j}$ must be scheduled at time $n+1+1=n+2$ (\ie if $\overline{x}_j$ is \texttt{true} in $\varphi$) so that to satisfy
$A^{t}(l_{x_j}, l_{\overline{x}_j})$ and $(z', l_{\overline{x}_j}, 1, \lambda)$;
otherwise, if $l_{x_j}$ is scheduled at time $n+2$ (\ie if $x_j$ is \texttt{true} in $\varphi$),
then $l_{\overline{x}_j}$ must be scheduled at time $n+1+0=n+1$ (\ie if $\overline{x}_j$ is \texttt{false} in $\varphi$)
so that to satisfy $A^{h}(l_{x_j}, l_{\overline{x}_j})$ and $(l_{\overline{x}_j}, z', 0, \lambda)$.
Notice that the literals $\alpha_i,\beta_i,\gamma_i$ of $\varphi$
	are thus instances of the nodes $l_{x_i}$ or $l_{\overline{x}_i}$ described in~\figref{FIG:Var_j}.

\begin{figure}[tb]
\centering
\subfloat[Gadget for 3-CNF-TQBF positive $x_j$ and negative $\neg x_j$ literal.]{
\begin{tikzpicture}[arrows=->,scale=.8,node distance=2 and 2]
	\node[node] (zero1) {$z'$};
	\node[node, above right=of zero1] (nX) {$l_{\overline{x}_j}$};
	\node[node, above left=of zero1] (X) {$l_{x_j}$};
	\node[node, below=of zero1, label={below: $\bf 0$}] (zero) {$z$};
	\node[node, left=of zero] (Txj) {$t_{x_j}$};
 	%arcs
	%%%%%%% X's arcs %%%%%%
	\draw[] (zero1) to [bend left=40] node[below,timeLabel,xshift=-1ex] {$1$} (X);
	\draw[] (X) to [bend left=40] node[above,timeLabel,xshift=1ex] {$0$} (zero1);
	\draw[dashed, thick] (zero1) to [bend right=15] node[timeLabel,xshift=1.5ex] {$0$} (X);
	\draw[dotted, thick] (X) to [bend right=15] node[timeLabel,xshift=-2ex] {$-1$} (zero1);
	%%%%%%% nX's arcs %%%%%
	\draw[] (zero1) to [bend left=40] node[above,timeLabel,xshift=-1ex] {$1$} (nX);
	\draw[] (nX) to [bend left=40] node[below,timeLabel,xshift=1ex] {$0$} (zero1);
	\draw[dashed, thick] (zero1) to [bend left=15] node[timeLabel,xshift=-1ex] {$0$} (nX);
	\draw[dotted, thick] (nX) to [bend left=15] node[timeLabel,xshift=1ex] {$-1$} (zero1);
	\draw[] (zero) to [bend right=0] node[right,timeLabel] {$[n+1,n+1]$} (zero1);
	\draw[] (zero) to [bend right=10] node[above] {$\ldots$} (Txj);
	\draw[] (Txj) to [bend right=10] node[below] {$\ldots$} (zero);
	\draw[] (Txj) to [bend left=10] node[left,timeLabel] {$[n+1-j,n+1-j]$} (X);
\end{tikzpicture}\label{FIG:Var_j}
}
%\end{figure}
\quad
%\begin{figure}[th]
\subfloat[Gadget for 3-CNF-TQBF clause $\C_i = (\alpha_i \vee \beta_i \vee \gamma_i)$
where each $\alpha_i, \beta_i, \gamma_i$ is a positive or negative literal.]{
\begin{tikzpicture}[arrows=->,scale=.85,node distance=1.5 and 2]
% 	\draw[help lines,step=10pt] (-93pt,-155pt) grid (93pt, 25pt);
	\clip (-94pt,-155pt) rectangle (94pt, 25pt);
	\node[node,label={above:$\bf n+2$}] (one) {$\C_i$};
	\node[node,below =of one] (beta) {$\beta_i$};
	\node[node,left=of beta] (alpha) {$\alpha_i$};
	\node[node,right=of beta] (gamma) {$\gamma_i$};
 	\node[node,label={below:$\bf n+1$}, below=of beta] (zero) {$z'$};
 	\coordinate (fakeL) at ($(alpha.west)+(-1,0)$);
 	\coordinate (fakeR) at ($(gamma.east)+(1,0)$);
	%arcs
	%%%%%%% zero/one arcs %%%%%
 	\draw[] (zero.west) .. controls ($(fakeL)+(0,-10mm)$) and ($(fakeL)+(0,10mm)$) .. node[left,pos=.2] {$+1$} (one);
  	\draw[] (one) .. controls ($(fakeR)+(0,10mm)$) and ($(fakeR)+(0,-10mm)$) .. node[right,pos=.2] {$-1$} (zero.east);
	%%%%%%% alpha's arcs %%%%%%
	\draw[] (alpha) to [bend right=20] node[timeLabel,below] {} (zero);
	\draw[] (zero) to [bend right=20] (alpha);
	\draw[dotted, thick] (alpha) to [bend right=5] (zero);
	\draw[dashed, thick] (zero) to [bend right=5] (alpha);
	\draw[dotted, thick] (alpha) to [bend right=25] node[timeLabel,below] {$0$} (one.south west);
	%%%%%%% beta's arcs %%%%%
	\draw[] (zero) to [bend left=25] (beta);
	\draw[] (beta) to [bend left=25] (zero);
	\draw[dotted, thick] (beta) to [bend right=10] (zero);
	\draw[dashed, thick] (zero) to [bend right=10] (beta);
	\draw[dotted, thick] (beta) to [] node[timeLabel,right] {$0$} (one.south);
	%%%%%%% gamma's arcs %%%%%%
	\draw[] (gamma) to [bend left=20] (zero);
	\draw[] (zero) to [bend left=20] (gamma);
	\draw[dotted, thick] (gamma) to [bend left=5] (zero);
	\draw[dashed, thick] (zero) to [bend left=5] (gamma);
	\draw[dotted, thick] (gamma) to [bend left=25] node[timeLabel,below] {$0$} (one.south east);
\end{tikzpicture}\label{FIG:Cl_i}
}
\caption{Gadgets used for variables and clauses in the reduction from 3-CNF-TQBF to \GHyDCC.}\label{fig.gadget}
\end{figure}
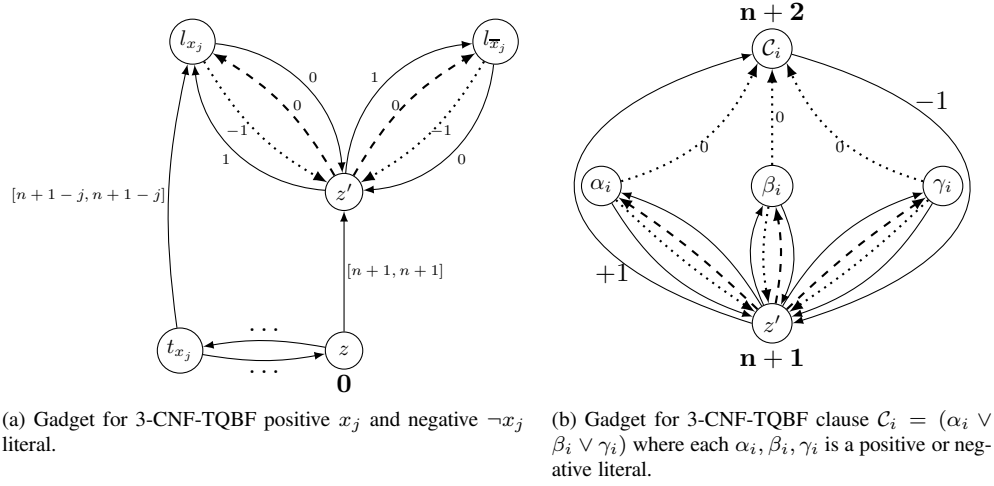

Finally, we describe the gadget for the clauses:
for each $i\in [m]$, the CHyTN $\Gamma_{\varphi}$ contains
	a node $\C_i$ for each clause $\C_i=(\alpha_i\vee \beta_i\vee\gamma_i)$ of $\varphi$; also, each node $\C_i$ is connected by:

-- a multi-tail hyperarc with head in $\C_i$ and tails over the literals $\alpha_i, \beta_i, \gamma_i$
	occurring in $\C_i$ and having weight $0$ and label $\lambda$, \ie by a multi-tail LTHC:
\[ A^c(\alpha_i, \beta_i, \gamma_i) \triangleq  \Big(\{\alpha_i, \beta_i, \gamma_i\}, \C_i,
		\langle w(\alpha_i), w(\beta_i), w(\gamma)_i \rangle=\langle 0, 0, 0\rangle,
			\langle L(\alpha_i), L(\beta_i), L(\gamma)_i\rangle=\langle\lambda, \lambda, \lambda \rangle \Big),\]
	for some literals $\alpha_i, \beta_i, \gamma_i\in \{l_{x_j}, l_{\overline{x}_j}\mid 1\leq j\leq n\}$.

-- two standard and opposite LTHCs, $(z', C_i, 1, \lambda)$ and $(C_i, z', -1, \lambda)$, with node $z'$.

See \figref{FIG:Cl_i} for an illustration of the clauses' gadget; the dashed arrows form the multi-head LTHCs and the dotted arrows form the multi-tail LTHCs.
Every node of $\Gamma_{\varphi}$ has an empty label, \ie $L(v)=\lambda$ for every $v\in V$.
The rationale of the clauses' gadget is that, for each $i$, at least one of the $\alpha_i, \beta_i, \gamma_i$
must occur at the same time instant of $C_i$ (i.e., at least one must occur at time $n+2$, because one of the literals must be true)

This concludes our description of $\Gamma_{\varphi}$.

More formally and succinctly, the CHyTN $\Gamma_{\varphi}=\langle V, \A, L, \Ord, \Ord{V}, P \rangle$ is defined as follows:
\begin{itemize}
	\item $P\triangleq \{ x_j \mid ``\forall x_j" \text{ appears in the prefix of } \varphi\}$;
	\item \begin{itemize} \item $V\triangleq \{z,z'\}\cup \{t_{x_j}\mid 1\leq j\leq n\}\cup \{p_{x_j}\mid x_j\in P\}\cup \\
				\cup \{l_{x_j} \mid 1\leq j \leq n\}\cup \{l_{\overline{x}_j}\mid 1\leq j \leq n\}\cup\{\C_i \mid 1\leq i\leq m\}$;
	\item ${\Ord}V\triangleq  \{p_{x_j}\mid x_j\in P\}$ and $\Ord(x_j)=p_{x_j}$ for every $x_j\in P$;
	\item $L(v)=\lambda$ for every $v\in V$;
		\end{itemize}
	\item $\displaystyle \A \triangleq  \bigcup_{j : ``\exists x_j" \in \varphi} \exists\text{-Qnt}_j \cup
		\bigcup_{j : ``\forall x_j" \in \varphi} \forall\text{-Qnt}_j \cup
		 \bigcup_{j=1}^n \text{Var}_j \cup \bigcup_{i=1}^m\text{Cla}_i \cup \big\{ (z, z', n+1, \lambda), (z', z, -n-1, \lambda) \big\}$, \\ where:
	\begin{itemize}
	\item $\exists\text{-Qnt}_j\triangleq \Big\{ (z,t_{x_j}, j+1, \lambda), (t_{x_j}, z, -j, \lambda) \Big\}$; \\
			This defines the existential quantifier gadget as depicted in \figref{FIG:timeline_existential};
	\item $\forall\text{-Qnt}_j\triangleq \Big\{  (z, p_{x_j}, j-1, \lambda), (p_{x_j}, z, -j+1, \lambda), \\
			(z, t_{x_j}, j+1, x_j), (t_{x_j}, z, -j-1, x_j), (z, t_{x_j}, j, \neg x_j), (t_{x_j}, z, -j, \neg x_j) \Big\}$; \\
			This defines the universal quantifier gadget as depicted in \figref{FIG:timeline_universal};
	\item $\text{Var}_j=\Big\{(z', l_{x_j}, 1, \lambda), (l_{x_j}, z', 0, \lambda),
					(z', l_{\overline{x}_j}, 1, \lambda), (l_{\overline{x}_j}, z', 0, \lambda), \\
			A^t_j \triangleq  %A^t(l_{x_j}, l_{\overline{x}_j})\triangleq
	\Big(\{l_{x_j}, l_{\overline{x}_j}\}, z', \langle w_{A^t_j}(l_{x_j}), w_{A^t_j}(l_{\overline{x}_j})\rangle = \langle -1, -1\rangle,
				\langle L_{A^t_j}(l_{x_j}), L_{A^t_j}(l_{\overline{x}_j})\rangle = \langle \lambda, \lambda\rangle \Big), \\
	 	A^h_j \triangleq  %A^h(l_{x_j}, l_{\overline{x}_j})\triangleq
	\Big(z', \{l_{x_j}, l_{\overline{x}_j}\}, \langle w_{A^h_j}(l_{x_j}), w_{A^h_j}(l_{\overline{x}_j}) \rangle = \langle 0, 0\rangle,
			\langle L_{A^h_j}(l_{x_j}), L_{A^h_j}(l_{\overline{x}_j})\rangle = \langle \lambda, \lambda\rangle \Big), \\
				(t_{x_j}, l_{x_j}, n+1-j, \lambda), (l_{x_j}, t_{x_j}, -n-1+j, \lambda) \Big\}$. \\
			This defines the variable gadget for $x_j$ as depicted in \figref{FIG:Var_j};
		\item $\text{Cla}_i=\Big\{(z', \C_j, 1), (\C_j, z', -1), \\
 			A^c_i\triangleq  \Big(\{\alpha_j, \beta_j, \gamma_j\}, \C_j,
				\langle w_{A^c_i}(\alpha_j), w_{A^c_i}(\beta_j), w_{A^c_i}(\gamma_j) \rangle = \langle 0,0,0\rangle, \\
				\langle L_{A^c_i}(\alpha_j), L_{A^c_i}(\beta_j), L_{A^c_i}(\gamma_j) \rangle =
											\langle \lambda,\lambda,\lambda \rangle \Big) \Big\}$.\\
	This defines the clause gadget for clause $\C_j = (\alpha_i \vee \beta_i \vee \gamma_i)$ as depicted in \figref{FIG:Cl_i}.
	\end{itemize}
\end{itemize}

Notice that $|V|\leq 1+4n+m=O(m+n)$ and $m_{\A}\leq 16n+5m=O(m+n)$; the transformation is thus linear.

\underline{Correctness.}
Let us show that $\varphi$ is \texttt{true} if and only if $\Gamma_{\varphi}$ is dynamically-consistent.

($\Rightarrow$)
Assume $\varphi$ is \texttt{true}, so Player-$\exists$ has a strategy to satisfy $\bigwedge_{i=1}^m (\alpha_i \vee \beta_i \vee \gamma_i)$
no matter how Player-$\forall$ decides to assign the universally quantified variables of $\varphi$.
Suppose that Player-$\exists$ and Player-$\forall$ alternate their choices by assigning a truth value to the variables of $\varphi$;
we can construct a dynamic and viable execution strategy $\sigma\in\S_{\Gamma_{\varphi}}$ for $\Gamma_{\varphi}$ by reflecting these choices, as follows.
The nodes $z$ and $z'$ are scheduled at time $0$ and $n+1$ (respectively) under all possible scenarios.
For each $j=1, \ldots, n$, the node $t_{x_j}$ is scheduled at time $j$ if $x_j$ is set to
\texttt{true} in $\varphi$, either by Player-$\exists$ or Player-$\forall$, otherwise at time $j+1$;
and, when $x_j$ is quantified universally in $\varphi$, the node $p_{x_j}$ is scheduled at time $j-1$ under all possible scenarios;
also, the node $l_{x_j}$ is scheduled at time $n+2$ if $x_j$ is set to \texttt{true} in $\varphi$,
either by Player-$\exists$ or Player-$\forall$, otherwise at time $n+1$;
symmetrically, $l_{\overline{x}_j}$ is scheduled at time $n+1$ if $x_j$ is \texttt{true} in $\varphi$, otherwise at time $n+2$.
Finally, for each $i=1, \ldots, m$, the node $\C_i$ is scheduled at time $n+2$ under all possible scenarios.
It is easy to check that all LTHCs of $\Gamma_\varphi$ are satisfied by $\sigma$ under all possible scenarios,
so $\sigma$ is viable for $\Gamma_\varphi$; moreover, since $\sigma$ reflects the alternating choices of Player-$\exists$ and Player-$\forall$, then $\sigma$ is also dynamic.
Therefore, $\Gamma_\varphi$ is dynamically-consistent.

($\Leftarrow$) Vice versa, assume that $\Gamma_{\varphi}$ is dynamically-consistent.
Let $\sigma\in\S_{\Gamma_{\varphi}}$ be a viable and dynamic execution strategy for $\Gamma_{\varphi}$.
Firstly, we argue that $\sigma$ is integer valued, \ie that $[\sigma(s)]_v\in \Z$ for every $v\in V$ and $s\in\Sigma_{\Gamma_{\varphi}}$.
Indeed, since $\sigma$ is viable, it is easy to check that the scheduling time of $z$, $z'$, $\C_i$ (for every $i=1, \ldots, m$)
and $p_{x_j}$ (for every universally quantified variable $x_j$ in $\varphi$) is forced to be $0$, $n+1$, $n+2$ and $j-1$ (respectively);
also, for each universally quantified variable $x_j$ in $\varphi$, the scheduling time of $p_{x_j}$ is forced to be $j-1$, and
that of $t_{x_j}$ is forced to be either $j$ or $j+1$ according to whether $x_j$ is \texttt{true} or $\texttt{false}$ in the current scenario.
Still, for each existentially quantified variable $x_j$ in $\varphi$,
the two standard LTHCs $(z,t_{x_j}, j+1, \lambda)$ and $(t_{x_j}, z, -j, \lambda)$ allow $t_{x_j}$ to be scheduled anywhere within $[j,j+1]$,
\ie even at non-integer values. However, on one side, the scheduling time of $l_{x_j}$ is forced to be that of $t_{x_j}$ \emph{plus} $n+1-j$,
on the other side, $l_{x_j}$ must be scheduled either at time $n+1$ or $n+2$ because of the multi-head $A^h(l_{x_j},
l_{\overline{x}_j})$ and multi-tail $A^t(l_{x_j}, l_{\overline{x}_j})$ LTHCs (respectively).
Thus, for $\sigma$ to be viable, $t_{x_j}$ must be scheduled either at time $j$ or $j+1$. Therefore, $\sigma$ is integer valued.
Now, suppose to execute $\sigma$ step-by-step over the integer line;
we can construct a strategy for Player-$\exists$ by
reflecting the integer choices that the Planner makes to schedule the nodes of $\Gamma_\varphi$, as follows.
For each existentially quantified variable $x_j$ in $\varphi$,
Player-$\exists$ sets $x_j$ to $\texttt{true}$ if the Planner schedules $t_{x_j}$ at time $j+1$
(\ie if $l_{x_j}$ is scheduled at time $n+2$, and $l_{\overline{x}_j}$ at time $n+1$),
and to $\texttt{false}$ otherwise (\ie if $t_{x_j}$ is at time $j$, $l_{x_j}$ at time $n+1$ and $l_{\overline{x}_j}$ at time $n+2$).
Then, since $\sigma$ is viable, for each clause $\C_i$ of $\varphi$,
at least one of the literals $\alpha_i, \beta_i, \gamma_i$ must be $\texttt{true}$, thanks to the multi-tail LTHC $A^c(\alpha_i, \beta_i, \gamma_i)$;
and since $\sigma$ is also dynamic, then Player-$\exists$ wins, so $\varphi$ is $\texttt{true}$.

 To conclude, notice that any LTHC $A\in\A$ of $\Gamma_{\varphi}$ has weights $w_A(\cdot)\in\{-1, 0, 1\}$ and size $|A|\leq 3$.
Since any hyperarc with three heads (tails) can be replaced by two hyperarcs each having at most two heads (tails),
then \GHyDCC remains \PSPACE-hard even if $w_A(\cdot)\in\{-1, 0, 1\}$ and $|A|\leq 2$ for every multi-tail/head LTHC $A\in \A$.
Also notice that $w_a(\cdot) \in [-n-1, n+1]\cap\Z$ and $\ell_a\in \{p, \neg p\mid p\in P\}\cup\{\lambda\}$
holds for every weight $w_a$ and label $\ell_a$ appearing in any \emph{standard} LTHC $a\in\A$. This concludes the proof.
\end{proof}

Theorem~\ref{Teo:pspacehardness} motivates the study of consistency problems on CHyTNs having either
only multi-head or only multi-tail hyperarcs. Since we are interested in dynamic consistency, where time moves only forward of course,
and the execution strategy depends only on past observations, from now on we shall consider only multi-head CHyTNs.

\subsection{$\epsilon$-dynamic consistency}
In CHyTN{s}, decisions about the precise timing of actions are postponed until execution time,
when information gathered from the execution of the observation events can be taken into account.
However, the Planner is allowed to factor in an observation, and modify its strategy in response to it,
only strictly after the observation has been made (whence the strict inequality in Definition~\ref{def:scenario_history}).
Notice that this definition does not take into account the actual reaction time, which, in most applications, is non-negligible.
In order to deliver algorithms that can also deal with the \emph{reaction time} $\epsilon$ of the Planner
we now introduce $\epsilon$-dynamic consistency, a refined notion of dynamic consistency.
The intuition underlying Definition~\ref{def:epsilonconsistency} is that to model a specific kind of disjunctive constraint:
given a small real number $\epsilon>0$, for any two scenarios $s_1, s_2\in \Sigma_P$ and any event $u\in V^+_{s_1, s_2}$, the scheduling time
of $u$ under $s_1$ must be greater or equal to either that of $u$ under $s_2$ or that of $v$ under $s_2$ plus $\epsilon$
for at least one $v\in\Delta(s_1; s_2)$. Let us remind the fact that, from now on, our CHyTNs admit only multi-head hyperarcs.
The definition of $\epsilon$-dynamic consistency follows below.

\begin{Def}[$\epsilon$-dynamic consistency]\label{def:epsilonconsistency}
Given any CHyTN $\langle V, \A, L, \Ord, \Ord{V}, P \rangle$ and any real number $\epsilon\in (0, \infty)$,
an execution strategy $\sigma\in\S_{\Gamma}$ is \emph{$\epsilon$-dynamic} if it satisfies all the $H_\epsilon\text{-constraints}$,
namely, for any two scenarios $s_1, s_2\in \Sigma_P$ and any event $u\in V^+_{s_1, s_2}$,
the execution strategy $\sigma$ satisfies the following constraint $H_{\epsilon}(s_1;s_2;u)$:
\[
	[\sigma(s_1)]_u \geq \min\Big(\big\{[\sigma(s_2)]_u\big\}\cup\big\{[\sigma(s_1)]_v + \epsilon\mid v\in\Delta(s_1; s_2)\big\}\Big)
\]
We say that a CHyTN $\Gamma$ is \emph{$\epsilon$-dynamically-consistent} if it admits $\sigma\in\S_{\Gamma}$ which is both viable and $\epsilon$-dynamic.

The problem of checking whether a given CHyTN is $\epsilon$-dynamically-consistent is named \emph{\eHyDCC}.
\end{Def}

It follows directly from Definition~\ref{def:epsilonconsistency} that,
whenever $\sigma\in S_{\Gamma}$ satisfies some $H_\epsilon(s_1;s_2;u)$, then $\sigma$ satisfies $H_{\epsilon'}(s_1;s_2;u)$
for every $\epsilon’\in(0, \epsilon]$ as well. This proves the following lemma.
\begin{Lem}
Let $\Gamma$ be a CHyTN.
If $\Gamma$ is $\epsilon$-dynamically-consistent for some real $\epsilon>0$,
then $\Gamma$ is $\epsilon'$-dynamically-consistent for every $\epsilon'\in (0, \epsilon]$.
\end{Lem}
Given any dynamically-consistent CHyTN,
we may ask for the maximum reaction time $\epsilon$ of the Planner beyond which the network is no longer dynamically-consistent.
\begin{Def}[Reaction time $\hat{\epsilon}$]\label{def:reaction_time}
Let $\Gamma$ be a CHyTN.
Let $\hat{\epsilon}\triangleq\hat{\epsilon}(\Gamma)$ be the least upper bound of the set of all real numbers $\epsilon>0$
such that $\Gamma$ is $\epsilon$-dynamically-consistent, \ie
\[
	\hat{\epsilon}\triangleq\hat{\epsilon}(\Gamma)\triangleq\sup\{\epsilon>0\mid \Gamma\text{ is } \epsilon\text{-dynamically-consistent}\}.
\]
\end{Def}
Let us consider the (affinely) extended real numbers $\overline{\RR}\triangleq \RR\cup\{-\infty, \infty\}$,
where \emph{every} subset $S$ of $\overline{\RR}$ has an infimum and a supremum.
Particularly, recall that $\sup\emptyset=-\infty$ and, if $S$ is unbounded above, then $\sup S=\infty$.

If $\Gamma$ is dynamically-consistent,
then $\hat{\epsilon}(\Gamma)$ exists and $\hat{\epsilon}(\Gamma)\neq -\infty$
(\ie the set on which we have taken the supremum in Definition~\ref{def:reaction_time} is non-empty),
as it is now proved in Lemma~\ref{lem:dynamic_impl_epsilon}.
\begin{Lem}\label{lem:dynamic_impl_epsilon}
Let $\sigma$ be a dynamic execution strategy for the CHyTN $\Gamma$.
Then, there exists a sufficiently small real number $\epsilon\in (0, \infty)$ such that $\sigma$ is $\epsilon$-dynamic.
\end{Lem}
\begin{proof}
Let $s_1, s_2\in \Sigma_P$ be two scenarios and let us consider any event $u\in V^+_{s_1, s_2}$.
Since $\sigma$ is dynamic, then by Lemma~\ref{lem:dynamicimplequality} the following implication necessarily holds:
\[\Big( \bigwedge_{v\in\Delta(s_1; s_2)} [\sigma(s_1)]_u \leq [\sigma(s_1)]_v \Big) \Rightarrow [\sigma(s_1)]_u \geq [\sigma(s_2)]_u\;\; \tag{*}\]
Notice that, \wrt Lemma~\ref{lem:dynamicimplequality}, we have relaxed the equality $[\sigma(s_1)]_u = [\sigma(s_2)]_u$ in the implicand of (L\ref{lem:dynamicimplequality})
by introducing the inequality $[\sigma(s_1)]_u \geq [\sigma(s_2)]_u$.
At this point, we convert ($*$) from implicative to disjunctive form,
first by applying the rule of material
implication\footnote{The rule of material implication: $\models p\Rightarrow q \iff \neg p \vee q$.},
and then De Morgan's law\footnote{De Morgan's law: $\models \neg (p\wedge q)\iff \neg p\vee \neg q$.},
resulting in the following equivalent expression:
\[ \Big( [\sigma(s_1)]_u \geq [\sigma(s_2)]_u \Big) \vee \Big( \bigvee_{v\in\Delta(s_1; s_2)} [\sigma(s_1)]_u > [\sigma(s_1)]_v \Big) \tag{**} \]
Then, we argue that there exists a real number $\epsilon\in (0, \infty)$
such that the following disjunction holds as well:
\[ \Big( [\sigma(s_1)]_u \geq [\sigma(s_2)]_u \Big) \vee \Big( \bigvee_{v\in\Delta(s_1; s_2)} [\sigma(s_1)]_u \geq [\sigma(s_1)]_v +\epsilon \Big). \]
In fact, since the disjunction ($**$) necessarily holds, then
one can pick the following real number $\epsilon>0$:
\[ \epsilon\triangleq \min_{\langle s_1, s_2, u\rangle\in \Sigma_P\times\Sigma_P\times V^+_{s_1, s_2}} \epsilon(s_1;s_2;u), \]
where the values $\epsilon(s_1;s_2;u)\in (0, \infty)$ are defined as follows,
	for every $\langle s_1, s_2, u\rangle\in \Sigma_P\times\Sigma_P\times V^+_{s_1, s_2}$:
\[
\epsilon(s_1;s_2;u) \triangleq
\left\{
\begin{array}{ll}
1, \text{ if } [\sigma(s_1)]_u \geq [\sigma(s_2)]_u ;\\
\min\Big\{[\sigma(s_1)]_u - [\sigma(s_1)]_v \mid v\in\Delta(s_1; s_2),\; [\sigma(s_1)]_u > [\sigma(s_1)]_v\Big\}, \text{ otherwise}.
\end{array}
\right.
\]
This implies that $\sigma$ satisfies every $H_{\epsilon}$-constraint of $\Gamma$,
thus $\sigma$ is $\epsilon$-dynamic.
\end{proof}
Next, we prove a converse formulation of Lemma~\ref{lem:dynamic_impl_epsilon}.
\begin{Lem}\label{lem:epsilon_impl_dynamic}
Let $\sigma$ be an $\epsilon$-dynamic execution strategy for a CHyTN $\Gamma$, for some real number $\epsilon\in (0, \infty)$.

Then, $\sigma$ is dynamic.
\end{Lem}
\begin{proof}
For the sake of contradiction, let us suppose that $\sigma$ is not dynamic.
Let $F$ be the set of all the triplets $\langle u,s_1,s_2\rangle\in V^+_{s_1, s_2}\times\Sigma_P\times\Sigma_P$,
for which the implication (L\ref{lem:dynamicimplequality}) given in Lemma~\ref{lem:dynamicimplequality} does not hold.
Notice, $F\neq\emptyset$; indeed, since $\sigma$ is not dynamic, by Lemma~\ref{lem:dynamicimplequality}
there exists at least one $\langle u,s_1,s_2\rangle$ for which (L\ref{lem:dynamicimplequality}) doesn't hold.
So, it holds that $\langle u,s_1,s_2\rangle\in F$ if and only if the following two properties hold:
\begin{enumerate}
\item $[\sigma(s_1)]_u \leq [\sigma(s_1)]_v$, for every $v\in\Delta(s_1; s_2)$;
\item $[\sigma(s_1)]_u \neq [\sigma(s_2)]_u$.
\end{enumerate}
Let $\langle\hat{u}, \hat{s_1}\rangle$ be an event whose scheduling time $[\sigma(\hat{s_1})]_{\hat{u}}$ is minimum and for which ($1$) and ($2$) hold,
namely, let: \[ \langle\hat{u}, \hat{s_1}\rangle\triangleq \arg\min\Big\{[\sigma(s_1)]_u\mid \exists{s_2}\, \langle u,s_1,s_2\rangle\in F\Big\}.\]
Since $\langle \hat{u}, \hat{s_1}\rangle$ is minimum in $[\sigma(\hat{s_1})]_{\hat{u}}$, then
$[\sigma(\hat{s_1})]_{\hat{u}} \leq [\sigma(s_2)]_{\hat{u}}$ for every $s_2\in\Sigma_P$ such that $\langle\hat{u}, \hat{s_1}, s_2\rangle\in F$;
moreover, since $\langle\hat{u}, \hat{s_1}, s_2\rangle\in F$, then $[\sigma(\hat{s_1})]_{\hat{u}} \neq [\sigma(s_2)]_{\hat{u}}$ holds by (2),
so that $[\sigma(\hat{s_1})]_{\hat{u}} < [\sigma(s_2)]_{\hat{u}}$.
At this point, recall that $\sigma$ is $\epsilon$-dynamic by hypothesis, hence $[\sigma(\hat{s_1})]_{\hat{u}} < [\sigma(s_2)]_{\hat{u}}$ implies that
there exists $v\in\Delta(\hat{s_1}; s_2)$ such that:
\[ [\sigma(\hat{s_1})]_{\hat{u}}\geq [\sigma(\hat{s_1})]_v+\epsilon>[\sigma(\hat{s_1})]_v,\]
but this inequality contradicts item ($1$) above. Indeed, $F=\emptyset$ and $\sigma$ is thus dynamic.
\end{proof}
In Section~\ref{sect:epsilon}, the following theorem is proved.
\begin{Thm}\label{thm:boundexp}
For any dynamically-consistent CHyTN $\Gamma$, where $V$ is the set of events and $\Sigma_P$ is the set of scenarios,
	it holds that $\hat{\epsilon}(\Gamma) \geq |\Sigma_P|^{-1}|V|^{-1}$.
\end{Thm}

\begin{figure}[!htb]
\centering
\subfloat[The CSTN $\Gamma_{\frac{1}{2}}$.]{
\begin{tikzpicture}[arrows=->,scale=0.8, node distance=1.5 and 1.5]\label{FIG:fractional-CSTN-CSTN}
	\node[node, label={above:$Y_1?$}] (Y1) {$Y_1$};
 	\node[node, above left = of Y1, xshift=0ex, label={left:$\textbf{0}$}, label={above:\small $X_1?$}] (X1) {$X_1$};
	\node[node, above right = of Y1, xshift=0ex] (Z1) {$Z_1$};
	%arcs
	\draw[] (X1) to [] node[xshift=0ex, yshift=0ex,above] {\footnotesize $1, X_1 Y_1$} (Z1);
	\draw[] (X1) to [] node[xshift=-3.5ex, yshift=0ex,below] {\footnotesize $[2,2], \neg X_1$} (Y1);
	\draw[] (Y1) to [] node[xshift=3.5ex, yshift=0ex,below] {\footnotesize $[2,2], \neg Y_1$} (Z1);
	\end{tikzpicture}
	} \quad
	\subfloat[A viable and $\epsilon$-dynamic execution strategy for $\Gamma_{\frac{1}{2}}$.]{
	\begin{tikzpicture}[scale=0.8, level distance=50pt, sibling distance=1pt]\label{FIG:fractional-CSTN-strategy}
	\Tree [. \framebox{$[\sigma_1(s)]_{X_1}=0$}
		 \edge node[left,xshift=-1ex]{$s(X_1)=\top$};
		 [. \framebox{$[\sigma_1(s)]_{Y_1}=\frac{1}{2}$}
			\edge node[left, xshift=-1ex]{$s(Y_1)=\top$};
			[. \framebox{$[\sigma_1(s)]_{Z_1}=1$}
			]
			\edge node[right, xshift=1ex]{$s(Y_1)=\bot$};
			[. \framebox{$[\sigma_1(s)]_{Z_1}=\frac{5}{2}$}
			]
		 ]
		 \edge node[right,xshift=1ex]{$s(X_1)=\bot$};
		 [. \framebox{$[\sigma_1(s)]_{Y_1}=2$}
			\edge node[xshift=1.3ex]{$s(Y_1)=\top$ or $s(Y_1)=\bot$};
			[. \framebox{$[\sigma_1(s)]_{Z_1}=4$}
			]
		 ]
	]
\end{tikzpicture} }
\caption{A dynamically-consistent CSTN whose viable and dynamic execution strategies are fractional.}
\label{FIG:fractional-CSTN}
\end{figure}
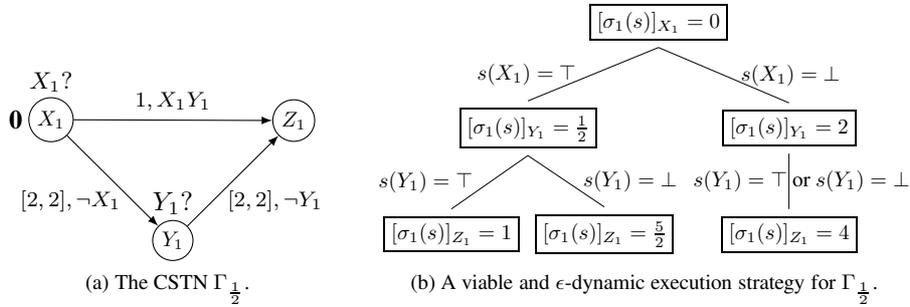

Notice that one really needs to consider rational values for $\hat{\epsilon}$, as it is shown in the following example.
\begin{Ex}\label{ex:CSTN-fractional}
Consider the \CSTN $\Gamma_{\frac{1}{2}}$ shown in \figref{FIG:fractional-CSTN-CSTN}.
The Planner needs to schedule and to observe $X_1$ at time $0$ under all possible scenarios.
But it is not viable to schedule $Y_1$ or $Z_1$ at time $0$, because $X_1$ and $Y_1$ may turn out to be $\bot$;
so $Y_1$ and $Z_1$ both need to be scheduled strictly after $0$. Next, assume that $X_1$ turns out to be $\top$ at time $0$.
Then, it is not viable to schedule $Y_1$ at time $1$,
because $Z_1$ needs to be scheduled within time $1$ if $Y_1$ is $\top$ and strictly after otherwise,
and the Planner can't react instantaneously to the observation made at $Y_1$.
Thus, if $X_1$ is $\top$ at time $0$, then $Y_1$ needs to be scheduled at time $t\in (0,1)$, \eg $t=\frac{1}{2}$.
The corresponding execution strategy is shown in \figref{FIG:fractional-CSTN-strategy}.
\end{Ex}
Also notice that, in Definition~\ref{def:consistency}, dynamic consistency was defined by strict-inequality and equality constraints.
However, by Theorem~\ref{thm:epsilonconsistency}, dynamic consistency can also be defined in terms
of $H_{\epsilon}$-constraints only (\ie no strict-inequalities are required).
\begin{Thm}\label{thm:epsilonconsistency}
Let $\Gamma$ be a CHyTN. Let $\epsilon\triangleq |\Sigma_P|^{-1}|V|^{-1}$.
Then, $\Gamma$ is dynamically-consistent if and only if $\Gamma$ is $\epsilon$-dynamically-consistent.
\end{Thm}

By Theorem~\ref{thm:epsilonconsistency}, any algorithm for checking $\epsilon$-dynamic consistency can be used to check dynamic consistency.

\subsection{A (pseudo) Singly-Exponential Time Algorithm for \DCC and \HyDCC}\label{subsect:Algo}

In this section, we present a (pseudo) singly-exponential time algorithm for solving \DCC and \HyDCC,
 also producing a dynamic execution strategy whenever the input CHyTN is dynamically-consistent.

The main result of this paper is summarized in the following theorem, which is proven in the next Section~\ref{sect:epsilon}.
\begin{Thm}\label{thm:mainresult}
The following two algorithmic results hold for CHyTNs.
\begin{enumerate}
\item There exists an $O \big(|\Sigma_P|^{2}|\A|m_{\A} + |\Sigma_P|^3|V||\A||P| + |\Sigma_P|^{3}|V|m_{\A} + |\Sigma_P|^4|V|^2|P|\big)WD$
	time deterministic algorithm for deciding \eHyDCC on input $\langle\Gamma, \epsilon\rangle$,
	for any CHyTN $\Gamma=\langle V, \A, L, \Ord, \Ord{V}, P \rangle$ and any rational number $\epsilon=N/D$ where $N,D\in \N_+$.
Particularly, given any $\epsilon$-dynamically-consistent CHyTN $\Gamma$,
	the algorithm returns as output a viable and $\epsilon$-dynamic execution strategy $\sigma\in \S_\Gamma$.
\item There exists an $O\big(|\Sigma_P|^{3}|V||\A|m_{\A} + |\Sigma_P|^4|V|^2|\A||P| + |\Sigma_P|^{4}|V|^2m_{\A} + |\Sigma_P|^5|V|^3|P|\big)W$
	time deterministic algorithm for checking \HyDCC on any input $\Gamma=\langle V, \A, L, \Ord, \Ord{V}, P \rangle$.
Particularly, given any dynamically-consistent CHyTN $\Gamma$,
	it returns as output a viable and dynamic execution strategy $\sigma\in \S_\Gamma$.
\end{enumerate}
Here, $W\triangleq \max_{a\in A} |w_a|$.
\end{Thm}

Since every \CSTN is also a CHyTN, Theorem~\ref{thm:mainresult} holds for \CSTN{s} as well.

We now present the reduction from \HyDCC to \HTNC.
Again, since any \CSTN is a CHyTN, the same argument reduces \DCC to \HTNC.
Firstly, we argue that any CHyTN can be viewed as a succinct representation
which can be expanded into an exponentially sized \HTN.

The \emph{Expansion} of \CSTN{s} is introduced below.

\begin{Def}[Expansion $\langle V^{\text{Ex}}_{\Gamma}, \Lambda^{\text{Ex}}_{\Gamma}\rangle$]\label{def:expansion}
Let $\Gamma=\langle V, \A, L, \Ord, {\Ord}V, P \rangle$ be a CHyTN.
Consider the family of distinct and disjoint HyTN{s} $\langle V_s, \A_s\rangle$,
one for each scenario $s\in\Sigma_P$, which is defined as follows (where $v_s\triangleq (v,s)$ for every $v\in V$ and $s\in\Sigma_P$):
\begin{align*}
	V_s  \triangleq \{ & v_s \mid v\in V^+_s\},  \\
	\A_s \triangleq \Big\{ & \big( t_s, \underbrace{ \{h^{(1)}_s, \ldots, h^{(k)}_s\}}_{\text{heads labeled with } s},
			\underbrace{ \langle w(h^{(1)}_s), \ldots, w(h^{(k)}_s) \rangle}_{\text{corresponding weights}} \big)
				\,\Big|\,  \\
				&  \big( \underbrace{t}_{\text{tail}}, \underbrace{ \{h^{(1)}, \ldots, h^{(k)}\}}_{\text{heads}},
							\underbrace{ \langle w(h^{(1)}), \ldots, w(h^{(k)}) \rangle}_{\text{corresponding weights}} \big) \in \A^+_s \Big\} .
\end{align*}
(Of course, in the above notation, $k=1$ when $\Gamma$ is a \CSTN, whereas $k\in\N_+$ when $\Gamma$ is a CHyTN.)

Next, we define the \emph{expansion} $\langle V^{\text{Ex}}_{\Gamma}, \Lambda^{\text{Ex}}_{\Gamma}\rangle$ of $\Gamma$ as follows:
\[\langle V^{\text{Ex}}_{\Gamma}, \Lambda^{\text{Ex}}_{\Gamma}\rangle\triangleq
	\Big(\bigcup_{s\in\Sigma_P}V_s, \bigcup_{s\in\Sigma_P} \A_s\Big).\]
\end{Def}

Notice that $V_{s_1}\cap V_{s_2}=\emptyset$ whenever $s_1\neq s_2$ and that
$\langle V^{\text{Ex}}_{\Gamma}, \Lambda^{\text{Ex}}_{\Gamma}\rangle$ is an \STN/\HTN with at most $|V^{\text{Ex}}_{\Gamma}|\leq|\Sigma_P|\cdot |V|$ nodes
and size at most $|\Lambda^{\text{Ex}}_{\Gamma}|\leq |\Sigma_P|\cdot |\A|$.

We now show that the expansion of a CHyTN can be enriched with some (extra) multi-head hyperarcs
in order to model $\epsilon$-dynamic consistency, by means of a particular \HTN which is denoted by $\H_{\epsilon}(\Gamma)$.
\begin{Def}[\HTN $\H_{\epsilon}(\Gamma)$]\label{def:Hepsilonzero}
Let $\Gamma=\langle V, \A, L, \Ord, {\Ord}V, P \rangle$ be a CHyTN.
Given any real number $\epsilon\in(0, \infty)$, the \HTN $\H_{\epsilon}(\Gamma)$ is defined as follows:
\begin{itemize}
\item For every two scenarios $s_1, s_2\in\Sigma_P$ and for every event node $u\in V^+_{s_1, s_2}$,
define a hyperarc $\alpha\triangleq\alpha_{\epsilon}(s_1; s_2; u)$ as follows
(with the intention to model $H_{\epsilon}(s_1; s_2; u)$ from Def.~\ref{def:epsilonconsistency}):
\[\alpha_{\epsilon}(s_1; s_2; u)\triangleq \big( t_\alpha, H_\alpha, w_\alpha\big),
				\;\; \forall\, s_1, s_2\in\Sigma_P\text{ and } u\in V^+_{s_1, s_2}. \]
where:
\begin{itemize}
\item $t_\alpha\triangleq u_{s_1}$ is the tail of the (multi-head) hyperarc $\alpha_{\epsilon}(s_1;s_2;u)$;
\item $H_\alpha\triangleq \{u_{s_2}\}\cup \Delta(s_1; s_2)$ is the set of the heads of $\alpha_{\epsilon}(s_1;s_2;u)$;
\item $w_\alpha(u_{s_2})\triangleq 0$, and $w_\alpha(v)\triangleq - \epsilon$ for each $v\in\Delta(s_1; s_2)$.
\end{itemize}
\item Consider the expansion $\langle V^{\text{Ex}}_{\Gamma}, \Lambda^{\text{Ex}}_{\Gamma}\rangle$ of $\Gamma$.
Then, $\H_{\epsilon}(\Gamma)$ is defined as
$\H_{\epsilon}(\Gamma)\triangleq \big( V^{\text{Ex}}_{\Gamma}, \A_{H_{\epsilon}}\big)$,
where,
\[\A_{H_{\epsilon}}\triangleq \Lambda^{\text{Ex}}_{\Gamma} \cup \bigcup_{\substack{s_1,s_2\in\Sigma_P \\ u\in V^+_{s_1, s_2}}}\alpha_{\epsilon}(s_1;s_2;u).\]
\end{itemize}
\end{Def}

Notice that each $\alpha_{\epsilon}(s_1; s_2; u)$ has size $|\alpha_{\epsilon}(s_1; s_2; u)| = 1+\Delta(s_1;s_2)\leq 1+|P|$.

Here below, Algorithm~\ref{algo:pseudocode_construct_H} provides a pseudocode for constructing $\H_{\epsilon}(\Gamma)$.

\begin{algorithmenv}
\begin{algorithm}[H]\label{algo:pseudocode_construct_H}
\caption{$\texttt{construct\_}\H$$(\Gamma, \epsilon)$}
\KwIn{a CHyTN $\Gamma\triangleq\langle V, \A, L, \Ord, {\Ord}V, P \rangle$,	a rational number $\epsilon>0$}
\ForEach{($s\in\Sigma_P$)}{
	$V_s\leftarrow\{v_s \mid v\in V^+_{s}\}$; \tcp{ see Def.~\ref{def:expansion}}
	$A_s\leftarrow\{a_s\mid a\in A^+_s\}$; \tcp{ see Def.~\ref{def:expansion}}
}
$\displaystyle V^{\text{Ex}}_{\Gamma}\leftarrow \cup_{\substack{s\in\Sigma_P}}V_s$\;
$\displaystyle \Lambda^{\text{Ex}}_{\Gamma}\leftarrow \cup_{\substack{s\in\Sigma_P}} A_s$\;
\ForEach{($s_1, s_2\in\Sigma_P$ \texttt{\bf and} $u\in V^+_{s_1, s_2}$)}{
	$t_\alpha\leftarrow u_{s_1}$\;
	$H_\alpha\leftarrow \{u_{s_2}\}\cup \Delta(s_1; s_2)$\;
	$w_\alpha(u_{s_2})\leftarrow 0$\;
	\ForEach{$v\in\Delta(s_1; s_2)$}{
		$w_\alpha(v_{s_1})\leftarrow -\epsilon$\;
		}
		$\alpha_{\epsilon}(s_1; s_2; u)\leftarrow \big( t_\alpha, H_\alpha, w_\alpha \big)$\;
}
$\displaystyle\A_{\H_{\epsilon}}\leftarrow \Lambda^{\text{Ex}}_{\Gamma} \cup
			\bigcup_{\substack{s_1,s_2\in\Sigma_P \\ u\in V^+_{s_1, s_2}}}\alpha_{\epsilon}(s_1;s_2;u)$\;
$\H_{\epsilon}(\Gamma)\leftarrow \big( V^{\text{Ex}}_{\Gamma}, \A_{H_{\epsilon}} \big)$\;
\Return{$H_{\epsilon}(\Gamma)$;}
\end{algorithm}\smallskip
\caption{Constructing $\H_{\epsilon}(\Gamma)$.}
\end{algorithmenv}

\begin{Ex}
An excerpt of the \HTN $\H_{\epsilon}(\Gamma_0)$ corresponding to the \CSTN $\Gamma_0$ of \figref{FIG:cstn1} is depicted in \figref{FIG:cstn22htn};
here, two scenarios $s_1\triangleq p\wedge q$ and $s_4\triangleq \neg p\wedge \neg q$ are considered,
	on top we have ${\Gamma_0}^+_{s_1}$, whereas ${\Gamma_0}^+_{s_4}$ is below, finally,
		the corresponding hyperconstraints $H_\epsilon(s_1;s_4;u)$ and $H_\epsilon(s_4;s_1;u)$ are depicted as dashed hyperarcs.
	\end{Ex}

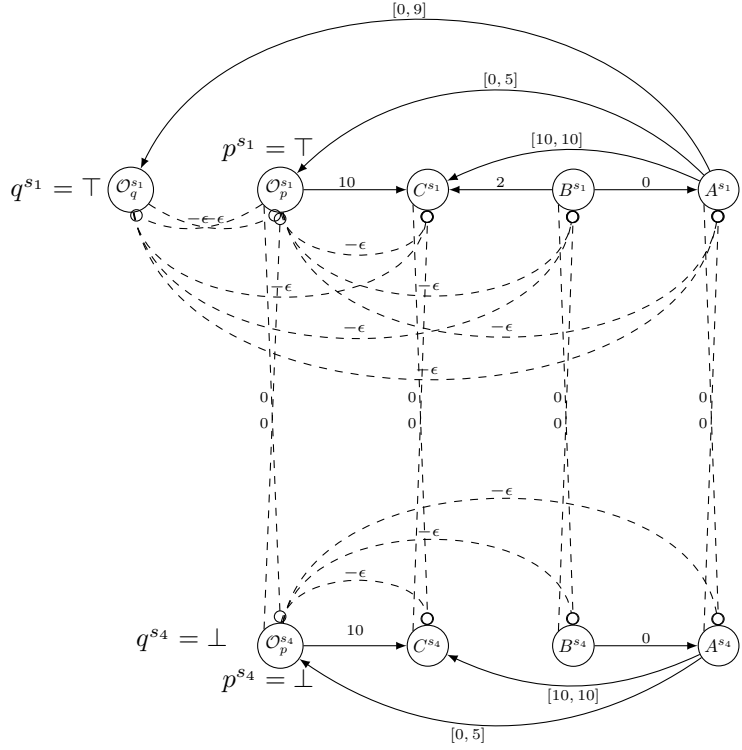
\begin{figure}[!htb]
\centering
\begin{tikzpicture}[arrows=->, scale=.68, node distance=8 and 2]
 	\node[node] (A') {$A^{s_1}$};
	\node[node,xshift=0ex,left=of A'] (B') {$B^{s_1}$};
	\node[node,xshift=0ex,left=of B'] (C') {$C^{s_1}$};
	\node[node,xshift=0ex,left=of C', label={above,xshift=-1ex:$p^{s_1}=\top$}] (P') {$\Ord_p^{s_1}$};
	\node[node,xshift=0ex,left=of P', label={left:$q^{s_1}=\top$}] (Q') {$\Ord_q^{s_1}$};
	\node[node, below=of P', label={below,xshift=-1ex,yshift=.5ex:$p^{s_4}=\bot$},
					label={left,xshift=-2ex,yshift=.5ex:$q^{s_4}=\bot$}] (P'''') {$\Ord_p^{s_4}$};
	\node[node, right=of P'''',yshift=0ex] (C'''') {$C^{s_4}$};
	\node[node, right=of C'''',yshift=0ex] (B'''') {$B^{s_4}$};
	\node[node, right=of B'''',yshift=0ex] (A'''') {$A^{s_4}$};

	%arcs s1
	\draw[] (A') to [bend right=25] node[xshift=-2ex, yshift=0ex,timeLabel,scale=1.2,above] {$[10,10]$} (C');
	\draw[] (B') to [] node[timeLabel,scale=1.2,above] {$2$} (C');
	\draw[] (B') to [] node[timeLabel,scale=1.2,above] {$0$} (A');
	\draw[] (A') to [bend right=45] node[xshift=0ex, yshift=0ex, timeLabel,scale=1.2,above] {$[0,5]$} (P');
	\draw[] (A') to [bend right=65] node[xshift=-2ex,yshift=0ex, timeLabel,scale=1.2,above] {$[0,9]$} (Q');
	\draw[] (P') to [] node[xshift=-1ex, timeLabel,scale=1.2,above] {$10$} (C');

	%arcs s4
	\draw[] (A'''') to [bend left=25] node[xshift=0ex, yshift=-.1ex,timeLabel,scale=1.2,below] {$[10,10]$} (C'''');
	\draw[] (B'''') to [] node[timeLabel,scale=1.2,above] {$0$} (A'''');
	\draw[] (A'''') to [bend left=35] node[xshift=-4ex, yshift=0ex, timeLabel,scale=1.2,below] {$[0,5]$} (P'''');
	\draw[] (P'''') to [] node[xshift=0ex,yshift=1ex, timeLabel,scale=1.2,above] {$10$} (C'''');

	%hyperarcs s4->s1
	\draw[>=o, multiHead] (A''''.north west) to [] node[xshift=-1ex, yshift=2ex,timeLabel,scale=1.2,above] {$0$} (A'.south);
	\draw[>=o, multiHead] (P') to [bend right=80] node[xshift=0ex, yshift=0ex,timeLabel,scale=1.3,above] {$-\epsilon$} (A'.south);
	\draw[>=o, multiHead] (Q') to [bend right=80] node[xshift=0ex, yshift=0ex,timeLabel,scale=1.3,above] {$-\epsilon$} (A'.south);
	\draw[>=o, multiHead] (B''''.north west) to [] node[xshift=-1ex, yshift=2ex,timeLabel,scale=1.2,above] {$0$} (B'.south);
	\draw[>=o, multiHead] (P') to [bend right=80] node[xshift=0ex, yshift=0ex,timeLabel,scale=1.3,above] {$-\epsilon$} (B'.south);
	\draw[>=o, multiHead] (Q') to [bend right=80] node[xshift=0ex, yshift=0ex,timeLabel,scale=1.3,above] {$-\epsilon$} (B'.south);
	\draw[>=o, multiHead] (C''''.north west) to [] node[xshift=-1ex, yshift=2ex,timeLabel,scale=1.2,above] {$0$} (C'.south);
	\draw[>=o, multiHead] (P') to [bend right=80] node[xshift=0ex, yshift=0ex,timeLabel,scale=1.3,above] {$-\epsilon$} (C'.south);
	\draw[>=o, multiHead] (Q') to [bend right=80] node[xshift=0ex, yshift=0ex,timeLabel,scale=1.3,above] {$-\epsilon$} (C'.south);
	\draw[>=o, multiHead] (P''''.north west) to [] node[xshift=-1ex, yshift=2ex,timeLabel,scale=1.2,above] {$0$} (P'.south);
	\draw[>=o, multiHead] (Q') to [bend right=30] node[xshift=0ex, yshift=0ex,timeLabel,scale=1.3,above] {$-\epsilon$} (P'.south);
	\draw[>=o, multiHead] (P') to [bend left=30] node[xshift=0ex, yshift=0ex,timeLabel,scale=1.3,above] {$-\epsilon$} (Q'.south);

	%hyperarcs s1->s4
	\draw[>=o, multiHead] (A'.south west) to [] node[xshift=-1ex, yshift=-2ex,timeLabel,scale=1.2,above] {$0$} (A''''.north);
	\draw[>=o, multiHead] (P'''') to [bend left=80] node[xshift=0ex, yshift=0ex,timeLabel,scale=1.3,above] {$-\epsilon$} (A''''.north);
	\draw[>=o, multiHead] (B'.south west) to [] node[xshift=-1ex, yshift=-2ex,timeLabel,scale=1.2,above] {$0$} (B''''.north);
	\draw[>=o, multiHead] (P'''') to [bend left=80] node[xshift=0ex, yshift=0ex,timeLabel,scale=1.3,above] {$-\epsilon$} (B''''.north);
	\draw[>=o, multiHead] (C'.south west) to [] node[xshift=-1ex, yshift=-2ex,timeLabel,scale=1.2,above] {$0$} (C''''.north);
	\draw[>=o, multiHead] (P'''') to [bend left=80] node[xshift=0ex, yshift=0ex,timeLabel,scale=1.3,above] {$-\epsilon$} (C''''.north);
	\draw[>=o, multiHead] (P'.south west) to [] node[xshift=-1ex, yshift=-2ex,timeLabel,scale=1.2,above] {$0$} (P''''.north);
\end{tikzpicture}
\caption{An excerpt of the \HTN $\H_{\epsilon}(\Gamma_0)$ corresponding to the \CSTN $\Gamma_0$ of Fig.~\ref{FIG:cstn1},
	in which two scenarios, $s_1$ and $s_4$, are considered.}
\label{FIG:cstn22htn}
\end{figure}

The following establishes the connection between dynamic consistency of CHyTN{s} and consistency of \HTN{s}.
\begin{Thm}\label{thm:mainreduction}
Given any CHyTN $\Gamma=\langle V, \A, L, \Ord, {\Ord}V, P \rangle$,
there exists a sufficiently small real number $\epsilon\in (0, \infty)$ such that
the CHyTN $\Gamma$ is dynamically-consistent if and only if the \HTN $\H_{\epsilon}(\Gamma)$ is consistent.

Moreover, the \HTN $\H_{\epsilon}(\Gamma)$ has at most
$|V_{\H_{\epsilon}}|\leq |\Sigma_P|\cdot |V|$ nodes, $|\A_{\H_{\epsilon}}|=O(|\Sigma_P||\A| + |\Sigma_P|^2|V|)$ hyperarcs,
and it has size $m_{\A_{\H_{\epsilon}}}=O(|\Sigma_P| m_{\A} + |\Sigma_P|^2|V|\, |P|)$.
\end{Thm}
\begin{proof}
For any real number $\epsilon\in (0,\infty)$,
let $\H_{\epsilon}(\Gamma)=\langle V^{\text{Ex}}_{\Gamma}, \A_{H_{\epsilon}} \rangle$ be the \HTN of Definition~\ref{def:Hepsilonzero}.

(1) By Definitions~\ref{def:expansion}~and~\ref{def:Hepsilonzero},
	$|V_{\H_{\epsilon}}|=|V^{\text{Ex}}_\Gamma|\leq |\Sigma_P|\cdot |V|$; also,
		$|\A_{\H_{\epsilon}}| = |\Lambda^{\text{Ex}}_{\Gamma}| +
			\big| \bigcup_{s_1,s_2\in\Sigma_P; u\in V^+_{s_1, s_2}}\alpha_{\epsilon}(s_1;s_2;u)\big|=
				O(|\Sigma_P||\A| + |\Sigma_P|^2|V|)$,
					and since $\alpha_{\epsilon}(s_1;s_2;u)$ has at most $P$ heads,
						then $m_{\A_{\H_{\epsilon}}}=O(|\Sigma_P| m_{\A} + |\Sigma_P|^2|V|\, |P|)$.

(2) We claim that, for any $\epsilon>0$, $\H_{\epsilon}(\Gamma)$ is consistent
	if and only if $\Gamma$ is $\epsilon$-dynamically-consistent.

($\Rightarrow$) Given any feasible schedule $\phi:V_\Gamma^{\text{Ex}}\rightarrow\RR$ for the \HTN $\H_{\epsilon}(\Gamma)$,
let $\sigma_\phi(s)\in\S_{\Gamma}$ be the execution strategy
defined as follows:
\[
	[\sigma_\phi(s)]_v\triangleq \phi(v_s),\text{ for every } v_s\in V_{\Gamma}^{\text{E}},\text{ where } v\in V\text{ and }s\in\Sigma_P.
\]
Notice that each hyperarc $\alpha_{\epsilon}(s_1; s_2; u)$ is satisfied by $\phi$
if and only if the corresponding $H_{\epsilon}$-constraint $H_{\epsilon}(s_1; s_2; u)$ is satisfied by $\sigma_{\phi}$;
moreover, recall that $\Lambda^{\text{Ex}}_{\Gamma}\subseteq \A_{H_{\epsilon}}$,
and that $\Lambda^{\text{Ex}}_{\Gamma}$ contains all the \emph{original} standard/hyper
difference constraints of $\Gamma$ (\ie those induced by $\A$, by means of Def.~\ref{def:expansion}).
At this point, since $\phi$ is feasible for the \HTN $\H_{\epsilon}(\Gamma)$,
then $\sigma_{\phi}$ must be viable and $\epsilon$-dynamic for $\Gamma$ (because it satisfies all the required constraints).

Therefore, $\Gamma$ is $\epsilon$-dynamically-consistent.

($\Leftarrow$) Given any viable and $\epsilon$-dynamic execution strategy $\sigma\in\S_{\Gamma}$, for some real number $\epsilon\in (0,\infty)$,
let $\phi_{\sigma}:V^{\text{Ex}}_{\Gamma}\rightarrow\RR$ be the schedule of the \HTN $\H_{\epsilon}(\Gamma)$
defined as follows:
\[
	\phi_{\sigma}(v_s)\triangleq [\sigma(s)]_v\text{ for every } v_s\in V^{\text{Ex}}_{\Gamma},\text{ where } v\in V\text{ and } s\in\Sigma_P.
\]
Also in this case, we have that $\Lambda^{\text{Ex}}_{\Gamma}\subseteq \A_{H_{\epsilon}}$,
and a moment's reflection reveals that each hyperarc $\alpha_{\epsilon}(s_1; s_2; u)$ is satisfied by $\phi_{\sigma}$
if and only if  $H_{\epsilon}(s_1; s_2; u)$ is satisfied by $\sigma$.
At this point, since $\sigma$ is viable and $\epsilon$-dynamic for the CHyTN $\Gamma$,
then $\phi_{\sigma}$ must be feasible for $\H_{\epsilon}(\Gamma)$.
Therefore, $\H_{\epsilon}(\Gamma)$ is consistent.

This proves that, for any $\epsilon\in (0,\infty)$, $\H_{\epsilon}(\Gamma)$ is consistent
if and only if $\Gamma$ is $\epsilon$-dynamically-consistent.

(3) At this point, by composition with (1), Lemma~\ref{lem:dynamic_impl_epsilon} implies that
there exists a sufficiently small real number $\epsilon\in (0,\infty)$ such that
$\Gamma$ is dynamically-consistent if and only if $\H_{\epsilon}(\Gamma)$ is consistent.
\end{proof}

At this point, we are in the position to show the pseudocode for checking \eHyDCC,
	it is given in Algorithm~\ref{algo:solve_epsilonzeroDCC}:

\begin{algorithmenv}
\begin{algorithm}[H]\label{algo:solve_epsilonzeroDCC}
\caption{\texttt{check\_CHyTN-$\epsilon$-DC}$(\Gamma, \epsilon)$}

\KwIn{a CHyTN $\Gamma\triangleq\langle V, \A, L, \Ord, {\Ord}V, P \rangle$,
a rational number $\epsilon\triangleq N/D$, for $N,D\in\N_+$}

$\H_{\epsilon}(\Gamma)\leftarrow\texttt{construct\_}\H(\Gamma, \epsilon)$; \tcp{ref. Algorithm~\ref{algo:pseudocode_construct_H}}

\ForEach{($A=\langle t_A, H_A, w_A\rangle\in\A_{\H_{\epsilon}(\Gamma)}$ \texttt{\bf and} $h\in H_A$)}{
	$w_A(h)\leftarrow w_A(h)\cdot D$; \tcp{scale all weights of $\H_{\epsilon}(\Gamma)$, from $\Q$ to $\Z$}
}

$\phi\leftarrow \texttt{check\_\HTNC}(\H_{\epsilon}(\Gamma))$; \tcp{ref. Thm~\ref{Teo:MainAlgorithms}}

\If{($\phi$ is a feasible schedule of $\H_{\epsilon}(\Gamma)$)}{

\ForEach{(event node $v\in V_{\H_{\epsilon}(\Gamma)}$)}{
	$\phi(v)\leftarrow \phi(v)/ D$; \tcp{re-scale back to size the scheduling time, from $\Z$ to $\Q$, \wrt $\epsilon$}
}

\Return{$\langle\texttt{YES}, \phi\rangle$;}

}
\lElse{
\Return{$\texttt{NO}$}
}
\end{algorithm}\smallskip
\caption{Checking \eHyDCC on input $(\Gamma, \epsilon)$.}
\end{algorithmenv}

whereas, the pseudocode for checking \HyDCC is provided in Algorithm~\ref{algo:solve_DCC}, here below:

\begin{algorithmenv}
\begin{algorithm}[H]\label{algo:solve_DCC}
%\footnotesize
\caption{\texttt{check\_\DCC/\HyDCC}$(\Gamma)$}

\KwIn{a CHyTN $\Gamma\triangleq\langle V, \A, L, \Ord, {\Ord}V, P \rangle$}

$\hat{\epsilon} \leftarrow |\Sigma_P|^{-1}|V|^{-1}$; \tcp{ref. Thm.~\ref{thm:epsilonconsistency}}

\Return{\texttt{check\_CHyTN-$\epsilon$-DC$(\Gamma, \hat{\epsilon})$};}

\end{algorithm}
\caption{Checking \HyDCC on input $\Gamma$.}
\end{algorithmenv}

Notice that the latter (Algorithm~\ref{algo:solve_DCC}) invokes the former (Algorithm~\ref{algo:solve_epsilonzeroDCC}); more details follow.
\paragraph{Description of Algorithm~\ref{algo:solve_DCC}} Firstly, Algorithm~\ref{algo:solve_DCC}
computes a sufficiently small \emph{rational} number $\epsilon\in (0, \infty)\cap\Q$, by relying on Theorem~\ref{thm:epsilonconsistency},
\ie it is set $\hat{\epsilon}\triangleq |\Sigma_P|^{-1}|V|^{-1}$ (line~1).
Secondly, Algorithm~\ref{algo:solve_epsilonzeroDCC} is invoked on input $(\Gamma,\hat\epsilon)$.
At this point, Algorithm~\ref{algo:solve_epsilonzeroDCC} firstly constructs $\H_{\hat{\epsilon}}(\Gamma)$ (line~1 of
Algorithm~\ref{algo:solve_epsilonzeroDCC}) by invoking Algorithm~\ref{algo:pseudocode_construct_H},
and then it scales \emph{every} hyperarc's weight, appearing in $\H_{\hat{\epsilon}}(\Gamma)$, from $\Q$ to $\Z$ (at lines~2-3).
This is done by multiplying each weight by a factor $D$ (line~3), where $D\in \N_+$ is the denominator of $\hat\epsilon$ (\ie $D=|\Sigma_P|\cdot |V|$).
Thirdly, $\H_{\hat{\epsilon}}(\Gamma)$ is solved with the \HTNC-Checking algorithm underlying Theorem~\ref{Teo:MainAlgorithms} (at line~4),
\ie within the underlying algorithmic engine, an instance of the \HTNC problem is solved
	by reducing it to the problem of determining winning regions in a carefully constructed
		\MPG (see~\cite{CPR2014, CPR2015} for the details of such a reduction).
At this point, if the \HTNC algorithm outputs \texttt{YES}, together with a feasible schedule $\phi$ of $\H_{\hat{\epsilon}}(\Gamma)$,
then the time values of $\phi$ are scaled back to size \wrt $\hat{\epsilon}$,
and then $\langle \texttt{YES}, \phi \rangle$ is returned as output (lines~5-8);
otherwise, the output is simply \texttt{NO} (at line~9). Still, notice that, thanks to Item~3 of Theorem~\ref{Teo:MainAlgorithms},
we could also return a negative certificate, because negative instances are well
characterized in terms of generalized negative cycles (see Definition~\ref{def:negative_cycle}).

\begin{Rem}
The same algorithm, with essentially the same upper bound on its running time and space,
works also in case we allow for arbitrary boolean formulae as labels, rather than just conjunctions.
\end{Rem}
%\begin{Rem}\label{rem:rep_scenario}
%The algorithm lends itself to further optimizations.
%In fact, two scenarios $s_1$ and $s_2$ are \emph{equivalent} \wrt $\Gamma$ when they induce the same restriction network,
%\ie when $\Gamma^+_{s_1}=\Gamma^+_{s_2}$. This defines an equivalence relation on $\Sigma_P$.
%For any equivalence class, we may pick a representative scenario, and consider the set of representative scenarios only.
%Thus, inside Algorithm~\ref{algo:solve_DCC},
%$\Sigma_P$ can be regarded as the set of all the representative scenarios (one for each equivalence class).
%In general, the number of representative scenarios $|\Sigma_P|$ is bounded above by $2^{\min (|P|, l)}$,
%where $l$ is the number of distinct labels that appear in $\Gamma$,
%but in practice it may well be the case
%that $|\Sigma_P|=o(2^{\min(|P|, l)})$,
%or even $|\Sigma_P|=O(\text{poly}(\min(|P|, l)))$.
%\end{Rem}
\begin{Rem}
We remark that the \HTN/\MPG algorithm that is at the heart of
our approach requires integer weights (\ie it requires that $w(u,v)\in\Z$ for every $(u,v)\in A$);
somehow, we could not play it differently (see \cite{CPR2014,CPR2015} for a discussion).
Moreover, the algorithm always computes an integer solution to \HTN{s}/\MPG{s} and,
therefore, it always computes \emph{rational} feasible schedules for the CHyTN{s} given as input.
As such, it seems to us that this ``requirement`` actually turns out to be a plus in practice.
It is actually the integer assumption that allows us to analyze the algorithm quantitatively,
also presenting a sharp lower bounding analysis on the critical value of the reaction time $\hat{\epsilon}$,
where the CHyTN transits from being, to not being, dynamically-consistent.
We believe that these issues deserve much attention,
and going into them required a ``discrete" approach to the notion of numbers.
\end{Rem}

The correctness and the time complexity of
	Algorithms~\ref{algo:solve_epsilonzeroDCC}~and~\ref{algo:solve_DCC} is analyzed in Section~\ref{sect:epsilon}.

\section{Bounding Analysis on the Reaction Time $\hat{\epsilon}$}\label{sect:epsilon}
In this section we present an asymptotically sharp lower bound for $\hat{\epsilon}(\Gamma)$,
that is the critical value of reaction time where the CHyTN transits from being,
to not being, dynamically-consistent. The proof technique introduced in this analysis is applicable more generally,
when dealing with linear difference constraints which include strict inequalities.
This bound implies that Algorithm~\ref{algo:solve_DCC} is a (pseudo) singly-exponential time algorithm for solving \HyDCC.

To begin, we are going to provide a proof of Theorem~\ref{thm:boundexp}; for this, let us firstly introduce some further notation.

Let $\Gamma\triangleq\langle V, \A, L, \Ord, {\Ord}V, P \rangle$ be a dynamically-consistent CHyTN.
By Theorem~\ref{thm:mainreduction}, there exists $\epsilon>0$ such that the \HTN $\H_{\epsilon}(\Gamma)$ is consistent.
Then, let $\phi:V^{\text{Ex}}_{\Gamma}\rightarrow\RR$ be a feasible schedule for $\H_{\epsilon}(\Gamma)$.
For any hyperarc $A=\langle t_A, H_A, w_A\rangle\in\A_{\H_{\epsilon}}$, define a standard arc $a_A$ as follows:
\[a_A\triangleq \langle t_A, \hat{h}, w_A(\hat{h})\rangle,
\text{ where } \hat{h}\triangleq \arg\min_{h\in H_A}\big( \phi(h)-w_A(h) \big).\]
Then, notice that the network	$T^{\phi}_{\epsilon}(\Gamma)\triangleq \langle V^{\text{Ex}}_{\Gamma}, \bigcup_{A\in\A_{\H_\epsilon}} a_A\rangle$ is always an \STN.
Moreover, a moment's reflection reveals that, by definition of $\hat{h}$ as above, then $\phi$ is a feasible schedule for the \STN $T^{\phi}_{\epsilon}(\Gamma)$.

At this point, assuming $v\in V_\Gamma^{\text{Ex}}$, let us consider the \emph{fractional part} $r_v$ of $\phi_v$, \ie \[r_v\triangleq \phi_v-\lfloor\phi_v\rfloor.\]
Then, let $R\triangleq \{r_v\}_{v\in V^{\text{Ex}}_\Gamma}$ be the set of all the fractional parts induced by $V^{\text{Ex}}_\Gamma$.
Sort $R$ by the common ordering on $\RR$ and assume that $S\triangleq \{r_1, \ldots, r_k\}$ is the resulting ordered set (without repetitions), \ie $|S|=k$, $S=R$, $r_1<\ldots<r_k$.
Now, let $\texttt{pos}(v)\in [k]$ be the (unique) index position such that: \[ r_{\texttt{pos}(v)}=r_v. \]
Then, we define a new fractional part $r'_v$ as follows:
\begin{equation} r'_v\triangleq \frac{\texttt{pos}(v)-1}{|\Sigma_P|\cdot |V|}, \tag{NFP} \end{equation}
	and a new schedule function as follows:
\begin{equation} \phi'_v\triangleq \lfloor \phi_v\rfloor + r'_v. \tag{NSF} \end{equation}

Then the following holds.
\begin{Rem}\label{rem:invariant}
	Notice that (NFP) doesn't alter the ordering relation among the fractional parts,
	\ie \[ r'_u<r'_v\iff r_u<r_v, \text{ for any } u,v\in V^{\text{Ex}}_\Gamma.\]
	Moreover, since $\big( \texttt{pos}(v)-1 \big) < |\Sigma_P|\cdot |V|$,
		observe that (NSF) doesn't change the value of any integer part,
	\ie \[\lfloor \phi'_u\rfloor = \lfloor\phi_u\rfloor, \text{ for any } u\in V^{\text{Ex}}_\Gamma.\]
\end{Rem}

We are now in the position to prove Theorem~\ref{thm:boundexp}.
\begin{proof}[Proof of Theorem~\ref{thm:boundexp}]
Let $\Gamma\triangleq\langle V, \A, L, \Ord, {\Ord}V, P \rangle$ be dynamically-consistent. By Theorem~\ref{thm:mainreduction}
there exists $\epsilon'>0$ such that $\H_{\epsilon'}(\Gamma)$ is consistent
and it admits some \emph{feasible} schedule $\phi:V^{\text{Ex}}_\Gamma\rightarrow\RR$.
As mentioned, $\phi$ is feasible for the \STN $T^{\phi}_{\epsilon'}(\Gamma)$.
Now, let $\hat{\epsilon}\triangleq |\Sigma_P|^{-1}|V|^{-1}$.
Moreover, let $T^{\phi}_{\hat{\epsilon}}(\Gamma)$ be the \STN obtained from $T^{\phi}_{\epsilon'}(\Gamma)$ simply by
replacing, in the weights of the arcs, each weight $-\epsilon'$ with $-\hat{\epsilon}$.
We argue that $\phi'$ (as defined in (NSF) \wrt $\phi,V,\Sigma_P$), is a feasible schedule for the \STN $T^{\phi}_{\hat{\epsilon}}(\Gamma)$.
Indeed, every constraint of $T^{\phi}_{\hat{\epsilon}}(\Gamma)$ has form $\phi_v-\phi_u\leq w$, for some $w\in\Z$ or $w=-\hat{\epsilon}$.
\begin{itemize}
\item Consider the case $w\in\Z$. Notice that $\phi_v-\phi_u\leq w$ holds because $\phi$ is feasible for the \STN $T^{\phi}_{\epsilon'}(\Gamma)$.
Then, it is not difficult to see that $\phi'_v-\phi'_u\leq w$ holds as well, because of Remark~\ref{rem:invariant}.
%\begin{align*}
%\phi'_v-\phi'_u & =  \lfloor \phi_v\rfloor + r'_v - \lfloor \phi_u\rfloor - r'_u & \text{[by def. of (NSF)]} \\
%\end{align*}
%holds because of Remark~\ref{rem:invariant}.
\item Consider the case $w=-\hat{\epsilon}$.
Notice that $\phi_v-\phi_u\leq -\epsilon'$ holds because $\phi$ is feasible for the \STN $T^{\phi}_{\epsilon'}(\Gamma)$.

Then, notice that the following implication always holds,
	\[\phi_v-\phi_u\leq -\epsilon' \Longrightarrow \phi_v\neq \phi_u.\]
Hence, again by Remark~\ref{rem:invariant}, we can conclude that $\phi'_v\neq \phi'_u$.
At this point, we observe that the temporal distance between $\phi'_u$ and $\phi'_v$
is, therefore, at least $\hat{\epsilon}$ by definition of (NSF) and (NFP), \ie \[\phi'_u-\phi'_v\geq |\Sigma_P|^{-1}|V|^{-1}=\hat{\epsilon}.\]
That is to say, $\phi'_v - \phi'_u \leq -\hat{\epsilon}$.
\end{itemize}
This proves that $\phi'$ is a feasible schedule also for the \STN $T^{\phi}_{\hat{\epsilon}}(\Gamma)$.
Since $T^{\phi}_{\hat{\epsilon}}(\Gamma)$ is thus consistent, then,
a moment's reflection reveals that $\H_{\hat{\epsilon}}(\Gamma)$ is consistent as well thanks to the same schedule $\phi'$.

Therefore, by Theorem~\ref{thm:mainreduction}, the CHyTN $\Gamma$ is
$\hat{\epsilon}$-dynamically-consistent, provided that $\hat{\epsilon}\triangleq|\Sigma_P|^{-1}|V|^{-1}$.
\end{proof}

The correctness proof and the time complexity of Algorithm~\ref{algo:solve_DCC} is given next.
\begin{proof}[Proof of Theorem~\ref{thm:mainresult}]
To begin, notice that some of the temporal constraints introduced during
the reduction step depend on a sufficiently small parameter $\hat{\epsilon}\in (0, \infty)\cap\Q$,
whose magnitude turns out to depend on the size of the input CHyTN.
It is proved below that the time complexity of the algorithm depends multiplicatively on $D$,
where $\hat{\epsilon}=N/D$ for some $N,D\in\N_+$.
By Theorem~\ref{thm:boundexp}, $\hat{\epsilon}(\Gamma) \geq |\Sigma_P|^{-1}|V|^{-1}$;
	so line~1 of Algorithm~\ref{algo:solve_DCC} is correct.
Therefore, as a corollary of Theorem~\ref{thm:mainreduction},
	we obtain that Algorithm~\ref{algo:solve_DCC} correctly decides \DCC.

Concerning its time complexity,
	the most time-expensive step of the algorithm is clearly line~4 of Algorithm~\ref{algo:solve_epsilonzeroDCC},
which relies on Theorem~\ref{Teo:MainAlgorithms} in order to solve an instance of \HTNC on input $\H_{\epsilon}(\Gamma)$.
From Theorem~\ref{thm:mainreduction} we have an upper bound on the size of $\H_{\epsilon}(\Gamma)$,
while Theorem~\ref{Teo:MainAlgorithms} gives us a pseudo-polynomial upper bound for the computation time.
Also, recall that we scale weights by a factor $D$ at lines~2-3 of Algorithm~\ref{algo:solve_epsilonzeroDCC},
where $\hat{\epsilon}=N/D$ for some $N,D\in\N_+$. Thus, by composition, Algorithm~\ref{algo:solve_DCC}
decides \HyDCC in a time $T^{\text{Algo3}}_{\Gamma}$ which is bounded as follows,
where $W\triangleq \max_{a\in A} |w_a|$ and $D\in \N_+$:
\[
	T^{\text{Algo3}}_{\Gamma}= O\Big(\big(| V_{\H_{\epsilon}(\Gamma)} | + |\A_{\H_{\epsilon}(\Gamma)}|\big) m_{\A_{\H_{\epsilon}}(\Gamma)}\Big)WD.
\]
Whence, taking into account the upper bound on the size of $\H_{\epsilon}(\Gamma)$ give by Theorem~\ref{thm:mainreduction}, the following holds:
\begin{align*}
	T^{\text{Algo3}}_{\Gamma} &= O\Big( \big(|\Sigma_P||V| + |\Sigma_P||\A| + |\Sigma_P|^2|V|\big)(|\Sigma_P| m_{\A} + |\Sigma_P|^2|V|\, |P|) \Big) WD \\
										 &= O\Big(\cancel{|\Sigma_P|^{2}|V|m_{\A}} + \cancel{|\Sigma_P|^3|V|^2|P|}
										 					+ |\Sigma_P|^{2}|\A|m_{\A} + |\Sigma_P|^3|\A||V||P|
																	+ |\Sigma_P|^{3}|V|m_{\A} + |\Sigma_P|^4|V|^2|P| \Big)WD \\
										 &= O\Big(|\Sigma_P|^{2}|\A|m_{\A} + |\Sigma_P|^3|\A||V||P| + |\Sigma_P|^{3}|V|m_{\A} + |\Sigma_P|^4|V|^2|P|\Big)WD
\end{align*}

By Theorem~\ref{thm:boundexp}, it is sufficient to check $\epsilon$-dynamic consistency for $\hat{\epsilon}=|\Sigma_P|^{-1} |V|^{-1}$.

Therefore, the following worst-case time bound holds on Algorithm~\ref{algo:solve_DCC}:
\[
	T^{\text{Algo3}}_{\Gamma} = O\Big(|\Sigma_P|^{3}|V||\A|m_{\A} + |\Sigma_P|^4|\A||V|^2|P| + |\Sigma_P|^{4}|V|^2m_{\A} + |\Sigma_P|^5|V|^3|P|\Big)W.
\]

Since $|\Sigma_P|\leq 2^{|P|}$, the (pseudo) singly-exponential time bound follows.
\end{proof}

At this point,
a natural question is whether the lower bound given by Theorem~\ref{thm:boundexp}
can be improved up to $\hat{\epsilon}(\Gamma)=\Omega(|V|^{-1})$.
In turn, this would improve the time complexity of Algorithm~\ref{algo:solve_DCC} by a factor $|\Sigma_P|$.
However, the following theorem shows that this is not the case,
by exhibiting a \CSTN for which $\hat{\epsilon}(\Gamma) = 2^{-\Omega(|P|)}$.
This proves that the lower bound given by Theorem~\ref{thm:boundexp} is (almost) asymptotically sharp.

\begin{Thm}\label{thm:exponential_epsilon}
For each $n\in \N_+$ there exists a \CSTN $\Gamma^n$ such that:
	\[\hat{\epsilon}(\Gamma^n) < 2^{-n+1} = 2^{-|P^n|/3+1}, \]
where $P^n$ is the set of boolean variables of $\Gamma^n$.
\end{Thm}

\begin{proof}

For each $n\in\N_+$, we define a \CSTN $\Gamma^n\triangleq\langle V^n, A^n, L^n, \Ord^n, {\Ord}V^n, P^n \rangle$ as follows.

See Fig.~\ref{FIG:exponentialepsilon-counterexample-CSTN} for a clarifying illustration.

\begin{figure}[!htb]
\centering
\begin{tikzpicture}[arrows=->,scale=.95, node distance=1.5 and 1.5]
	\node[node, label={above:$Y_1?$}] (Y1) {$Y_1$};
 	\node[node, above left = of Y1, xshift=0ex, label={left:$\textbf{0}$}, label={above:\small $X_1?$}] (X1) {$X_1$};
	\node[node, above right = of Y1, xshift=0ex, label={above:\small $Z_1?$}] (Z1) {$Z_1$};
	\node[node, below = of Y1, yshift=-10ex, label={above:\small $Y_2?$}] (Y2) {$Y_2$};
	\node[node, above left = of Y2, xshift=0ex, label={above:\small $X_2?$}] (X2) {$X_2$};
	\node[node, above right = of Y2, xshift=0ex, label={above:\small $Z_2?$}] (Z2) {$Z_2$};
	\node[below left = of Y2, xshift=0ex] (fakeUL) {};
	\node[below right = of Y2, xshift=0ex] (fakeUR) {};
	\node[node, blackNode, scale=0.2, below = of Y2] (p1) {};
	\node[node, blackNode, scale=0.2, below = of p1, yshift=20ex] (p2) {};
	\node[node, blackNode, scale=0.2, below = of p2, yshift=20ex] (p3) {};
	\node[below left = of p3, yshift=8ex] (fakeDL) {};
	\node[below = of p3, yshift=8ex] (fakeDC) {};
	\node[below right = of p3, yshift=8ex] (fakeDR) {};
	\node[node, below = of fakeDC, yshift=-15ex, label={above:\small $Y_n?$}] (Yn) {$Y_n$};
	\node[node, above left = of Yn, xshift=0ex, label={above:\small $X_n?$}] (Xn) {$X_n$};
	\node[node, above right = of Yn, xshift=0ex, label={above:\small $Z_n?$}] (Zn) {$Z_n$};
	%arcs
	\draw[] (X1) to [] node[xshift=0ex, yshift=0ex,above] {\footnotesize $1, X_1 Y_1$} (Z1);
	\draw[] (X1) to [] node[xshift=-3ex, yshift=0ex,below] {\footnotesize $[2,2], \neg X_1$} (Y1);
	\draw[] (Y1) to [] node[xshift=3ex, yshift=0ex,below] {\footnotesize $[2,2], \neg Y_1$} (Z1);
	\draw[] (X1) to [bend right=45] node[xshift=-4ex, yshift=0ex,above] {\footnotesize $[5,5], Z_1$} (X2);
	\draw[] (Z1) to [bend left=45] node[xshift=7ex, yshift=0ex,above] {\footnotesize $[5,5], \neg Z_1 X_2 Y_2$} (Z2);
	\draw[] (Y1) to [] node[xshift=-3ex, yshift=0ex,above] {\footnotesize $[5,5], \neg Z_1$} (X2);
	\draw[] (Y1) to [] node[xshift=5ex, yshift=0ex,above]  {\footnotesize $[5,5], Z_1 X_2 Y_2$} (Z2);
	\draw[] (X2) to [] node[xshift=-3ex, yshift=0ex,below] {\footnotesize $[2,2], \neg X_2$} (Y2);
	\draw[] (Y2) to [] node[xshift=3ex, yshift=0ex,below] {\footnotesize $[2,2], \neg Y_2$} (Z2);
	\draw[] (Y2) to [] node[xshift=-3ex, yshift=0ex,above] {\footnotesize $[5,5], \neg Z_2$} (fakeUL);
	\draw[] (Y2) to [] node[xshift=5ex, yshift=0ex,above] {\footnotesize $[5,5], Z_2 X_3 Y_3$} (fakeUR);
	\draw[] (X2) to [bend right=45] node[xshift=-4ex, yshift=0ex,above] {\footnotesize $[5,5], Z_2$} (fakeUL);
	\draw[] (Z2) to [bend left=45] node[xshift=7ex, yshift=0ex,above] {\footnotesize $[5,5], \neg Z_2 X_3 Y_3$} (fakeUR);
	\draw[] (Xn) to [] node[xshift=-3ex, yshift=0ex,below] {\footnotesize $[2,2], \neg X_n$} (Yn);
	\draw[] (Yn) to [] node[xshift=3ex, yshift=0ex,below] {\footnotesize $[2,2], \neg Y_n$} (Zn);
	\draw[] (fakeDL) to [bend right=45] node[xshift=-5ex, yshift=1ex,above] {\footnotesize $[5,5], Z_{n-1}$} (Xn);
	\draw[] (fakeDR) to [bend left=45] node[xshift=8ex, yshift=1ex,above] {\footnotesize $[5,5], \neg Z_{n-1} X_n Y_n$} (Zn);
	\draw[] (fakeDC) to [] node[xshift=-2ex, yshift=-1ex,above] {\footnotesize $[5,5], \neg Z_{n-1}$} (Xn);
	\draw[] (fakeDC) to [] node[xshift=2ex, yshift=-1ex,above] {\footnotesize $[5,5], Z_{n-1} X_n Y_n$} (Zn);
\end{tikzpicture}
\caption{A CSTN $\Gamma^n$ such that $\hat{\epsilon}(\Gamma^n) = 2^{-\Omega(|P^n|)}$.}
\label{FIG:exponentialepsilon-counterexample-CSTN}
\end{figure}
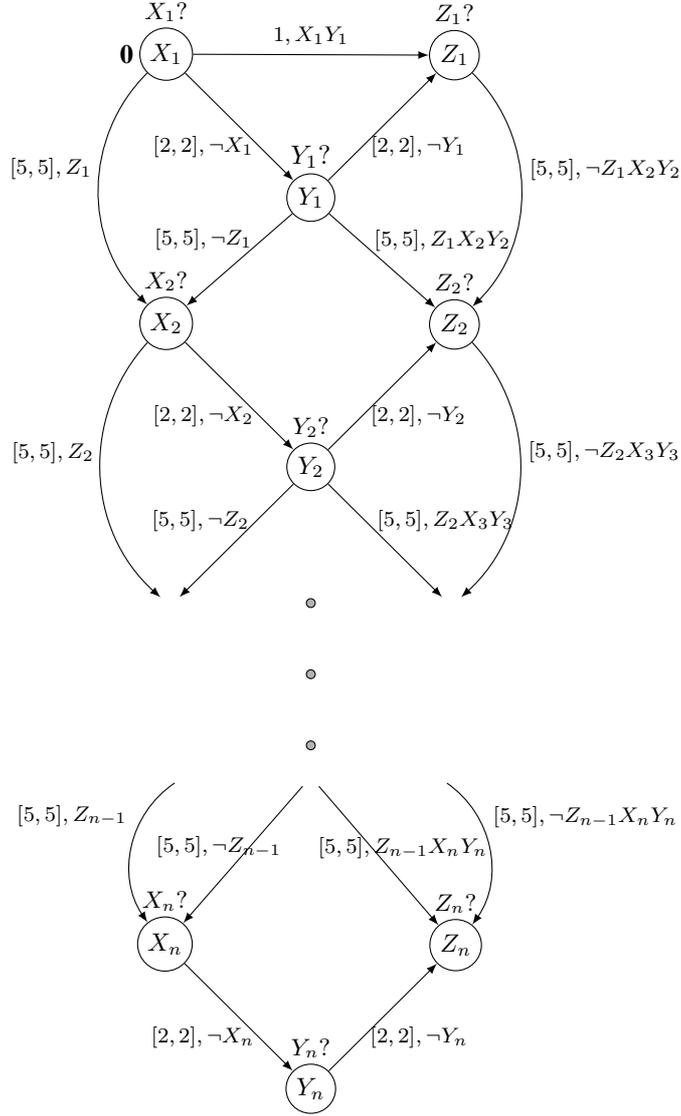

\begin{itemize}
\item $V^n\triangleq\{X_i, Y_i, Z_i \mid 1\leq i\leq n\}$;
\item $A^n\triangleq B \cup \bigcup_{i=1}^{n} C_i\cup \bigcup_{i=1}^{n-1} D_i$ \\ where:
\begin{itemize}
\item
$B\triangleq \{\langle X_1-v\leq 0, \lambda\rangle \mid v\in V^n\}\cup \{\langle Z_1 - X_1 \leq 1, X_1\wedge Y_1\rangle\}$;
\item
$C_i \triangleq \{\langle Y_i - X_i \leq 2, \neg X_i\rangle, \langle X_i - Y_i \leq -2, \neg X_i\rangle,
			\langle Z_i - Y_i \leq 2, \neg Y_i\rangle, \langle Y_i - Z_i \leq -2, \neg Y_i\rangle\}$;
\item
$D_i \triangleq \{\langle X_{i+1} - X_i \leq 5, Z_i\rangle, \langle X_i - X_{i+1}\leq -5, Z_i\rangle,
                         \langle X_{i+1} - Y_i\rangle\leq 5, \neg Z_i\rangle, \langle Y_i - X_{i+1}\leq -5, \neg Z_i\rangle,
		        \langle Z_{i+1}-Y_i\leq 5, Z_i\wedge X_{i+1}\wedge Y_{i+1}\rangle, \langle Y_i-Z_{i+1}\leq -5, Z_i\wedge X_{i+1}\wedge Y_{i+1}\rangle,
			 \langle Z_{i+1} - Z_i \leq 5, \neg Z_i \wedge X_{i+1}\wedge Y_{i+1}\rangle,
                         \langle Z_i - Z_{i+1}\leq -5, \neg Z_i \wedge X_{i+1}\wedge Y_{i+1}\rangle \}$;
\end{itemize}
\item $L^n(v)\triangleq\lambda$ for every $v\in V^n$;
${\Ord}V^n\triangleq V^n$;
$\Ord^n(v)\triangleq v$ for every $v\in {\Ord}V^n$;
$P^n\triangleq V^n$.
\end{itemize}

We exhibit an execution strategy $\sigma_n:\Sigma_{P^n}\rightarrow \Phi_{V^n}$,
	which we will show is dynamic and viable for $\Gamma^n$.

Let $\{\delta_i\}_{i=1}^{n}$ and $\{\Delta_i\}_{i=1}^{n}$ be two real valued sequences such that:
\[(1)\;\Delta_1\triangleq 1; (2)\; 0<\delta_i<\Delta_i;
(3)\; \Delta_{i}\triangleq \min(\delta_{i-1}, \Delta_{i-1}-\delta_{i-1}).\]
Then, the following also holds for every $1\leq i\leq n$:
\[ (4)\; 0 < \Delta_i \leq 2^{-i+1}, \]
where the equality holds if and only if $\delta_i = \Delta_i/2$.

Hereafter, provided that $s\in\Sigma_P$ and $\ell\in P^*$,
we will denote $\mathds{1}_{s(\ell)}\triangleq 1$	if $s(\ell)=\top$
	and $\mathds{1}_{s(\ell)}\triangleq 0$ if $s(\ell)=\bot$.

We are in the position to define $\sigma_n(s)$ for any $s\in\Sigma_P$:
\begin{itemize}
\item $[\sigma_n(s)]_{X_1}\triangleq 0$;
\item $[\sigma_n(s)]_{Y_1}\triangleq \delta_1 \mathds{1}_{s(X_1)} + 2\mathds{1}_{s(\neg X_1)}$;
\item $[\sigma_n(s)]_{Z_1}\triangleq$
	$\mathds{1}_{s(X_1\wedge Y_1)} + (2+[\sigma_n(s)]_{Y_1})\mathds{1}_{s(\neg X_1 \vee \neg Y_1)}$;
\item $[\sigma_n(s)]_{X_{i}}\triangleq 5 + [\sigma_n(s)]_{X_{i-1}}\mathds{1}_{s(Z_{i-1})} + \\
	+ [\sigma_n(s)]_{Y_{i-1}}\mathds{1}_{s(\neg Z_{i-1})}$, for any $2\leq i\leq n$;
\item $[\sigma_n(s)]_{Y_i}\triangleq$ $[\sigma_n(s)]_{X_i} + \delta_i\mathds{1}_{s(X_i)} + 2\mathds{1}_{s(\neg X_i)}$, for any $2\leq i \leq n$;
\item $[\sigma_n(s)]_{Z_{i}}\triangleq \big(5+[\sigma_n(s)]_{Y_{i-1}}\mathds{1}_{s(Z_{i-1})} + \\
                + [\sigma_n(s)]_{Z_{i-1}}\mathds{1}_{s(\neg Z_{i-1})}\big)\mathds{1}_{s(X_{i}\wedge Y_{i})} + \\
                + (2+[\sigma_n(s)]_{Y_{i}})\mathds{1}_{s(\neg X_i\vee \neg Y_i)}$, for any $2\leq i\leq n$;
\end{itemize}
Let us prove, by induction on $n\geq 1$, that $\sigma_n$ is viable and dynamic for $\Gamma^n$.

\begin{itemize}
	\item \textit{Base case}.
			Let $n=1$. Notice that $\Gamma^1$ almost coincides with the \CSTN $\Gamma_{\frac{1}{2}}$ described in Example~\ref{ex:CSTN-fractional};
				so, it is really needed that $0<\delta_1<1$.
				Then, by construction, $\sigma_1$ leads to the schedule depicted in Figure~\ref{fig:basecase}.
			This shows that $\sigma_1$ is viable and dynamic for $\Gamma^1$.
		%	(Notice that $\Gamma_1$ is equivalent to $\Gamma_{\not\N}$ from Section~\ref{sect:integerity})
\begin{figure}[!htb]
	\centering
	\begin{tikzpicture}[scale=0.8, level distance=50pt, sibling distance=1pt]
	\Tree [. \framebox{$[\sigma_1(s)]_{X_1}=0$}
		 \edge node[left,xshift=-1ex]{$s(X_1)=\top$};
		 [. \framebox{$[\sigma_1(s)]_{Y_1}=\delta_1$}
			\edge node[left, xshift=-1ex]{$s(Y_1)=\top$};
			[. \framebox{$[\sigma_1(s)]_{Z_1}=1$}
			]
			\edge node[right, xshift=1ex]{$s(Y_1)=\bot$};
			[. \framebox{$[\sigma_1(s)]_{Z_1}=\delta_1+2$}
			]
		 ]
		 \edge node[right,xshift=1ex]{$s(X_1)=\bot$};
		 [. \framebox{$[\sigma_1(s)]_{Y_1}=2$}
			\edge node[xshift=1.3ex]{$s(Y_1)=\top$ or $s(Y_1)=\bot$};
			[. \framebox{$[\sigma_1(s)]_{Z_1}=4$}
			]
		 ]
	]
	\end{tikzpicture}
	\caption{A viable and dynamic execution strategy for the base case $n=1$.}
	\label{fig:basecase}
\end{figure}
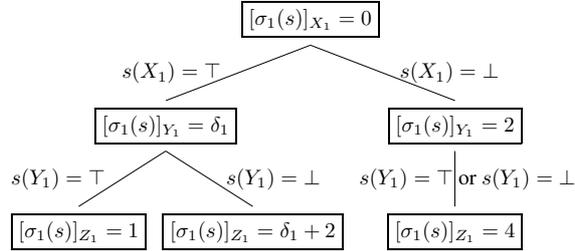
\item \textit{Inductive step}.
	Let us assume that $\sigma_{n-1}$ is viable and dynamic for $\Gamma^{n-1}$.
	Then, by construction, $[\sigma_n(s)]_v=[\sigma_{n-1}(s)]_v$ for every $s\in\Sigma_P$
	and $v\in V_{n-1}$. Hence, by induction hypothesis, $\sigma_n$ is viable and dynamic on $V_{n-1}$.
	Moreover, by construction, $\sigma_n$ leads to the schedule depicted in Figure~\ref{fig:inductivesteptop}
	and Figure~\ref{fig:inductivestepbot}. This shows that $\sigma_n$ is viable and dynamic even on $V_n\setminus V_{n-1}$.
	Thus, $\sigma_n$ is viable and dynamic for $\Gamma^n$, \ie $\Gamma^n$ is dynamically-consistent.
\begin{figure}[!htb]
		\centering
		\begin{tikzpicture}[scale=0.75, level distance=50pt,sibling distance=1pt]
		\Tree [.\framebox{$[\sigma_{n}(s)]_{Z_{n-1}}=[\sigma_{n-1}(s)]_{Z_{n-1}}$}
			\edge node[left, xshift=-1ex]{$s(Z_{n-1})=\top$};
		[. \framebox{$[\sigma_n(s)]_{X_n}= [\sigma_{n-1}(s)]_{X_{n-1}}+5$}
		 \edge node[left,xshift=-1ex]{$s(X_1)=\top$};
		 [. \framebox{$[\sigma_n(s)]_{Y_n}=[\sigma_n(s)]_{X_n} + \delta_n$}
			\edge node[left, xshift=-1ex]{$s(Y_n)=\top$};
			[. \framebox{$[\sigma_1(s)]_{Z_n}=[\sigma_{n-1}(s)]_{Y_{n-1}}+5$}
			]
			\edge node[right, xshift=1ex]{$s(Y_n)=\bot$};
			[. \framebox{$[\sigma_1(s)]_{Z_n}=[\sigma_n(s)]_{Y_n}+2$}
			]
		 ]
		 \edge node[right,xshift=1ex]{$s(X_n)=\bot$};
		 [. \framebox{$[\sigma_n(s)]_{Y_n}=[\sigma_n(s)]_{X_n}+2$}
			\edge node[xshift=1.3ex]{$s(Y_n)=\top$ or $s(Y_n)=\bot$};
			[. \framebox{$[\sigma_n(s)]_{Z_n}=[\sigma_n(s)]_{Y_n}+2$}
			]
		 ]
		]
		]
		\end{tikzpicture}
		\caption{A viable and dynamic execution strategy for the inductive step $n-1\leadsto n$ when $s(Z_{n-1})=\top$.}
		\label{fig:inductivesteptop}
	\end{figure}
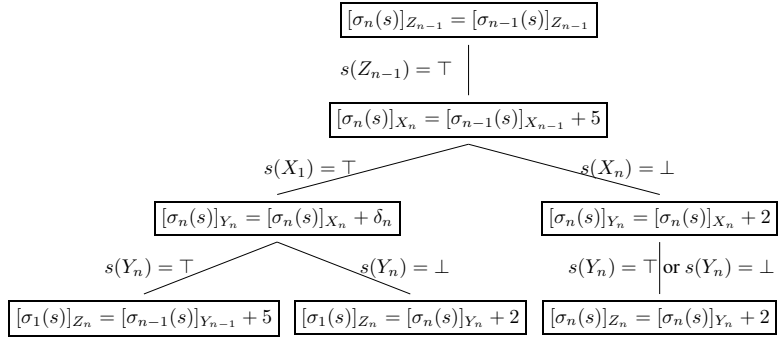
	\begin{figure}[!htb]
		\hspace*{-.45in}
		\subfloat[An execution strategy for the inductive step $n-1\leadsto n$ when $s(Z_{n-1})=\bot$]{
		\begin{tikzpicture}[scale=0.725,level distance=50pt,sibling distance=1pt]
		\Tree [. \framebox{$[\sigma_{n-1}(s)]_{Z_{n-1}}=[\sigma_{n}(s)]_{Z_{n-1}}$}
			\edge node[right, xshift=1ex]{$s(Z_{n-1})=\bot$};
			[. \framebox{$[\sigma_n(s)]_{X_n}= [\sigma_{n-1}(s)]_{Y_{n-1}}+5$}
			 \edge node[left,xshift=-1ex]{$s(X_n)=\top$};
			 [. \framebox{$[\sigma_n(s)]_{Y_n}=[\sigma_n(s)]_{X_n} + \delta_n$}
				\edge node[left, xshift=-1ex]{$s(Y_n)=\top$};
				[. \framebox{$[\sigma_n(s)]_{Z_n}=[\sigma_{n-1}(s)]_{Z_{n-1}}+5$}
				]
				\edge node[right, xshift=1ex]{$s(Y_n)=\bot$};
				[. \framebox{$[\sigma_n(s)]_{Z_n}=[\sigma_n(s)]_{Y_n}+2$}
				]
			 ]
			 \edge node[right,xshift=1ex]{$s(X_n)=\bot$};
		 	[. \framebox{$[\sigma_n(s)]_{Y_n}=[\sigma_n(s)]_{X_n}+2$}
				\edge node[xshift=1.3ex]{$s(Y_n)=\top$ or $s(Y_n)=\bot$};
				[. \framebox{$[\sigma_n(s)]_{Z_n}=[\sigma_n(s)]_{Y_n}+2$} ]
		 	]
			]
		]
		\end{tikzpicture}
		}
		\subfloat[An excerpt of $\Gamma^{n}$ relevant to the inductive step $n-1\leadsto n$.]{
			\quad
		\begin{tikzpicture}[arrows=->,scale=0.7, node distance=1.5 and 1.5]
			\node[node, label={above,yshift=.5ex:$Y_{n-1}?$}] (Y1) {$Y_{n-1}$};
			\node[node, above left = of Y1, xshift=0ex, label={above:\small $X_{n-1}?$}] (X1) {$X_{n-1}$};
			\node[node, above right = of Y1, xshift=0ex, label={above:\small $Z_{n-1}?$}] (Z1) {$Z_{n-1}$};
			\node[node, below = of Y1, yshift=-10ex, label={above,yshift=.5ex:\small $Y_n?$}] (Y2) {$Y_n$};
			\node[node, above left = of Y2, xshift=0ex, label={above:\small $X_n?$}] (X2) {$X_n$};
			\node[node, above right = of Y2, xshift=0ex, label={above:\small $Z_n?$}] (Z2) {$Z_n$};
				%arcs
			\draw[] (X1) to [] node[xshift=0ex, yshift=0ex,above] {\footnotesize $1, X_{n-1} Y_{n-1}$} (Z1);
			\draw[] (X1) to [] node[xshift=-3.6ex, yshift=0ex,below] {\footnotesize $[2,2], \neg X_{n-1}$} (Y1);
			\draw[] (Y1) to [] node[xshift=4.3ex, yshift=0ex,below] {\footnotesize $[2,2], \neg Y_{n-1}$} (Z1);
			\draw[] (X1) to [bend right=50] node[xshift=-4.75ex, yshift=-1ex,above] {\footnotesize $[5,5], Z_{n-1}$} (X2);
			\draw[] (Z1) to [bend left=50] node[xshift=-.75ex, yshift=-1ex,above] {\footnotesize $[5,5], \neg Z_{n-1} X_n Y_n$} (Z2);
			\draw[] (Y1) to [] node[xshift=-3.75ex, yshift=0ex,above] {\footnotesize $[5,5], \neg Z_{n-1}$} (X2);
			\draw[] (Y1) to [] node[xshift=6.25ex, yshift=0ex,above]  {\footnotesize $[5,5], Z_{n-1} X_n Y_n$} (Z2);
			\draw[] (X2) to [] node[xshift=-3ex, yshift=0ex,below] {\footnotesize $[2,2], \neg X_n$} (Y2);
			\draw[] (Y2) to [] node[xshift=3.5ex, yshift=0ex,below] {\footnotesize $[2,2], \neg Y_n$} (Z2);
		\end{tikzpicture}
			}
		\caption{The inductive step $n-1\leadsto n$ when $s(Z_{n-1})=\bot$.}
		\label{fig:inductivestepbot}
	\end{figure}
\end{itemize}

We claim that $\hat{\epsilon}(\Gamma^n) < 2^{-n+1}=2^{-|P^n|/3+1}$ for every $n\geq 1$.
Consider the following scenario $\hat{s}$ for $1\leq i\leq n$:
\[
\hat{s}(X_i)\triangleq \hat{s}(Y_i)\triangleq \top; \;\;\;
\hat{s}(Z_i)\triangleq \left\{
\begin{array}{ll}
\top, & \text{ if } \delta_i \leq \Delta_i/2 \\
\bot, & \text{ if } \delta_i > \Delta_i/2 \\
\end{array} \right.
\]
We shall assume that $\sigma$ is an execution strategy for $\Gamma^n$ and study necessary conditions to ensure that $\sigma$ is viable and dynamic,
provided that the observations follow the scenario $\hat{s}$. First, $\sigma$ must schedule $X_1$ at time $[\sigma(\hat{s})]_{X_1}=0$.
Then, since $\hat{s}(X_1)=\top$, we must have $0<[\sigma(\hat{s})]_{Y_1}<1$, because of the constraint $(Z_1-X_1\leq 1, X_1\wedge Y_1)$.
Stated otherwise, it is necessary that:
\[ 0 < [\sigma(\hat{s})]_{Y_1} - [\sigma(\hat{s})]_{X_1} < \Delta_1. \]
After that, since $\hat{s}(Y_1)=\top$, then $\sigma$ must schedule $Z_1$ at time $[\sigma(\hat{s})]_{Z_1}=1=\Delta_1$.
A moment's reflection reveals that almost identical necessary conditions now recur for $X_2, Y_2, Z_2$,
with the crucial variation that it will be necessary to require: $0 < [\sigma(\hat{s})]_{Y_2} < \Delta_2$.
Indeed, by proceeding inductively, it will be necessary that for every $1\leq i\leq n$ and every $n\in\N_+$:
\[0 < [\sigma(\hat{s})]_{Y_i} - [\sigma(\hat{s})]_{X_i} < \Delta_i.\]
As already observed in ($4$), $0<\Delta_n\leq 2^{-n+1}$.
Thus, any viable and dynamic execution strategy $\sigma$ for $\Gamma^n$ must satisfy:
\[\displaystyle0 < [\sigma(\hat{s})]_{Y_n} - [\sigma(\hat{s})]_{X_n} < \frac{1}{2^{n-1}}=\frac{1}{2^{|P^n|/3-1}}.\]
Therefore, once the Planner has observed the outcome $\hat{s}(X_n)=\top$ from the observation event $X_n$,
then he must react by scheduling $Y_n$ within time $2^{-n+1}=2^{-|P^n|/3+1}$ in the future \wrt $[\sigma(\hat{s})]_{X_n}$.

Therefore, $\hat{\epsilon}(\Gamma^n) < 2^{-n+1}=2^{-|P^n|/3+1}$ holds for every $n\geq 1$.
\end{proof}

\section{Related Works}\label{sect:relatedworks}
This section discusses of some alternative approaches offered by the current literature.
Recall that the article of \cite{TVP2003} has already been discussed in the introduction.

The work of \cite{Ci14} provided the first sound-and-complete
algorithm for checking the dynamic controllability of \CSTN{s} with Uncertainty (CSTNU),
and thus it can be employed for checking the dynamic consistency of \CSTN{s} as a special case.
The algorithm reduces the DC-Checking of CSTNUs to the problem of solving Timed Game Automata (TGA).
Nevertheless, no worst-case upper bound on the time complexity of the procedure was provided in~\cite{Ci14}.
Still, one may observe that solving TGAs is a problem of much higher complexity than solving \MPG{s}.
Compare the following known facts: solving 1-player
TGAs is $\PSPACE$-complete and solving 2-player TGAs is $\EXP$-complete;
on the contrary, the problem of determining \MPG{s} lies in
$\NP\cap\coNP$ and it is currently an open problem to prove whether it is in $\P$.
Indeed, the algorithm offered in~\cite{Ci14} has not been proven to be singly-exponentially time bounded,
to the best of our knowledge it is still open whether singly-exponential time TGA-based algorithms for \DCC do exist.

Next, a sound algorithm for checking the dynamic controllability of CSTNUs was given by Combi, Hunsberger, Posenato~in~\cite{CHP13}.
However, it was not shown to be complete. To the best of our knowledge,
it is currently open whether or not it can be extended in order to prove completeness \wrt the CSTNU model.

Regarding the particular CSTN model,~\cite{HPC15} presented,
at the same conference in which the preliminary version of this work appeared,
a sound-and-complete DC-checking algorithm for CSTNs.
It is based on the propagation of temporal constraints labeled by propositions.
However, to the best of our knowledge, the worst-case complexity of the algorithm is currently unsettled.
Also notice that the algorithm in \cite{HPC15} works on CSTN{s} only, regardless of the CHyTN model.
Indeed, we believe that our approach (based on tractable games plus reaction-time $\hat{\epsilon}$) and the approach
of \cite{HPC15} (based on the propagation of labeled temporal constraints) can benefit from each other;
for instance, recently~\cite{HP16} presented an alternative, equivalent semantics for $\epsilon$-dynamic consistency,
as well as a sound-and-complete $\epsilon$-DC-checking algorithm based on the propagation of labeled constraints.

Finally, in~\cite{CCR16}, it is introduced and studied $\pi$-DC,
a sound notion of dynamic consistency with an instantaneous reaction time,
\ie one in which the Planner is allowed to react to any observation at
the same instant of time in which the observation is made.
It turns out that $\pi$-DC is not equivalent to $\epsilon$-DC with $\epsilon=0$,
and that the latter is actually inadequate for modeling an instantaneous reaction-time.
Still, a simple reduction from $\pi$-DC-Checking to DC-Checking is identified;
combined with Theorem~\ref{thm:mainresult}, this provides a $\pi$-DC-Checking procedure whose time complexity
remains (pseudo) singly-exponential in the number of propositional variables.

\section{Conclusion}\label{sect:conclusions}
In this work we introduced the \emph{Conditional Hyper Temporal Network (CHyTN)} model,
a natural extension and generalization of both the \CSTN and the \HTN model which is obtained by blending them together.
We proved that deciding whether a given \CSTN or CHyTN is dynamically-consistent is \coNP-hard,
and that deciding whether a given CHyTN is dynamically-consistent is \PSPACE-hard,
provided that the input instances are allowed to include both multi-head and multi-tail hyperarcs.
In light of this, we focused on CHyTNs that allow only multi-head hyperarcs,
and offered the first deterministic (pseudo) singly-exponential time algorithm
for the problem of checking the dynamic consistency of multi-head CHyTNs,
also producing a dynamic execution strategy whenever the input CHyTN is dynamically-consistent.
As a byproduct, this provides the first sound-and-complete
(pseudo) singly-exponential time algorithm for checking the dynamic consistency of \CSTN{s}.
The algorithm is based on a novel connection between CHyTN{s} and Mean Payoff Games.
The presentation of such connection was mediated by the \HTN model.
The algorithm actually manages a few more general variants of the problem; 
\eg those where labels are not required to be conjunctions.
To summarize, at the heart of the algorithm a reduction to \MPG{s} is
mediated by the \HTN model. The CHyTN is dynamically-consistent if and only if the corresponding \MPG is
everywhere won, and a dynamic execution strategy can be conveniently read out by an everywhere winning positional strategy.
The size of this \MPG is at most polynomial in the number of the possible scenarios;
as such, the term at the corresponding time complexity exponent is linear,
at worst, in the number of the observation events. The same holds for the running time of the resulting algorithm.
In order to analyze the algorithm, we introduced a refined notion of dynamic consistency, named $\epsilon$-dynamic consistency,
also presenting a sharp lower bounding analysis on the critical
value of the reaction time $\hat{\varepsilon}$ where a CHyTN transits from being,
to not being, dynamically-consistent.

In future works we would like to settle the exact computational complexity of \DCC/\HyDCC,
as well as to extend our approach in order to check the dynamic controllability of \CSTN with Uncertainty~\cite{HPC12}.
An extensive experimental evaluation taking good account of optimizations and heuristics is also planned.

\paragraph*{Acknowledgments}
This work was supported by the \emph{Department of Computer Science, University of Verona, Verona, Italy} 
under \emph{Ph.D.} grant ``Computational Mathematics and Biology`` on a co-tutelle 
agreement with \emph{Laboratoire d'Informatique Gaspard-Monge (LIGM), Universit{\'e} Paris-Est, Marne-la-Vall{\'e}e, Paris, France}.

%\section*{References}
\bibliographystyle{plain}
\bibliography{biblio}
\end{document}